\documentclass[12pt]{article}

\usepackage[margin=1in]{geometry}
\usepackage{amsmath,amssymb,amsthm}
\usepackage{mathtools}
\usepackage{pdflscape}
\usepackage{rotating}
\usepackage{natbib}
\usepackage{tikz}
\usepackage{pgfplots}
\usepgfplotslibrary{groupplots}
\pgfplotsset{compat=1.18}
\usetikzlibrary{calc}
\definecolor{cA}{HTML}{2C5F8A}
\definecolor{cB}{HTML}{C25E00}
\definecolor{cC}{HTML}{8C9BA5}
\usepackage{hyperref}
\hypersetup{hidelinks}
\usepackage{booktabs}
\usepackage{graphicx}
\usepackage{setspace}
\usepackage[patch=none]{microtype}
\IfFileExists{bbm.sty}{\usepackage{bbm}}{\newcommand{\mathbbm}[1]{\mathbf{#1}}}
\usepackage{enumitem}
\newtheorem{theorem}{Theorem}
\newtheorem{lemma}{Lemma}
\newtheorem{proposition}{Proposition}
\newtheorem{corollary}{Corollary}
\newtheorem{assumption}{Assumption}
\newtheorem{remark}{Remark}

\newcommand{\E}{\mathbb{E}}
\newcommand{\Prob}{\mathbb{P}}
\newcommand{\R}{\mathbb{R}}
\newcommand{\ind}{\mathbbm{1}}
\newcommand{\plim}{\overset{p}{\longrightarrow}}
\newcommand{\dlim}{\overset{d}{\longrightarrow}}
\DeclareMathOperator*{\argmax}{arg\,max}
\newcommand{\N}{\mathcal{N}}

\begin{document}

\title{Semiparametric Dynamic Logit Model with Endogenous Networks}
\author{Brice Romuald Gueyap Kounga\thanks{Department of Economics, Western University. Email: \href{mailto:bgueyapk@uwo.ca}{\texttt{bgueyapk@uwo.ca}}. Special thanks to Nail Kashaev, Salvador Navarro, Roy Allen, and Tim
Conley for their comments and suggestions in the preparation of this paper. The views expressed herein and any remaining errors are mine.}}
\date{ \today}
\maketitle


\begin{abstract}
This paper studies the identification and estimation of a semiparametric dynamic logit model with nonparametric link formation, in which a time-varying index of unobserved characteristics affects outcomes and links. The outcome depends on covariates, the past outcome, a fixed effect, and an unknown function of that index. Conditioning on a switch by the same agent across waves removes the fixed effect but not the change in the unknown function. When this change is common within network types and followed by a stable period, comparing two agents of the same network type, meaning they have the same probability of linking with each of the other agents, removes it and identifies covariate effects and state dependence. The pairwise conditional likelihood estimator is consistent, including when the index is continuous, and asymptotically normal when the index takes finitely many values and matches are exact. In simulations, bias falls with sample size, and coverage is near nominal among successful fits. The fixed-effects dynamic logit and pooled logits with peer-average controls remain biased when the index moves. A second estimator, conditioning on recovered network classes, is more precise. A feasibility exercise on the Glasgow Teenage Friends and Lifestyle Study examines how past use and pocket money affect adolescents' recorded tobacco or cannabis use and documents the data requirements.
\end{abstract}

\medskip

\noindent \textit{Keywords:} Dynamic logit, semiparametric estimation, social networks, network formation, conditional maximum likelihood, state dependence, unobserved heterogeneity.

\medskip

\noindent \textit{JEL Classification:} C14, C23, C25, D85.

\section{Introduction}

Teenagers who smoke this year are likely to smoke next year, and their friends often smoke too \citep{fletcher2010, nakajima2007}. Farmers deciding whether to adopt a new crop technology rely on their own past experience, and farmers who exchange information with each other often make the same decision \citep{conley2010}. Two mechanisms can produce such patterns. The first is state dependence: a past choice directly raises the probability of the current one, as when smoking becomes a habit. The second is unobserved characteristics (e.g., a taste for risk, impulsivity, patience, or ability) that affect both the choice and the way people form ties, so that choices persist and cluster among friends or neighbors without any habit. To study such choices, the econometrician should therefore use a model that allows for state dependence and for an endogenous network. The network is endogenous when people form their ties partly because of the same unobserved characteristics that drive their choices.

This paper studies such a model: a dynamic logit with an endogenous network, in which the unobserved characteristics that drive both the choice and the ties may change over time. Its main result is that matching students on their network type, the probabilities with which they befriend each other student, identifies the effects of past choices and of observed covariates when the change in the unobserved term between two waves is common to students of the same network type and is followed by a stable wave. The levels of the unobserved term and the fixed effects may differ across matched students: conditioning on each student's own switch removes the levels, and the comparison across matched students removes the common change. I propose two estimators and examine the method in simulations and in data on adolescent smoking.

The economic question is which mechanism drives persistence, and the distinction matters for policy. If persistence comes from state dependence, an intervention that stops a student from smoking for one year also lowers her probability of smoking in later years. If persistence comes from stable characteristics, the student returns to the same probability of smoking once the intervention ends \citep{heckman1981}. Separating the two mechanisms is hard, because the characteristics that drive smoking, friendships, and persistence are not observed. A risk-tolerant student may be more likely to smoke, to keep smoking, and to befriend other risk-takers, and the researcher has no data on risk tolerance. Friendships and smoking are then related for reasons that have nothing to do with habit or with observed covariates such as pocket money.

Existing methods do not resolve this problem, for different reasons. A pooled logit omits the heterogeneity, which is correlated with past smoking and with the covariates, so its estimates are inconsistent in general. Adding friends' characteristics as controls does not repair this, because the controls capture only the part of the heterogeneity reflected in friends' observables. Methods that condition on fixed effects, such as the dynamic logit of \citet{honore2000panel}, or that otherwise do not use the network \citep{ouyang2024}, remove the heterogeneity only if it does not change over time. During adolescence, attitudes to risk and to social norms do change. Methods for endogenous networks specify a model of link formation \citep{goldsmith2013, hsieh2016, arduini2015} or rely on control functions or exclusion restrictions \citep{johnsson2021, kashaev2025}. Each rests on restrictions on the formation process that differ from the ones used here. This paper restricts the arguments of the link rule but leaves its functional form unknown rather than assuming a parametric form for it.

\label{loc:contribution}This paper provides identification and estimation results for a dynamic logit model with a fixed effect and a time-varying unobserved characteristic that also drives link formation, without a parametric model of the links. It combines the path conditioning of the dynamic logit with the network-type matching of \citet{auerbach2022} and \citet{gueyap2026}. Conditioning on the order in which a student switches between two waves, with the covariates held stable over the following wave, removes any characteristic that does not change over the window. A student's network type is the list of her probabilities of befriending each other student, given her past smoking. Comparing two students of the same network type and with the same past smoking then removes the change in the unobserved term, provided the change is common to students of the same type. \citet{auerbach2022} and \citet{gueyap2026} use a single cross-section and therefore require equal levels of the unobserved term within a network type. This paper requires equal changes, which allows a fixed effect that matching is not assumed to equalize. The price is a set of further restrictions: stability of the unobserved term over the wave after the transition, strict exogeneity of the covariates and the unobserved characteristic, a dyadic link rule on a scalar index, and enough support among the matched comparisons. In the smoking example, the common change and the stability that follows it describe a cohort observed around a scheduled school transition: students of the same network type, who befriend the same kinds of peers, meet the same new environment, so their unobserved terms move by the same amount over the transition, whatever their levels and their fixed effects, and stay put over the following wave. Section~\ref{sec:identification} and Appendix~\ref{sec:appendix_window} make these conditions precise. Under these restrictions the covariate and state dependence coefficients are identified.

The network is useful because a student's friendships carry indirect information about the characteristic that drives them. The observed network does not reveal a student's own value of the characteristic, but it does reveal whether two students have the same network type. The codegree of two students is the number of friends they have in common. The codegree distance between two students aggregates, over all third students, the differences between their codegrees with each third student. As the network grows, it vanishes exactly when the two students share a network type, under the support and separation conditions of Proposition~\ref{prop:bridge}.

The outcome is whether student $i$ smokes at wave $t$. The model is a dynamic logit in which the outcome depends on observed covariates $X_{it}$, such as pocket money; on the student's own smoking at the previous wave, $y_{it-1}$; on a fixed effect $\mu_i$; and on an unknown function $\lambda$ of an index of unobserved characteristics $w_{it}$ that changes over time. The model applies to any agents, and the example refers to students. The parameters of interest are $\beta$, the coefficients on the covariates, and $\alpha$, the coefficient on past smoking, which measures state dependence on the log-odds scale. Habit is one interpretation of a positive $\alpha$, not an identified mechanism. Two unobserved terms enter: the fixed effect $\mu_i$ is the part of the unobserved characteristics that does not change over time and does not enter the friendship rule, such as a physiological susceptibility to nicotine, and the function $\lambda$ is an unknown function of the student's index of unobserved social characteristics $w_{it}\in[0,1]$, such as risk attitude, attachment to norms, and impulsivity, which change during adolescence. Following \citet{auerbach2022}, I call $\lambda$ the social influence function. The name refers to the student's own pull toward smoking. The function is not an effect of friends' smoking, since the outcome equation contains neither friends' outcomes nor the realized network, and the model identifies no peer effect. The same $w_{it}$ governs whom the student befriends, which is why the network is informative about $w_{it}$ and why the network is endogenous. With $\lambda \equiv 0$, the model is the fixed-effects dynamic logit of \citet{honore2000panel}. Reading $\beta$ and $\alpha$ as causal effects requires the assumptions of Section~\ref{sec:identification}. The function $\lambda$, the fixed effects, and the link rule are nuisance components, which the method removes without estimating, and the level of $\lambda$ is not recovered.

The identification argument uses at least four dates: the initial outcome and three waves. Take two students with the same initial outcome and the same network type at wave 1, who switch in opposite directions between waves 1 and 2 and whose covariates and social influence term do not change over wave 3. Conditioning on the switch cancels each student's fixed effect and level of $\lambda$ but not the change in $\lambda$ between waves 1 and 2. Under the identifying assumption, the change is the same for both students and cancels when they are compared. The probability that one rather than the other started smoking then depends only on $(\beta, \alpha)$ and on observed variables. Since the codegree distance identifies which students share a network type (Proposition~\ref{prop:bridge}), $\beta$ and $\alpha$ are identified under overlap and rank conditions (Theorem~\ref{thm:identification}), together with the common change in $\lambda$ within each network class (Corollary~\ref{cor:lambda}).

The identifying assumption, Assumption~\ref{ass:persistence}, has two parts. The change in the social influence function between waves 1 and 2 is the same for all students of the same network type at wave 1 with the same initial outcome, and the function does not change between waves 2 and 3. The assumption allows a shift common to every student and excludes changes that differ across students of the same network type, such as a characteristic with its own idiosyncratic innovations. It cannot be tested from the data alone, because the social influence function is never observed. The maintained assumptions (Assumption~\ref{ass:sampling}) are logistic errors, strict exogeneity of the covariates and of the unobserved characteristic, and a link rule on the two students' unobserved characteristics and past outcomes, with independent link shocks and no observed covariates.

The estimator is a conditional maximum likelihood estimator over pairs of switching students. Each pair is weighted by its estimated codegree distance and by the stability of each student's covariates over the following wave. The weights are kernels, and their bandwidths set how similar two students must be to count as a match. The estimator is consistent, including with continuous network types and covariates, with bandwidths that shrink at the rates of Assumption~\ref{ass:kernel} (Theorem~\ref{thm:consistency}). Suppose that the characteristic takes finitely many values and that covariate changes have an atom at zero. Exact matches then occur with positive probability, and the estimated weights select them with probability approaching one. The estimator is root-$n$ consistent and asymptotically normal, with a variance clustered by student (Theorem~\ref{thm:normality}). With three waves, the same holds for a bandwidth computed from the first-wave network (Corollary~\ref{cor:bandwidth}). With continuous types, the paper establishes consistency only, for the reasons given in Section~\ref{sec:continuous}. When the characteristic takes finitely many values, switchers whose estimated distances fall below the bandwidth, with the same initial outcome, form classes whose members share the change in $\lambda$. Conditioning on the number of switches in each direction within a class gives a second, class-conditional estimator. This estimator is root-$n$ consistent and asymptotically normal with one transition period (Proposition~\ref{prop:class}) and more precise in the simulations. In practice, a gap in the estimated distances and classes that do not change with the bandwidth are what the class estimator needs, although neither certifies that the classes are correct. The pairwise estimator is the general one, the only one the theory covers when types are continuous, and the safer choice when the classes are uncertain.

Simulations in designs that satisfy the assumptions agree with the theory. The bias of the pairwise estimator falls with the sample size and its confidence intervals have close to nominal coverage among successful fits. The class-conditional estimator is more precise than the pairwise estimator. Pooled logits, with or without peer-average controls, and the within-agent estimator of \citet{honore2000panel} remain biased when the characteristic moves, with coverage that falls as the sample grows.

\label{loc:introapp}I apply the estimators to the Glasgow Teenage Friends and Lifestyle Study, a panel of 96 students observed over three waves, to study how past smoking and pocket money affect smoking, with the outcome defined as any tobacco or cannabis use. The fits treat friendship links as symmetric, set the unobserved initial outcome to zero, drop the covariate-stability condition, which no switching student satisfies, and compare all pairs of switchers rather than network-matched pairs, because at every localized bandwidth the pairwise and class-conditional likelihoods are separated. The conditional fits are stabilized by the Firth penalization. I find that the point estimate of state dependence is positive in each of the five unlocalized fits, falling from 2.22 in a pooled logit and 2.10 with peer-average controls to 1.83 when each student is compared with herself over time, 1.41 when students are paired with each other, and 0.89 when all switchers are compared within a single class. The pocket money coefficient, per pound of monthly pocket money, is imprecisely estimated in every specification, and alcohol use is the covariate most strongly associated with smoking. These fits are descriptive: they implement the conditioning but not the network matching, their ordering is sensitive to the penalization, and the baseline unpenalized maxima exist while several robustness specifications are separated. A dataset suited to the method needs an observed initial outcome, enough stable switchers within recoverable network classes, and a sufficiently informative measurement of friendships.

Section~\ref{sec:lit} reviews the literature and Section~\ref{sec:model} presents the model. Sections~\ref{sec:identification} and~\ref{sec:estimation} give the identification results and the two estimators, and Section~\ref{sec:asymptotics} establishes consistency and the finite-type distribution theory. Section~\ref{sec:montecarlo} reports the simulations, Section~\ref{sec:empirical} examines feasibility in the Glasgow data, and Section~\ref{sec:conclusion} concludes. Appendix~\ref{sec:appendix} contains the proofs, Appendix~\ref{sec:appendix_lambda} explains what is identified about the social influence function, Appendix~\ref{sec:appendix_ordered} notes that the binary estimator applies to an ordered outcome dichotomized at one known cutoff, and Appendices~\ref{sec:appendix_window}--\ref{sec:appendix_figures} describe the observation window, report the descriptive fits for the Glasgow data, and collect supplementary simulation tables.

\section{Literature Review}\label{sec:lit}
This paper lies at the intersection of three strands of the econometrics literature: dynamic logit models, semiparametric panel data methods, and the econometrics of endogenous social networks. Each literature provides tools for part of the problem considered here: conditional likelihood removes time-invariant heterogeneity in dynamic logit models, semiparametric panel methods handle unknown nuisance functions for continuous outcomes, and the network-based methods reviewed below use linking behavior to control for latent characteristics in a single cross-section.

An important starting point for panel logit methods is \citet{chamberlain1980}, who showed that the logistic specification admits sufficient statistics for individual fixed effects, allowing consistent estimation without specifying the distribution of unobserved heterogeneity. \citet{honore2000panel} extended this approach to dynamic settings by exploiting conditional likelihood arguments based on pairs of observations with identical sufficient statistics, requiring covariates to remain locally constant across adjacent periods. More recent contributions have expanded identification and estimation methods for dynamic logit models. \citet{honore2025moment} derive moment conditions that eliminate nuisance parameters without relying on conditional likelihood, while \citet{honore2025ordered} and \citet{muris2025} extend these ideas to ordered outcomes. \citet{shin2025} derives sufficient statistics for a first-order Markov feedback process in a dynamic panel logit and shows that conditional likelihood identification then requires restrictions on that process and on the initial condition. The present paper maintains strict exogeneity of the covariates and of the latent characteristic. \citet{ouyang2024} develop a two-step maximum score approach for dynamic binary choice that does not rely on the logistic distribution. It matches observations on an estimated index and relies on additional restrictions on serial dependence and on the tails of the error distribution. \citet{dano2025} study logit models with general fixed-effect structures for panel and network data, including individual effects, time trends, and two-way and dyadic effects. Outside the logistic family, \citet{honore2006bounds} show that the dynamic parameter is generally only partially identified, and \citet{chamberlain2010} characterizes the special informational role of the logistic error. These results delimit what any approach in this class, including the present one, can deliver without the logistic assumption. Much of this literature treats unobserved heterogeneity as a time-invariant individual effect. The richer structures of \citet{dano2025} allow time variation, but through fixed effects whose structure is specified in terms of the observed indices of units and periods, such as an individual effect combined with a time trend or with time effects. The heterogeneity considered here keeps the individual fixed effect of that literature and adds an unknown function of a latent characteristic whose movement is not restricted to a parametric trend: at the transitions used for identification, its increment is shared within network classes, and the latent term is stable over the following wave. The added term can be removed only because agents with the same network type share that increment, so that the network supplies the information needed to compare them. The present paper addresses this case by combining conditional likelihood arguments, which remove the fixed effect, with network-type matching, which removes the movement, and local temporal smoothing, which selects agents whose covariates are stable after the transition. Grouped fixed-effects methods \citep{bonhomme2015, bonhomme2022} also accommodate heterogeneity that varies over time, by classifying units into groups and estimating group-specific effects. The class-conditional estimator of Section~\ref{sec:class_estimator} also works with groups, but its groups are network classes recovered from the codegree distance, and it removes their effects by conditioning on each group's number of switches, as conditional logistic regression does for matched sets \citep{breslow1980}, rather than estimating them.

A second related literature studies semiparametric estimation in panel data models with unknown nuisance functions. The partially linear framework of \citet{robinson1988} established that finite-dimensional parameters can be estimated consistently without specifying the functional form of a nonparametric component. Extensions to panel data include the kernel-based estimators of \citet{henderson2008} and the marginal integration approach of \citet{qian2012}, both of which accommodate unknown functions while controlling for individual heterogeneity. These methods, however, are developed for continuous outcomes with linear or partially linear conditional mean models. They do not address nonlinear dynamic logit models in which the nuisance function depends on an unobserved characteristic that is itself recovered only indirectly through network information.

The identification of peer effects and endogenous social interactions has been central since \citet{manski1993}, who showed that endogenous interactions, contextual effects, and correlated unobservables are generically not separately identifiable without additional restrictions. A large empirical literature, surveyed by \citet{bramoulle2020} and \citet{jackson2017}, documents the importance of social interactions in education \citep{sacerdote2001,calvo2009}, consumption \citep{degiorgi2020}, labor markets \citep{beugnot2013}, and health behaviors \citep{christakis2008}. Existing approaches to endogenous network formation generally rely on one of three strategies: a parametric or semiparametric model for link formation \citep{goldsmith2013,hsieh2016,arduini2015}, identification from equilibrium restrictions in strategic network formation models \citep{badev2021}, or control-function or exclusion-restriction arguments under alternative identifying assumptions \citep{johnsson2021,kashaev2025}. These approaches differ substantially, but all require assumptions on the network formation process that may be difficult to justify empirically. The network-type approach instead combines a nonparametric link function with restrictions connecting network profiles to outcome heterogeneity and uses the similarity of observed linking behavior to recover equivalence classes of latent characteristics.

A related line of work exploits network data to control for unobserved heterogeneity through the structure of the graph itself: \citet{graham2017} models degree heterogeneity in network formation with agent-specific effects, and \citet{zeleneev2020} develops matching on unobserved network types under nonparametric heterogeneity, in a spirit close to the codegree approach used here but for outcomes defined on the network rather than for a dynamic panel of individual choices. \citet{gaoniu2026} study identification in dynamic dyadic network formation models with dyad fixed effects, where the object of interest is link formation itself. \citet{starck2025} and \citet{xu2025} use network data to control for unobservables in a cross-section, the former by forming comparison groups of similar agents to estimate treatment effects, the latter through a network control function for nonparametric partial effects that generalizes \citet{auerbach2022}.

The papers most closely related to the present work are \citet{auerbach2022} and \citet{gueyap2026}. \citet{auerbach2022} introduced the idea of matching agents according to their network formation behavior in partially linear regression models, showing that network similarity can control for latent heterogeneity under an equivalence restriction, without parametrically specifying the link function. His framework is static and applies to continuous outcomes with partially linear conditional means. \citet{gueyap2026} extended this idea to semiparametric logit models in a static cross-sectional setting, establishing point identification through codegree-based matching and proposing a kernel-weighted maximum likelihood estimator. Both papers require agents with the same network type to have the same unobserved term in the outcome equation. The present paper extends this line of research to dynamic panel logit models and relaxes that requirement to equal changes of the unobserved term, which allows an individual fixed effect that is not a separate argument of the link function and that network matching is not assumed to equalize. The dynamic setting introduces identification challenges that do not arise in static environments. State dependence interacts with time-varying latent heterogeneity, requiring conditional likelihood arguments together with local temporal smoothing, while the evolution of the network requires matching to hold across time rather than at a single cross-section. The resulting framework combines dynamic panel logit methods with network-type matching to identify and estimate a semiparametric logit model under endogenous network formation without imposing a parametric specification on the link formation process.

\section{Model}\label{sec:model}

Consider a panel of $n$ agents observed over $T+1$ periods $t = 0, 1, \ldots, T$, where $T \geq 3$. Each agent $i \in \{1, \ldots, n\}$ is characterized by an observed vector of time-varying covariates $X_{it} \in \R^k$, a binary outcome $y_{it} \in \{0,1\}$, an index of unobserved social characteristics $w_{it} \in [0,1]$, and an unobserved agent-specific term $\mu_i \in \R$ that is constant over time. The econometrician observes the outcomes $(y_{i0}, y_{i1}, \ldots, y_{iT})$, including the initial outcome $y_{i0}$, the covariates $(X_{i1}, \ldots, X_{iT})$, and the network $D_t$ defined below at each period $t \geq 1$. Neither the index $w_{it}$ nor the term $\mu_i$ is ever observed.

The binary outcome is generated by the dynamic semiparametric logit model:
\begin{equation}\label{eq:outcome}
y_{it} = \ind\left\{X_{it}'\beta + \alpha y_{it-1} + \lambda(w_{it}) + \mu_i - \varepsilon_{it} \geq 0\right\}, \quad t = 1, \ldots, T,
\end{equation}
where $\beta \in \R^k$ is a vector of slope parameters, $\alpha \in \R$ captures state dependence, $\lambda: [0,1] \to \R$ is an unknown measurable function, which I call the social influence function, $\mu_i$ is an agent effect, and $\varepsilon_{it}$ is an idiosyncratic error following a standard logistic distribution $F(x) = e^x/(1+e^x)$. The explanatory variable $X_{it}$ does not contain an intercept, as the terms $\lambda(w_{it})$ and $\mu_i$ absorb it. The index is partially linear because $X_{it}$ enters linearly while $\lambda$ is unrestricted in functional form.

\label{loc:mu}The unobserved heterogeneity has two components, distinguished by whether they are arguments of the link function. The index $w_{it}$ collects unobserved social characteristics, such as motivation, trust, risk tolerance, or social capital, that influence both the agent's economic behavior and the formation of social ties. Unlike the heterogeneity in standard fixed-effects models, $w_{it}$ may vary across both agents and time periods, accommodating settings where the relevant social characteristics evolve. The agent effect $\mu_i$ collects the persistent unobserved characteristics that affect the outcome but are not arguments of the link function. When a persistent characteristic affects the links only in part, the part that the links do not reveal belongs to $\mu_i$. Agents of the same network type need not have the same $\mu_i$, and the analysis does not require it. In the smoking application, $\mu_i$ may stand for a physiological susceptibility to nicotine or for the smoking environment at home. It is treated as a fixed effect: it is constant over time, it may depend on $(X_i, w_i)$ and on the initial outcome in an arbitrary way, and nothing is assumed about its distribution. The two components are eliminated by different arguments. The agent effect cancels within each agent over time, as in the conditional likelihood approach of \citet{chamberlain1980} and \citet{honore2000panel}. The change of $\lambda(w_{it})$ over time, which a within-agent comparison cannot remove, is eliminated by comparing agents whose network types coincide.

In addition to outcomes and covariates, the researcher observes a binary adjacency matrix $D_t$ of dimension $n \times n$ at each period $t$. The element $D_{ijt} = 1$ if there is a direct link between agents $i$ and $j$ at time $t$, and $D_{ijt} = 0$ otherwise. Self-links are excluded ($D_{iit} = 0$) and all links are undirected ($D_{ijt} = D_{jit}$). The network evolves over time according to a nonparametric link formation model:
\begin{equation}\label{eq:network}
D_{ijt} = \ind\left\{f(w_{it}, w_{jt}; y_{it-1}, y_{jt-1}) \geq \eta_{ijt}\right\} \ind(i \neq j), \quad i,j = 1, \ldots, n, \quad t = 1, \ldots, T,
\end{equation}
where $f:[0,1]^2 \times \{0,1\}^2 \to [0,1]$ is an unknown jointly symmetric measurable function and $\eta_{ijt}$ is an idiosyncratic time-varying disturbance. The link formation function $f$ depends on the agents' current indices of unobserved social characteristics $(w_{it}, w_{jt})$ and their lagged outcomes $(y_{it-1}, y_{jt-1})$, allowing network structure to be influenced by prior behavior. This creates a dynamic feedback between outcomes and network formation.

Joint symmetry means $f(w, w';\, y, y') = f(w', w;\, y', y)$: exchanging the two agents, together with their lagged outcomes, leaves the linking probability unchanged, as undirected links require. For each sample size $n$ the link shocks form an $n \times n$ symmetric array $\{\eta_{ijt}\}_{i,j=1}^n$, and the dependence on $n$ is left implicit.

The link formation function $f$ in~\eqref{eq:network} excludes the observed covariates $X_{it}$, a restriction of the model. Allowing links to depend directly on observed characteristics would require conditional network profiles and additional support and estimation arguments, an extension not studied here.

This specification has four features. First, the model nests two existing models. With $\lambda \equiv 0$ it is the dynamic logit model with individual fixed effects of \citet{honore2000panel}. With $\mu_i \equiv 0$, $\alpha = 0$, $w_{it} = w_i$, and a link function that does not depend on the lagged outcomes, it is the static semiparametric logit model with an endogenous network of \citet{gueyap2026}. The contribution of the paper is the combination: a time-varying unobserved term that the network reveals, on top of a fixed effect that it does not. Second, the inclusion of lagged outcomes in the link formation function captures empirically relevant patterns such as homophily on past behavior, for example the tendency of program adopters to form ties with other adopters. Third, the time-varying nature of $w_{it}$ accommodates evolving social characteristics, though the identification strategy also covers the simpler case of time-invariant $w_i$. Finally, the estimands are the slope and state dependence parameters $(\beta, \alpha)$, and the functions $\lambda$ and $f$ and the effects $\mu_i$ are nuisance components.

\label{loc:scope}Two clarifications concern the scope of the model. First, the name ``social influence function'' refers to the agent's own unobserved heterogeneity, the pull toward the outcome exerted by the characteristics that also shape the agent's ties. It does not refer to an effect of other agents' outcomes, and the paper identifies no peer effect. The outcome equation contains neither friends' outcomes nor the realized adjacency matrix, so the network serves only as information about a confounder. In particular, identifying state dependence while allowing for latent heterogeneity that is related to the network does not identify the effect of changing a particular friend's smoking, which would require friends' outcomes in the outcome equation and restrictions that this paper does not impose. The level of $\lambda$ is not separately identified from $\mu_i$, which can absorb any time-invariant component of it. What can be recovered is its change between waves (Corollary~\ref{cor:lambda}). Second, although the functional form of $f$ is left unspecified, the network model is not unrestricted: links are conditionally independent across dyads given the agents' indices and lagged outcomes, the index is scalar, and links are symmetric.

Figure~\ref{fig:network} illustrates the joint evolution of outcomes and network structure over three periods using the adolescent smoking example.

\begin{figure}[!htbp]
\centering
\begin{tikzpicture}[
    smoker/.style={circle, draw=black, fill=black, text=white, minimum size=22pt, font=\footnotesize\bfseries},
    nonsmoker/.style={circle, draw=gray!60, fill=gray!20, text=black, minimum size=22pt, font=\footnotesize\bfseries},
    link/.style={gray!50, line width=1pt},
    trans/.style={->, dashed, gray!60, line width=0.8pt},
    scale=0.95
]

\node[font=\bfseries] at (2, 4.5) {$t = 1$};
\node[font=\scriptsize, gray] at (2, 4.1) {Network $D_1$};

\draw[link] (1, 3) -- (3, 3);
\draw[link] (1, 3) -- (0.5, 1.8);
\draw[link] (3, 3) -- (3.5, 1.8);
\draw[link] (0.5, 1.8) -- (3.5, 1.8);
\draw[link] (0.5, 1.8) -- (2, 0.8);
\draw[link] (3.5, 1.8) -- (2, 0.8);

\node[nonsmoker] (A1) at (1, 3) {1};
\node[smoker]    (A2) at (3, 3) {2};
\node[nonsmoker] (A3) at (0.5, 1.8) {3};
\node[smoker]    (A4) at (3.5, 1.8) {4};
\node[nonsmoker] (A5) at (2, 0.8) {5};

\node[font=\bfseries] at (7, 4.5) {$t = 2$};
\node[font=\scriptsize, gray] at (7, 4.1) {Network $D_2$};

\draw[link] (6, 3) -- (8, 3);
\draw[link] (6, 3) -- (5.5, 1.8);
\draw[link] (8, 3) -- (8.5, 1.8);
\draw[link] (5.5, 1.8) -- (8.5, 1.8);
\draw[link] (8, 3) -- (7, 0.8);

\node[smoker]    (B1) at (6, 3) {1};
\node[smoker]    (B2) at (8, 3) {2};
\node[nonsmoker] (B3) at (5.5, 1.8) {3};
\node[smoker]    (B4) at (8.5, 1.8) {4};
\node[nonsmoker] (B5) at (7, 0.8) {5};

\node[font=\bfseries] at (12, 4.5) {$t = 3$};
\node[font=\scriptsize, gray] at (12, 4.1) {Network $D_3$};

\draw[link] (11, 3) -- (13, 3);
\draw[link] (11, 3) -- (10.5, 1.8);
\draw[link] (13, 3) -- (13.5, 1.8);
\draw[link] (11, 3) -- (13.5, 1.8);
\draw[link] (10.5, 1.8) -- (12, 0.8);

\node[smoker]    (C1) at (11, 3) {1};
\node[smoker]    (C2) at (13, 3) {2};
\node[smoker]    (C3) at (10.5, 1.8) {3};
\node[nonsmoker] (C4) at (13.5, 1.8) {4};
\node[nonsmoker] (C5) at (12, 0.8) {5};

\draw[trans] (4.2, 2.2) -- (5.0, 2.2);
\draw[trans] (9.2, 2.2) -- (10.0, 2.2);

\node[smoker, minimum size=14pt] at (2, -0.3) {};
\node[font=\scriptsize, right] at (2.4, -0.3) {$y_{it} = 1$ (Smoker)};
\node[nonsmoker, minimum size=14pt] at (5.3, -0.3) {};
\node[font=\scriptsize, right] at (5.7, -0.3) {$y_{it} = 0$ (Non-smoker)};
\draw[link] (9, -0.3) -- (9.7, -0.3);
\node[font=\scriptsize, right] at (9.8, -0.3) {Friendship link};

\end{tikzpicture}
\caption{Evolution of outcomes and network structure over three periods. Black nodes represent smokers ($y_{it} = 1$) and gray nodes represent non-smokers ($y_{it} = 0$). Agent~1 transitions from non-smoking to smoking between $t = 1$ and $t = 2$ and persists at $t = 3$, illustrating state dependence ($\alpha$). Agent~3 begins smoking at $t = 3$ as evolving unobserved characteristics increase risk-taking. Agent~4 quits smoking at $t = 3$ following a shift in unobserved characteristics and simultaneously forms a new link with agent~1. The friendship network changes at each period: the econometrician observes both outcomes and links but cannot directly observe the latent characteristics $w_{it}$ that drive both processes.}
\label{fig:network}
\end{figure}

The model accommodates a range of empirical settings. In the adolescent smoking application of \citet{fletcher2010}, who uses the longitudinal Add Health data, $y_{it}$ indicates whether student $i$ smokes in survey wave $t$, the friendship network $D_t$ is observed at each wave and evolves as students form and dissolve friendships, and $y_{it-1}$ captures state dependence through habit formation and addiction. The unobserved characteristic $w_{it}$ can be interpreted as a latent index of sensation-seeking or risk attitude that evolves during adolescence, affecting both the propensity to smoke and the tendency to befriend other risk-taking students. Since $w_{it}$ enters both the outcome equation and the link formation process, standard logit estimates of the slope and state dependence parameters are biased.

A similar structure arises in the agricultural technology adoption setting of \citet{conley2010}, who study pineapple farmers in Ghana making repeated planting decisions across growing seasons. Here $y_{it}$ indicates whether farmer $i$ adopts the new technology in season $t$, the information network $D_t$ records which farmers exchange agricultural advice and evolves as farmers update their contacts, and $y_{it-1}$ captures learning-by-doing and switching costs. The unobserved characteristic $w_{it}$ represents latent farming ability and openness to innovation, which may evolve over the seasons for reasons other than the farmer's own past adoption. More able farmers are both more likely to adopt and more sought after as advisors, generating the same confounding pathway between network structure and outcomes.

The framework also applies to the study of juvenile delinquency by \citet{patacchini2012}, who use longitudinal friendship network data from Add Health. In this setting $y_{it}$ indicates whether adolescent $i$ engages in delinquent behavior in wave $t$, the friendship network $D_t$ evolves across survey waves, and $y_{it-1}$ captures persistence due to reputation effects, reduced stigma, or accumulated criminal capital. The unobserved characteristic $w_{it}$ represents an evolving index of impulsivity or attachment to conventional norms, which influences both the propensity toward delinquency and the formation of friendships with other delinquent youth through conformist motives.

\section{Identification}\label{sec:identification}

This section develops the main identification results. I first define the concept of network-type equivalence for the dynamic setting, then establish point identification of the parameters $(\beta, \alpha)$.

The functions $\lambda$ and $f$ cannot be identified separately. The latent index $w_{it}$ is not observed, and the network data alone cannot pin it down. To see this, consider the homophily specification $f(w_{it}, w_{jt}; y_{it-1}, y_{jt-1}) = 1 - (w_{it} - w_{jt})^2$. If $\{w_{it}\}_{i,t}$ generates a sequence of adjacency matrices $\{D_t\}_{t}$, then the transformed profile $\{\tilde{w}_{it}\}_{i,t} = \{1 - w_{it}\}_{i,t}$ generates exactly the same matrices, since $(w_{it} - w_{jt})^2 = ((1 - w_{it}) - (1 - w_{jt}))^2=(\tilde{w}_{it} -\tilde{w}_{jt})^2$. The econometrician cannot distinguish between $w_{it}$ and $\tilde{w}_{it}$ from the observed links. If the outcome function is also relabeled as $\tilde{\lambda}(u)=\lambda(1-u)$, then $\tilde{\lambda}(\tilde{w}_{it})=\lambda(w_{it})$ and outcome probabilities are unchanged. The numerical labels of the latent index and the functions indexed by those labels are therefore not separately identified. The agent effect adds a second layer: since $\mu_i$ may depend on $w_i$, the composite value $\lambda(w_{it})$ is not identified either. What is identified is its common movement between the transition waves, on the network classes that the codegree distance recovers (Corollary~\ref{cor:lambda}). The identification strategy targets the structural parameters $(\beta, \alpha)$.

\subsection{Network-Type Equivalence}

The parameters of interest are $\beta$ and $\alpha$. The functions $\lambda$ and $f$ are treated as nuisance components. Following the approach developed in \citet{gueyap2026} for the static case, identification relies on comparing agents with equivalent linking profiles. In the static case the comparison requires their unobserved contributions to outcomes to be equal, whereas in the dynamic case with an agent effect it requires only that the changes of those contributions between the transition waves be equal.

In the dynamic setting, each agent $i$'s type at period $t$ consists of two components: the index of unobserved social characteristics $w_{it}$ and the observed lagged outcome $y_{it-1}$. The link formation function $f(w_{it}, w_{jt};\, y_{it-1}, y_{jt-1})$ in~\eqref{eq:network} depends on both components for each agent. Since the lagged outcomes are observed and discrete, they can be conditioned on directly from the data, reducing the identification problem to controlling for the unobserved component $w_{it}$ through network matching.

For each fixed pair of lagged outcomes $(y, y') \in \{0,1\}^2$, define the \emph{conditional individual link function} of agent $i$ at period $t$ as:
\begin{equation*}
f_{w_{it};\, y, y'}(s) = f(w_{it}, s;\, y, y'), \quad s \in [0,1].
\end{equation*}
This is the conditional probability that an agent whose index is $w_{it}$ and whose lagged outcome is $y$ links with agents of every other social characteristic in $[0,1]$ whose lagged outcome is $y'$.
For each fixed $(y, y')$, this is a function of a single argument $s \in [0,1]$, exactly as in the static individual link function $f_{w_i}(s) = f(w_i, s)$ of \citet{auerbach2022} and \citet{gueyap2026}. The key difference is that the dynamic setting introduces four such functions, one for each configuration $(y, y') \in \{0,1\}^2$, whereas the static setting has only one.

To quantify the similarity between two agents' linking behavior, define the \emph{network-type distance} at period $t$ for agents sharing the same lagged outcome $y_{it-1} = y_{jt-1} = y$ as:
\begin{align}\label{eq:netdist}
\rho_{ijt} = \|f_{w_{it};\, y, \cdot}-f_{w_{jt};\, y, \cdot}\|_2&=\left[\sum_{y' \in \{0,1\}} \int_0^1 \big(f_{w_{it};\, y, y'}(s) - f_{w_{jt};\, y, y'}(s)\big)^2 \, ds\right]^{1/2} \notag\\
&= \left[\sum_{y' \in \{0,1\}} \int_0^1 \big(f(w_{it}, s;\, y, y') - f(w_{jt}, s;\, y, y')\big)^2 \, ds\right]^{1/2}.
\end{align}
This aggregates the $L^2$ distance between the conditional individual link functions over both possible partner lagged outcomes. When $\rho_{ijt} = 0$, agents $i$ and $j$ produce the same conditional linking probabilities for almost every partner characteristic $s \in [0,1]$ and each lagged outcome $y' \in \{0,1\}$, and they are then called \emph{conditionally network-type equivalent} at period $t$.

Conditional network-type equivalence does not imply that the two agents share the same value of the index of unobserved social characteristics. To see this, consider the link formation function $f(w_{it}, w_{jt}; y_{it-1}, y_{jt-1}) = 1 - (w_{it} - \tfrac{1}{2})^2 - (w_{jt} - \tfrac{1}{2})^2$. The individual link function is $f_{w;\, y, y'}(s)  = 1 - (w - \tfrac{1}{2})^2 - (s - \tfrac{1}{2})^2$, which depends on $w$ only through $(w - \tfrac{1}{2})^2$. Agents with $w_{it} = 0.3$ and $w_{jt} = 0.7$ produce identical linking probabilities with every other agent, so $\rho_{ijt} = 0$ even though $w_{it} \neq w_{jt}$. In the adolescent smoking application, a low-effort student and a high-effort student may form friendships with the same types of peers yet differ in their underlying characteristic. The identification strategy does not require distinguishing between them, only that their social influence functions change by the same amount between the transition waves. Figure~\ref{fig:network_equiv} illustrates such a pair.

\begin{figure}[!htbp]
\centering
{\small (a) Conditionally network-type equivalent agents}\\[4pt]
\begin{tikzpicture}[
    smoker/.style={circle, draw=black, fill=black, text=white, minimum size=22pt, font=\footnotesize\bfseries},
    nonsmoker/.style={circle, draw=gray!60, fill=gray!20, text=black, minimum size=22pt, font=\footnotesize\bfseries},
    link/.style={gray!50, line width=1pt},
    eqlink/.style={blue, line width=1.2pt},
    trans/.style={->, dashed, gray!60, line width=0.8pt},
    scale=0.95
]

\node[font=\bfseries] at (2, 4.5) {$t = 1$};
\node[font=\scriptsize, gray] at (2, 4.1) {Network $D_1$};

\draw[eqlink] (1, 3) -- (3, 3);        
\draw[eqlink] (3, 3) -- (3.5, 1.8);    
\draw[link] (0.5, 1.8) -- (2, 0.8);  
\draw[eqlink] (3.5, 1.8) -- (2, 0.8);  
\draw[eqlink] (1, 3) -- (2, 0.8);      

\node[nonsmoker] (A1) at (1, 3) {1};
\node[smoker]    (A2) at (3, 3) {2};
\node[nonsmoker] (A3) at (0.5, 1.8) {3};
\node[smoker]    (A4) at (3.5, 1.8) {4};
\node[nonsmoker] (A5) at (2, 0.8) {5};

\node[font=\bfseries] at (7, 4.5) {$t = 2$};
\node[font=\scriptsize, gray] at (7, 4.1) {Network $D_2$};

\draw[eqlink] (6, 3) -- (8, 3);              
\draw[eqlink] (6, 3) -- (8.5, 1.8);          
\draw[eqlink] (6, 3) -- (7, 0.8);            
\draw[eqlink] (5.5, 1.8) -- (8, 3);          
\draw[eqlink] (5.5, 1.8) -- (8.5, 1.8);      
\draw[eqlink] (5.5, 1.8) -- (7, 0.8);        
\draw[link] (8.5, 1.8) -- (7, 0.8);          

\node[smoker]   (B1) at (6, 3) {1};
\node[smoker]    (B2) at (8, 3) {2};
\node[smoker]   (B3) at (5.5, 1.8) {3};
\node[nonsmoker] (B4) at (8.5, 1.8) {4};
\node[nonsmoker] (B5) at (7, 0.8) {5};

\node[font=\bfseries] at (12, 4.5) {$t = 3$};
\node[font=\scriptsize, gray] at (12, 4.1) {Network $D_3$};

\draw[link] (11, 3) -- (13, 3);          
\draw[eqlink] (11, 3) -- (13.5, 1.8);      
\draw[eqlink] (13, 3) -- (13.5, 1.8);      
\draw[link] (10.5, 1.8) -- (12, 0.8);    
\draw[link] (13.5, 1.8) -- (12, 0.8);    

\node[smoker]    (C1) at (11, 3) {1};
\node[smoker]    (C2) at (13, 3) {2};
\node[nonsmoker] (C3) at (10.5, 1.8) {3};
\node[smoker]    (C4) at (13.5, 1.8) {4};
\node[nonsmoker] (C5) at (12, 0.8) {5};

\draw[trans] (4.2, 2.2) -- (5.0, 2.2);
\draw[trans] (9.2, 2.2) -- (10.0, 2.2);

\node[smoker, minimum size=14pt] at (3.5, -0.3) {};
\node[font=\scriptsize, right] at (3.9, -0.3) {$y_{it} = 1$ (smoker)};
\node[nonsmoker, minimum size=14pt] at (7, -0.3) {};
\node[font=\scriptsize, right] at (7.4, -0.3) {$y_{it} = 0$ (non-smoker)};

\end{tikzpicture}
\\[8pt]
{\small (b) Movement of the social influence term under Assumption~\ref{ass:persistence}}\\[4pt]
\begin{tikzpicture}[
    scale=0.95,
    l1/.style={blue, line width=1.2pt},
    l4/.style={blue, line width=1.2pt, dashed},
    lg/.style={gray!70, line width=1pt},
    m1/.style={circle, fill=blue, inner sep=1.7pt},
    m4/.style={circle, draw=blue, fill=white, inner sep=1.4pt, line width=0.9pt},
    mg/.style={circle, fill=gray!70, inner sep=1.5pt},
    chg/.style={<->, blue, line width=0.8pt, >=stealth}
]
\draw[->, gray!80] (0,0) -- (0,3.55) node[above, font=\scriptsize, black] {$\lambda(w_{it})$};
\draw[gray!80] (0,0) -- (10.2,0);
\foreach \x/\t in {1.6/1, 5.0/2, 8.4/3} {
  \draw[gray!80] (\x,0.08) -- (\x,-0.08);
  \node[below, font=\scriptsize] at (\x,-0.1) {$t = \t$};
}
\node[font=\scriptsize, gray, anchor=east, align=right] at (4.85,3.3) {transition:\\ common change $d_1$ within a type};
\node[font=\scriptsize, gray, anchor=west, align=left] at (5.15,3.3) {stable:\\ no change, $d_2 \equiv 0$};
\draw[gray!50, dashed] (5.0,0) -- (5.0,3.08);
\draw[lg] (1.6,0.24) -- (5.0,0.60) -- (8.4,0.60);   
\draw[lg] (1.6,1.20) -- (5.0,0.96) -- (8.4,0.96);   
\draw[lg] (1.6,2.40) -- (5.0,2.04) -- (8.4,2.04);   
\draw[l1] (1.6,0.72) -- (5.0,1.68) -- (8.4,1.68);   
\draw[l4] (1.6,1.92) -- (5.0,2.88) -- (8.4,2.88);   
\node[mg] at (1.6,0.24) {}; \node[mg] at (5.0,0.60) {}; \node[mg] at (8.4,0.60) {};
\node[mg] at (1.6,1.20) {}; \node[mg] at (5.0,0.96) {}; \node[mg] at (8.4,0.96) {};
\node[mg] at (1.6,2.40) {}; \node[mg] at (5.0,2.04) {}; \node[mg] at (8.4,2.04) {};
\node[m1] at (1.6,0.72) {}; \node[m1] at (5.0,1.68) {}; \node[m1] at (8.4,1.68) {};
\node[m4] at (1.6,1.92) {}; \node[m4] at (5.0,2.88) {}; \node[m4] at (8.4,2.88) {};
\draw[chg] (0.9,0.72) -- (0.9,1.68) node[midway, left, font=\scriptsize, blue] {$d_1$};
\draw[chg] (0.9,1.92) -- (0.9,2.88) node[midway, left, font=\scriptsize, blue] {$d_1$};
\draw[blue!50, dotted] (0.9,0.72) -- (1.6,0.72); \draw[blue!50, dotted] (0.9,1.68) -- (5.0,1.68);
\draw[blue!50, dotted] (0.9,1.92) -- (1.6,1.92); \draw[blue!50, dotted] (0.9,2.88) -- (5.0,2.88);
\node[right, font=\scriptsize, blue] at (8.5,2.88) {agent 4};
\node[right, font=\scriptsize, gray] at (8.5,2.04) {agent 2};
\node[right, font=\scriptsize, blue] at (8.5,1.68) {agent 1};
\node[right, font=\scriptsize, gray] at (8.5,0.96) {agent 3};
\node[right, font=\scriptsize, gray] at (8.5,0.60) {agent 5};
\end{tikzpicture}

\caption{Network-type equivalence and the movement of the social influence term. Panel (a): conditionally network-type equivalent agents. At $t = 1$, agents~1 and~4 are connected to the same set of third agents (blue links): both link to agents~2 and~5, and neither links to agent~3. Here the realized neighborhoods are consistent with $\rho_{14,1} = 0$ provided $y_{1,0} = y_{4,0}$. By contrast, the realized neighborhoods of agents 2 and 5 differ, since agent~5 is linked with agent~3 while agent~2 is not. The illustration takes this to reflect different link functions, although a single realized graph cannot by itself rule out equal link probabilities. At $t = 2$, agents~1 and~3 share the same lagged outcome ($y_{1,1} = y_{3,1} = 0$) and are connected to the same set of third agents: both link to agents~2, 4, and~5. Their realized neighborhoods coincide in this illustration, which is consistent with $\rho_{13,2} = 0$. At $t = 3$, agents 1 and 2 share the same lagged outcome ($y_{1,2} = y_{2,2} = 1$) and, excluding their mutual link, are connected to the same set of third agents: among agents 3, 4, and 5, both link only to agent~4, consistent with $\rho_{12,3} = 0$. Panel (b): the social influence term $\lambda(w_{it})$ of the same five agents under Assumption~\ref{ass:persistence}. Agents 1 and 4, which are equivalent at $t = 1$, have different levels but the same change $d_1$ between $t = 1$ and $t = 2$ (parallel segments, equal arrows), while agents 2, 3, and 5, whose network types differ from theirs at $t = 1$, change by different amounts. Between $t = 2$ and $t = 3$, the term does not change for any agent. The levels of $\lambda$, like the agent effects $\mu_i$, are unrestricted. Agents 1 and 4 also switch in opposite directions between $t = 1$ and $t = 2$, which is the comparison that the identification argument uses.}
\label{fig:network_equiv}
\end{figure}

When the link function does not depend on lagged outcomes, so that
\[
f(w_{it},w_{jt};\,y_{it-1},y_{jt-1})=g(w_{it},w_{jt})
\]
for some function $g$, the conditional individual link function $f_{w_{it};\, y,y'}(s) = g(w_{it}, s)$ is the same for all $(y,y')$, and the distance reduces to $\rho_{ijt} = \sqrt{2}\,\|g_{w_{it}} - g_{w_{jt}}\|_2$. In this case, every period is informative for type matching, regardless of lagged outcomes.

I now state the formal assumptions. Throughout, $X_i = (X_{i1}, \ldots, X_{iT})$, $w_i = (w_{i1}, \ldots, w_{iT})$, and $\varepsilon_i = (\varepsilon_{i1}, \ldots, \varepsilon_{iT})$ denote agent $i$'s covariate path, latent-index path, and outcome shocks, and $\chi_i = (y_{i0}, X_i, w_i, \mu_i)$ collects the agent's primitives other than the shocks.

\begin{assumption}[Sampling and distributional structure]\label{ass:sampling}~
\begin{enumerate}[label=(\roman*)]
\item The vectors $(\chi_i, \varepsilon_i) = (y_{i0}, X_i, w_i, \mu_i, \varepsilon_i)$, $i = 1, \ldots, n$, are i.i.d.
\item For each $t$, the array $\{\eta_{ijt}\}_{i,j=1}^n$ is symmetric, with i.i.d.\ standard uniform entries above the diagonal. The arrays are independent across $t$ and independent of $\{(\chi_i, \varepsilon_i)\}_{i=1}^n$.
\item For each $t$, $w_{it}$ has a distribution $G_{w,t}$ on $[0,1]$, which may differ across periods.
\item The shocks $\varepsilon_{it}$ are standard logistic and i.i.d.\ across $i$ and $t$, and the vector $\varepsilon_i$ is independent of $\chi_i = (y_{i0}, X_i, w_i, \mu_i)$.
\item The outcomes $y_{it}$ and the adjacency matrices $D_t$, $t = 1, \ldots, T$, are generated by equations~\eqref{eq:outcome} and~\eqref{eq:network}.
\item $\lambda$ and $f$ are measurable, and $f$ is jointly symmetric: $f(w_1,w_2;\,y,y')=f(w_2,w_1;\,y',y)$.
\end{enumerate}
\end{assumption}

Assumption~\ref{ass:sampling}(i) includes the initial outcome and the agent effect in the vector that is sampled independently across agents. The initial outcome may depend on the covariate and latent-index paths and on the agent effect in an arbitrary way, and the agent effect may depend on the covariate and latent-index paths and on the initial outcome in an arbitrary way. Both are compatible with conditional likelihood, but part (iv) requires the outcome shocks to be independent of them. This excludes, for example, $y_{i0} = \ind\{\varepsilon_{i1} > 0\}$, under which $\varepsilon_{i1}$ would remain unconditionally logistic and independent of $(X_i, w_i)$ but would be truncated given the initial outcome, so that the first transition probability used in Lemma~\ref{lem:ratio} would be wrong. Part (i) allows temporal dependence within an agent's covariate and latent paths but requires the primitives to be drawn independently across agents, which rules out aggregate shocks common to all agents but does not restrict the dependence that arises endogenously through the network. Part (ii) specifies the link shocks for all waves jointly: within each wave they are symmetric ($\eta_{ijt} = \eta_{jit}$, consistent with undirected links) and i.i.d.\ above the diagonal, and across waves they are independent, so that the links of different dyads are independent given the agents' latent indices and realized lagged outcomes (dyads that share an agent share her realized lagged outcome, so they are not independent given the primitives alone). The independence of the link shocks from the agent primitives means that unobserved factors affecting both linking and outcomes operate through $w_{it}$. In particular, the agent effect $\mu_i$ does not enter the link function as a separate argument, which is what distinguishes it from $\lambda(w_{it})$, although it may be associated with the links through its dependence on $w_i$ and through the lagged outcomes that enter $f$. This is a restriction: any time-varying part of the unobserved heterogeneity must be a function of the index that drives links, and whatever affects the outcome without shaping ties must be constant over time. Within that restriction the split is a definition: any persistent characteristic that shapes ties belongs to $w_{it}$, and $\mu_i$ is what remains. The arrays have dimension $n \times n$ at sample size $n$, a dependence left implicit in the notation.

\label{loc:uniform}Assumption~\ref{ass:sampling}(iii) lets the distribution of the latent index differ across periods, since the index moves over time. Nothing in the identification argument uses a particular marginal distribution. The population codegree distance of Section~\ref{sec:estimation} integrates against the period-specific joint distribution of $(w_{lt}, y_{l,t-1})$. The uniform distribution of $\eta_{ijt}$ makes $f(w_{it}, w_{jt};\, y_{it-1}, y_{jt-1})$ a conditional linking probability.

Assumption~\ref{ass:sampling}(iv) contains two substantive restrictions. The logistic distribution is what makes the conditional likelihood work: the cancellation of the social influence terms relies on the log-linear structure of the logistic odds ratio, a property first exploited by \citet{andersen1970} and \citet{chamberlain1980} to eliminate fixed effects through sufficient statistics and shown by \citet{magnac2004} to be essentially unique to the logistic family. See also \citet{chamberlain2010} and, for what survives without it, the partial identification results of \citet{honore2006bounds}. The independence of $\varepsilon_i$ from $(y_{i0}, X_i, w_i, \mu_i)$ is used in the factorization of path probabilities in Lemma~\ref{lem:ratio}, which conditions on the entire covariate and latent-index paths, including their future values, and it excludes feedback from realized outcomes into future covariates or into the future latent index. This is a real restriction in the motivating applications: it rules out, for instance, a risk-attitude index $w_{it}$ that responds to the agent's own past smoking, and covariates such as alcohol use that respond to past smoking shocks. Serial dependence of any form in $(X_i, w_i)$ remains allowed. Assumption~\ref{ass:sampling}(v) ties the observed data to the structural equations, and part (vi) is the joint symmetry discussed after equation~\eqref{eq:network}.

A remark on conditioning is in order before the next assumption. When network types vary continuously, the events $\{\rho_{ijt} = 0\}$ and $\{f_{w_{it}} = \varphi\}$ have probability zero, and conditional expectations given these events are not defined in the elementary sense. For statements about population objects, such conditional expectations are defined as limits over shrinking neighborhoods, in the manner of Section 2.2.2 of \citet{auerbach2022}. Write $f_w := (f_{w;\,y,y'})_{(y,y') \in \{0,1\}^2}$ for the full profile of conditional individual link functions of a type $w$, and $\|f_u - \varphi\|_2^2 := \sum_{y,y' \in \{0,1\}} \int_0^1 \big(f(u,s;\,y,y') - \varphi_{y,y'}(s)\big)^2 ds$ for the distance between such a profile and a candidate profile $\varphi$. For any integrable random vector $\Pi_{ij}$ indexed at the pair level,
\begin{equation}\label{eq:condexp_limit}
\E\left[\Pi_{ij} \,\middle|\, \rho_{ijt} = 0\right] := \lim_{h \downarrow 0}\, \E\left[\Pi_{ij} \,\middle|\, \rho_{ijt} \le h\right],
\end{equation}
and, for objects $\Pi_j$ attached to a single agent $j$ and a network type $\varphi$ in the support of $f_{w_{jt}}$,
\begin{equation*}
\E\big[\Pi_j \,\big|\, f_{w_{jt}} = \varphi\big] := \lim_{h \downarrow 0}\, \E\Big[\Pi_j \,\Big|\, w_{jt} \in \big\{u \in [0,1] : \|f_u - \varphi\|_2 \le h\big\}\Big],
\end{equation*}
evaluated at $\varphi = f_{w_{it}}$ when the conditioning is on agent $i$'s network type. Analogously, for conditioning on covariate stability,
\begin{equation*}
\E\left[\,\cdot \,\middle|\, X_{i,s+2} = X_{i,s+1}\right] := \lim_{h \downarrow 0}\, \E\left[\,\cdot \,\middle|\, \|X_{i,s+2} - X_{i,s+1}\|^2 \le h\right].
\end{equation*}
Whenever a conditional expectation of this form is used, the relevant limit is assumed to exist. Separate limits of this kind need not coincide with the limit of a criterion that localizes in several directions at once with particular relative bandwidths. For that reason, the identification results in the continuous case (Theorem~\ref{thm:identification}(b) and Corollary~\ref{cor:lambda}(b)) localize jointly in the codegree distance and in covariate stability, with explicit bounds, and the estimation theory of Section~\ref{sec:asymptotics} works directly with the localized criterion that the estimator computes. The convention only gives a meaning to informal statements such as ``conditional on the network type,'' and no formal result relies on it. When types or covariates have discrete support, the conditioning events have positive probability, the limits coincide with elementary conditional expectations, and no limiting argument is needed.

The identification argument does not require agents with the same network type to have the same value of the social influence function. It requires a restriction on how $\lambda(w_{it})$ moves between periods. Let $\mathcal{S} \subseteq \{1, \ldots, T-2\}$ denote the set of transition periods used in estimation.

\begin{assumption}[Latent index dynamics]\label{ass:persistence}~
\begin{enumerate}[label=(\roman*)]
\item (Movement common to network types.) For each $t \le T-1$ there is a measurable function $d_t : [0,1] \to \R$ such that $\lambda(w_{i,t+1}) - \lambda(w_{it}) = d_t(w_{it})$ almost surely, and $d_t$ is uniformly continuous with respect to the network-type distance: for every $\kappa > 0$ there exists an $r > 0$ such that
\[
\sup_{u,v \in [0,1]:\; \|f_{u;\,y,\cdot} - f_{v;\,y,\cdot}\|_2 \le r \;\text{for some } y \in \{0,1\}} \big(d_t(u) - d_t(v)\big)^2 \;\le\; \kappa,
\]
where $\|f_{u;\,y,\cdot} - f_{v;\,y,\cdot}\|_2$ is the distance between conditional individual link functions in equation~\eqref{eq:netdist}. The magnitude of $d_t$ is unrestricted.
\item (Stability after the transition.) For each $s \in \mathcal{S}$, $d_{s+1} \equiv 0$.
\end{enumerate}
\end{assumption}

Assumption~\ref{ass:persistence} is the restriction that the time-varying generalization requires, and it is what the identification result from observables uses. Part (i) says that the change in the social influence term between consecutive periods is common to agents who form links in the same way and similar for agents who form links in similar ways. In particular, $d_t$ is constant on network-type equivalence classes, and it may be large and need not shrink with the sample. The uniform continuity is stated in the metric of equation~\eqref{eq:netdist}, in the manner of Assumption 3 in \citet{auerbach2022}, and the quantifier ``for some $y$'' matters: matched pairs share a single realized lagged outcome at each period, so the data certify proximity of the link functions only in the stratum the pair realizes. Part (ii) says that each transition period used in estimation is followed by a stable period.

The assumption restricts the changes of the social influence function, not its level. Two agents with the same network type may have different values of $\lambda(w_{it})$ at every date and different agent effects $\mu_i$. What they share is the movement of $\lambda$ between consecutive dates, which Figure~\ref{fig:network_equiv}(b) illustrates for the agents of panel (a). This is weaker than requiring equal values of $\lambda$ within a network type, as in \citet{auerbach2022} and \citet{gueyap2026}, and it is what allows the model to carry an agent effect that network matching does not equalize. The assumption is plausible when the same characteristics that drive the outcome also drive the formation of ties, so that a change in those characteristics shows up in both. It fails when the movement of the social influence function differs across agents of the same network type, as when the characteristic that moves affects the outcome but plays no role in the formation of ties and moves idiosyncratically. A movement common to every agent is allowed, whether or not it works through the formation of ties, because it is constant within every class. Since $\lambda$ is never observed, the assumption cannot be tested from the data alone. Three observations order its content. First, when $w_{it} = w_i$ is time-invariant, $d_t \equiv 0$ and both parts hold trivially. \label{loc:commonmap}Second, when the latent index evolves by a common map, $w_{i,t+1} = g_t(w_{it})$, the change is $d_t = \lambda \circ g_t - \lambda$, and part (i) holds if this induced change is constant on each network-type equivalence class and uniformly continuous in the metric of equation~\eqref{eq:netdist}. Movements of any size are then allowed. A common map that preserves the classes does not by itself ensure this when $\lambda$ varies within a class. For instance, with the lag-independent link function $f(u,v) = \tfrac14 + \tfrac14\big[(u - \tfrac12)^2 + (v - \tfrac12)^2\big]$, the indices $u$ and $1-u$ have identical link functions, the map $g(u) = 1 - u$ sends each index to a member of its own class, and yet with $\lambda(u) = u$ the change $d(u) = 1 - 2u$ equals $0.6$ at $u = 0.2$ and $-0.6$ at $u = 0.8$, so two equivalent indices move by different amounts. Examples in which the condition holds are a cohort-wide change in risk attitude between two survey waves whose size depends only on the student's position in the friendship structure, and the finite-type designs of Section~\ref{sec:montecarlo} that satisfy the assumption, in which the supported types have distinct link profiles, so that every equivalence class is a single type, and types move between network blocks by a common permutation. Third, part (ii) interacts with the choice of $\mathcal{S}$: if every interior period is used, it forces $d_2 = \cdots = d_{T-1} = 0$, so that the only unrestricted common movement is between periods 1 and 2. With longer panels, one may instead use alternating windows of movement and stability and restrict $\mathcal{S}$ accordingly. With $T = 3$, $\mathcal{S} = \{1\}$ and the restriction is $d_2 \equiv 0$. What the assumption excludes is movement that differs by a non-vanishing amount across agents with the same network type, for example a stationary autoregression with nondegenerate innovations. In that case, neither within-agent nor cross-agent comparisons remove the movement. Section~\ref{sec:asymptotics} relaxes parts (i) and (ii) to allow an idiosyncratic component whose maximum over the sample shrinks as the sample grows (Assumption~\ref{ass:drift}), which is the intermediate case of slowly evolving characteristics in short panels.

Appendix~\ref{sec:appendix_window} describes the observation window that the assumption requires in the smoking example, a scheduled school transition followed by a stable period, and how the waves of the Glasgow data map onto it.

Two consequences are used below. First, if agents $i$ and $j$ have the same network type at period $s$ in the stratum of their common lagged outcome, then $d_s(w_{is}) = d_s(w_{js})$, because the uniform-continuity condition applied at every $\kappa > 0$ forces $d_s$ to be constant on equivalence classes, and by part (i)
\[
\lambda(w_{i,s+1}) - \lambda(w_{is}) = d_s(w_{is}) = d_s(w_{js}) = \lambda(w_{j,s+1}) - \lambda(w_{js}),
\]
so the two agents' social influence terms change by the same amount between $s$ and $s+1$, whether or not their levels, or their agent effects, are equal. For agents whose network types are close but not equal, write $\varpi(r) = \sup\{|d_t(u) - d_t(v)| : \|f_{u;\,y,\cdot} - f_{v;\,y,\cdot}\|_2 \le r \text{ for some } y\}$ for the modulus of continuity in part (i), which tends to zero as $r \to 0$ and is what the estimation theory uses. Second, by part (ii), $\lambda(w_{m,s+2}) = \lambda(w_{m,s+1})$ for every agent $m$ and every $s \in \mathcal{S}$. The first consequence is what makes equivalence at period $s$, which the data can certify for pairs sharing the lagged outcome at $s-1$, sufficient for the cancellation in Lemma~\ref{lem:identity}. Equivalence at $s+1$ cannot be certified for the pairs the estimator uses, since their lagged outcomes differ at $s+1$ (Section~\ref{sec:estimation}).

\subsection{Identification of Slope and State Dependence Parameters}\label{sec:id_main}

The identification strategy combines three elements: (i) the Chamberlain conditional likelihood approach for logit models, which eliminates the agent effect $\mu_i$ and the level of $\lambda$ within each agent, (ii) cross-agent matching on network types to eliminate the change in the social influence function, and (iii) stability of covariates and of the social influence term in the period that follows the transition.

For any interior period $s \in \{1, \ldots, T-2\}$ (so that period $s+2$ is observed), define two outcome paths for agent $i$ that differ only in the ordering of outcomes at periods $s$ and $s+1$:
\begin{align*}
A_s &= \{y_{i0}, \ldots, y_{i,s-1}, \underbrace{0}_{y_{is}}, \underbrace{1}_{y_{i,s+1}}, y_{i,s+2}, \ldots, y_{iT}\}, \\
B_s &= \{y_{i0}, \ldots, y_{i,s-1}, \underbrace{1}_{y_{is}}, \underbrace{0}_{y_{i,s+1}}, y_{i,s+2}, \ldots, y_{iT}\},
\end{align*}
where the outcomes outside periods $s$ and $s+1$ are fixed at the same values in both paths. The total $y_{is} + y_{i,s+1} = 1$ is the same for both paths. However, the lagged values entering the model differ at periods $s+1$ and $s+2$: in path $A_s$ the lag at $s+1$ is $0$ and the lag at $s+2$ is $1$, whereas in path $B_s$ the lag at $s+1$ is $1$ and the lag at $s+2$ is $0$.

\begin{lemma}[Within-agent likelihood ratio]\label{lem:ratio}
Under Assumption~\ref{ass:sampling}, conditional on $\chi_i = (y_{i0}, X_i, w_i, \mu_i)$ with $X_{i,s+2} = X_{i,s+1}$ and $\lambda(w_{i,s+2}) = \lambda(w_{i,s+1})$, the likelihood ratio for agent $i$ between paths $B_s$ and $A_s$ satisfies:
\begin{equation}\label{eq:ratio}
\frac{\Prob(y_i \in B_s \mid \chi_i)}{\Prob(y_i \in A_s \mid \chi_i)} = \exp\left[(X_{is} - X_{i,s+1})'\beta + \alpha(y_{i,s-1} - y_{i,s+2}) + \lambda(w_{is}) - \lambda(w_{i,s+1})\right].
\end{equation}
The ratio does not depend on $\mu_i$.
\end{lemma}

The proof is in Appendix~\ref{sec:proof_lemratio}. For $T = 3$ (so $s = 1$) the content of the lemma, and of the comparison with within-agent conditioning, can be seen in one display. Write $c_{it} = X_{it}'\beta + \lambda(w_{it}) + \mu_i$ for the index without the lag and $\Gamma(c) = \log(1 + e^{c+\alpha}) - \log(1 + e^{c})$. Without imposing any stability condition, the exact log odds of the two switched paths are
\begin{equation}\label{eq:worked}
\log\frac{\Prob(y_i \in B_1 \mid \chi_i)}{\Prob(y_i \in A_1 \mid \chi_i)}
= (X_{i1} - X_{i2})'\beta + \alpha(y_{i0} - y_{i3})
+ \big[\lambda(w_{i1}) - \lambda(w_{i2})\big] + \big[\Gamma(c_{i3}) - \Gamma(c_{i2})\big].
\end{equation}
The agent effect cancels in $c_{i1} - c_{i2}$, which is the covariate term plus the first bracket, because it is constant over time. What survives of $\lambda$ there is its change between periods 1 and 2, not its level. This is the within-agent elimination of a fixed effect of \citet{chamberlain1980} and \citet{honore2000panel}, applied to $\lambda(w_{it}) + \mu_i$. Equation~\eqref{eq:ratio} imposes the stability condition of the lemma, under which $c_{i3} = c_{i2}$ and the last bracket vanishes. Equation~\eqref{eq:worked} imposes no stability condition, which is why the $\Gamma$ terms appear. The last bracket is the period-3 term, and $\mu_i$ does enter it through the level $c_{i2}$ unless $c_{i3} = c_{i2}$. Since $\Gamma'(c) = F(c+\alpha) - F(c)$ lies in $(-1, 1)$, the bracket is bounded by $|c_{i3} - c_{i2}| \le \|\beta\|\,\|X_{i3} - X_{i2}\| + |\lambda(w_{i3}) - \lambda(w_{i2})|$, which is free of $\mu_i$, and it vanishes when covariates and social influence are stable between periods 2 and 3, which is what the covariate kernel $K_2$ of Section~\ref{sec:estimation} and Assumption~\ref{ass:persistence}(ii) target. Within-agent conditioning can remove the first bracket only by assuming it away, as the fixed-effect specification of \citet{honore2000panel} does. Cross-agent matching removes it without that assumption: for two agents whose social influence terms change by the same amount between periods 1 and 2, the bracket is the same for both, so it cancels in the ratio of their two ratios.

The likelihood ratio in~\eqref{eq:ratio} still contains the nuisance term $\lambda(w_{is}) - \lambda(w_{i,s+1})$. To eliminate it, I compare two agents, whose nuisance terms cancel if and only if $\lambda(w_{is}) - \lambda(w_{i,s+1}) = \lambda(w_{js}) - \lambda(w_{j,s+1})$. Only equal changes are needed. The levels of $\lambda$ and the agent effects may differ across the two agents, because each cancels within its own agent.

To state the comparison with its complete conditioning information, fix an interior period $s$ and two agents $i \neq j$. Write $y_{a,-s}$ for agent $a$'s outcomes at periods other than $s$ and $s+1$ (including $y_{a0}$), and let
\[
\mathcal{F}_{ijs} = \sigma\big(X_i, X_j, y_{i,-s}, y_{j,-s}\big).
\]
The pair is \emph{discordant} at $s$ if $y_{is} + y_{i,s+1} = y_{js} + y_{j,s+1} = 1$ and $y_{is} \neq y_{js}$, an event denoted $\mathcal{D}_{ijs}$. Given $\mathcal{F}_{ijs}$, the only outcome configurations compatible with it are $(y_i, y_j) \in A_s \times B_s$ and $(y_i, y_j) \in B_s \times A_s$, with the outcomes outside the window fixed at their realized values. Define the latent event and the covariate-stability event
\begin{align*}
\mathcal{C}_{ijs} &= \big\{\lambda(w_{is}) - \lambda(w_{i,s+1}) = \lambda(w_{js}) - \lambda(w_{j,s+1});\ \lambda(w_{a,s+2}) = \lambda(w_{a,s+1}) \text{ for } a \in \{i,j\}\big\},\\
\mathcal{X}_{ijs} &= \big\{X_{a,s+2} = X_{a,s+1} \text{ for } a \in \{i,j\}\big\},
\end{align*}
and the index
\[
V_{ijs}(b, a) = (\Delta_s X_j - \Delta_s X_i)'b + a\,(y_{j,s-1} - y_{j,s+2} - y_{i,s-1} + y_{i,s+2}), \qquad \Delta_s X_i = X_{is} - X_{i,s+1}.
\]

\begin{lemma}[Conditional likelihood identity]\label{lem:identity}
Under Assumption~\ref{ass:sampling}, on the event $\mathcal{C}_{ijs} \cap \mathcal{X}_{ijs}$,
\begin{equation}\label{eq:identity}
\Prob\big(y_i \in A_s,\, y_j \in B_s \,\big|\, \mathcal{F}_{ijs},\, w_i,\, w_j,\, \mathcal{D}_{ijs}\big) = F\big(V_{ijs}(\beta, \alpha)\big).
\end{equation}
The same identity holds with $(\mu_i, \mu_j)$ added to the conditioning information. Consequently, for every event $E$ in $\sigma(\mathcal{F}_{ijs}, w_i, w_j, \mu_i, \mu_j)$ with $E \subseteq \mathcal{C}_{ijs} \cap \mathcal{X}_{ijs}$ and positive probability, $\Prob(y_i \in A_s, y_j \in B_s \mid \mathcal{F}_{ijs}, E, \mathcal{D}_{ijs}) = F(V_{ijs}(\beta, \alpha))$.
\end{lemma}

\begin{proof}
By Assumption~\ref{ass:sampling}(i) and (iv), the outcome paths of agents $i$ and $j$ are conditionally independent given $(\chi_i, \chi_j)$, and each agent's conditional path distribution depends only on its own primitives. Hence, given $\mathcal{F}_{ijs}$ and $(w_i, w_j, \mu_i, \mu_j)$,
\[
\frac{\Prob(y_i \in A_s, y_j \in B_s \mid \cdot)}{\Prob(y_i \in B_s, y_j \in A_s \mid \cdot)} = \frac{\Prob(y_i \in A_s \mid \chi_i)}{\Prob(y_i \in B_s \mid \chi_i)} \cdot \frac{\Prob(y_j \in B_s \mid \chi_j)}{\Prob(y_j \in A_s \mid \chi_j)}.
\]
On $\mathcal{C}_{ijs} \cap \mathcal{X}_{ijs}$, Lemma~\ref{lem:ratio} applies to both agents, the agent effects drop out of both ratios, and the equality of the changes in $\lambda$ gives
\begin{align*}
\frac{\Prob(y_i \in A_s, y_j \in B_s \mid \cdot)}{\Prob(y_i \in B_s, y_j \in A_s \mid \cdot)}
&= \exp\big[(\Delta_s X_j - \Delta_s X_i)'\beta + \alpha(y_{j,s-1} - y_{j,s+2} - y_{i,s-1} + y_{i,s+2})\big]\\
&= \exp\big(V_{ijs}(\beta,\alpha)\big).
\end{align*}
Given $\mathcal{D}_{ijs}$ and the outcomes outside the window, exactly two configurations are possible, with conditional probabilities proportional to $\Pi$ and $\Pi'$, where $\Pi/\Pi' = \exp(V_{ijs})$. Hence the conditional probability of $A_s \times B_s$ is $\Pi/(\Pi + \Pi') = F(V_{ijs})$ given $(\mathcal{F}_{ijs}, w_i, w_j, \mu_i, \mu_j, \mathcal{D}_{ijs})$. Since the right side is $\mathcal{F}_{ijs}$-measurable, integrating out $(\mu_i, \mu_j)$ gives~\eqref{eq:identity}, and the last statement follows by iterated expectations, because the right side equals $F(V_{ijs})$ for every $(w_i, w_j, \mu_i, \mu_j)$ in $E$.
\end{proof}

Lemma~\ref{lem:identity} is a conditional likelihood identity on a latent event. It is not yet an identification result, because $\mathcal{C}_{ijs}$ involves the unobserved index. The network supplies an observable event contained in $\mathcal{C}_{ijs}$, and Proposition~\ref{prop:bridge} makes precise in what sense it is observable. Let $\delta_{ijs} = \delta_s(w_{is}, w_{js};\, y_{i,s-1})$ denote the population codegree distance between agents $i$ and $j$ at period $s$, defined in Section~\ref{sec:estimation} for pairs with $y_{i,s-1} = y_{j,s-1}$, and let $\hat{\delta}_{ijs}$ be its sample counterpart, a statistic of the observed networks and lagged outcomes. Following \citet{auerbach2022}, identification is understood relative to the growing-network experiment: the observed data are the outcomes, covariates, and networks $\{(y_i, X_i)\}_{i \le n}$, $\{D_t\}_{t \le T}$ of samples of size $n \to \infty$ drawn under Assumption~\ref{ass:sampling}, and a pair-level quantity is identified in this experiment when it is the probability limit, uniformly over pairs, of a statistic of the observed data. A parameter is identified when it is determined by the joint law of the observed data together with such limits.

\begin{proposition}[The codegree distance as an identified object]\label{prop:bridge}
Suppose Assumption~\ref{ass:sampling} holds.
\begin{enumerate}[label=(\alph*)]
\item For every $s$, $\max_{i \ne j}|\hat{\delta}_{ijs}^2 - \delta_{ijs}^2|\,\ind\{y_{i,s-1} = y_{j,s-1}\} = O_p((\log n/n)^{1/2})$, so $\delta_{ijs}$ is identified in the growing-network experiment.
\item Under Assumption~\ref{ass:smooth} (continuity of $f$ and full support), for pairs sharing the lagged outcome $y$ at $s-1$, $\delta_{ijs} = 0$ if and only if $\|f_{w_{is};\,y,\cdot} - f_{w_{js};\,y,\cdot}\|_2 = 0$, and the modulus
\[
\tau(h) := \sup\big\{\|f_{u;\,y,\cdot} - f_{v;\,y,\cdot}\|_2 : y \in \{0,1\},\ \delta_s(u,v;\,y)^2 \le h\big\}
\]
tends to zero as $h \downarrow 0$. Consequently, under Assumption~\ref{ass:persistence}(i), $\delta_{ijs} = 0$ implies $d_s(w_{is}) = d_s(w_{js})$, and $\delta_{ijs}^2 \le h$ implies $|d_s(w_{is}) - d_s(w_{js})| \le \varpi(\tau(h))$, with $\varpi$ the modulus of continuity of $d_s$ defined after Assumption~\ref{ass:persistence}.
\item In finite-type designs, under Assumption~\ref{ass:discrete}(i)--(ii) of Section~\ref{sec:finitetype}, $\delta_{ijs} = 0$ if and only if the two agents are network-type equivalent on the support, in which case $d_s(w_{is}) = d_s(w_{js})$, and every positive value of $\delta_{ijs}^2$ is at least $\bar\epsilon > 0$, with $\bar\epsilon$ as in Assumption~\ref{ass:discrete}(v).
\end{enumerate}
\end{proposition}

Part (a) is Lemma~\ref{lem:codegree_consist}(a) and (d) in Appendix~\ref{sec:appendix}, and parts (b) and (c) are proved in Appendix~\ref{sec:proof_thm1b}. Define the observable selection event and its neighborhood version,
\begin{align*}
\mathcal{O}_{ijs} &= \big\{y_{i,s-1} = y_{j,s-1},\ \delta_{ijs} = 0\big\} \cap \mathcal{X}_{ijs},\\
\mathcal{O}_{ijs}(h) &= \big\{y_{i,s-1} = y_{j,s-1},\ \delta_{ijs}^2 \le h\big\} \cap \big\{\|X_{a,s+2} - X_{a,s+1}\|^2 \le h \text{ for } a \in \{i,j\}\big\},
\end{align*}
and the regressor $Z_{ijs} = \big(\Delta_s X_i - \Delta_s X_j,\ y_{i,s-1} - y_{i,s+2} - y_{j,s-1} + y_{j,s+2}\big)'$, so that $V_{ijs}(b,a) = -Z_{ijs}'(b', a)'$. Write $\theta = (\beta', \alpha)'$. The event $\mathcal{O}_{ijs}$ has positive probability when types and covariates have discrete support, and probability zero when they are continuous. The theorem treats the two cases separately.

\begin{theorem}[Identification of $\beta$ and $\alpha$]\label{thm:identification}
Suppose Assumptions~\ref{ass:sampling} and~\ref{ass:persistence} hold and fix $s \in \mathcal{S}$.
\begin{enumerate}[label=(\alph*)]
\item (Exact matching.) Suppose that $\Prob(\mathcal{O}_{ijs} \cap \mathcal{D}_{ijs}) > 0$ and that $\delta_{ijs} = 0$ implies $d_s(w_{is}) = d_s(w_{js})$ for pairs sharing the lagged outcome at $s-1$. By Proposition~\ref{prop:bridge}, the implication holds in finite-type designs under Assumption~\ref{ass:discrete}(i)--(ii), without the continuity requirement in Assumption~\ref{ass:persistence}(i), and in general under Assumption~\ref{ass:smooth}. Then $\mathcal{O}_{ijs} \subseteq \mathcal{C}_{ijs} \cap \mathcal{X}_{ijs}$ and
\begin{equation}\label{eq:id}
\Prob\big(y_i \in A_s,\, y_j \in B_s \,\big|\, \mathcal{F}_{ijs},\, \mathcal{O}_{ijs},\, \mathcal{D}_{ijs}\big) = F\big(V_{ijs}(\beta, \alpha)\big).
\end{equation}
If, in addition, $\E[Z_{ijs}Z_{ijs}' \mid \mathcal{O}_{ijs}, \mathcal{D}_{ijs}]$ is positive definite, then $(\beta, \alpha)$ is the unique maximizer of the conditional expected log-likelihood $\E[m_{ijs}(b,a) \mid \mathcal{O}_{ijs}, \mathcal{D}_{ijs}]$, with $m_{ijs}$ defined in Section~\ref{sec:estimation}, and is therefore identified.
\item (Continuous types.) Suppose Assumption~\ref{ass:smooth} holds and $\Prob(\mathcal{O}_{ijs}(h) \cap \mathcal{D}_{ijs}) > 0$ for every $h > 0$. Then, on $\mathcal{O}_{ijs}(h)$,
\begin{equation}\label{eq:id_nbhd}
\begin{aligned}
\Prob\big(y_i \in A_s,\, y_j \in B_s \,\big|\, \mathcal{F}_{ijs}, w_i, w_j, \mu_i, \mu_j, \mathcal{D}_{ijs}\big) &= F\big(V_{ijs}(\beta, \alpha) + \Delta_{ijs}\big),\\
|\Delta_{ijs}| &\le \varpi_1(h) := \varpi(\tau(h)) + 2\|\beta\|\,h^{1/2},
\end{aligned}
\end{equation}
where the remainder $\Delta_{ijs}$ may depend on the agent effects but its bound does not. Hence $\big|\Prob(y_i \in A_s, y_j \in B_s \mid \mathcal{F}_{ijs}, \mathcal{O}_{ijs}(h), \mathcal{D}_{ijs}) - F(V_{ijs}(\beta,\alpha))\big| \le \varpi_1(h)/4 \to 0$ as $h \downarrow 0$. If, in addition, $\sup_{0 < h \le \bar h} \E[\|Z_{ijs}\|^4 \mid \mathcal{O}_{ijs}(h), \mathcal{D}_{ijs}] < \infty$ for some $\bar h > 0$ and $\liminf_{h \downarrow 0} \lambda_{\min}\big(\E[F(Z_{ijs}'\theta)(1 - F(Z_{ijs}'\theta))Z_{ijs}Z_{ijs}' \mid \mathcal{O}_{ijs}(h), \mathcal{D}_{ijs}]\big) > 0$, then for every sufficiently small $h$ the localized criterion $M_h(b,a) = \E[m_{ijs}(b,a) \mid \mathcal{O}_{ijs}(h), \mathcal{D}_{ijs}]$ has a unique maximizer $\theta_h$ on the parameter space $\Theta$, a compact convex set with $(\beta, \alpha)$ in its interior as in Assumption~\ref{ass:regularity}(ii), and $\theta_h \to (\beta, \alpha)$ as $h \downarrow 0$. In this sense $(\beta, \alpha)$ is identified as the limit of the localized comparison, which is the object that the estimator of Section~\ref{sec:estimation} computes.
\end{enumerate}
\end{theorem}

\begin{proof}
(a) On $\mathcal{O}_{ijs}$, the implication assumed in part (a) gives $d_s(w_{is}) = d_s(w_{js})$, hence $\lambda(w_{i,s+1}) - \lambda(w_{is}) = \lambda(w_{j,s+1}) - \lambda(w_{js})$. By Assumption~\ref{ass:persistence}(ii), $\lambda(w_{a,s+2}) = \lambda(w_{a,s+1})$ for $a \in \{i, j\}$. Thus $\mathcal{O}_{ijs} \subseteq \mathcal{C}_{ijs} \cap \mathcal{X}_{ijs}$. Since $\mathcal{O}_{ijs}$ is determined by $(X_i, X_j, y_{i,s-1}, y_{j,s-1}, w_{is}, w_{js})$, it belongs to $\sigma(\mathcal{F}_{ijs}, w_i, w_j)$, and~\eqref{eq:id} follows from Lemma~\ref{lem:identity}. The agent effects never enter. The expected conditional log-likelihood is then that of a correctly specified logit in the regressor $Z_{ijs}$, which is strictly concave with a unique maximizer at $(\beta, \alpha)$ under the rank condition. Part (b) is proved in Appendix~\ref{sec:proof_thm1b}.
\end{proof}

Four remarks clarify the theorem. First, the conditioning information is explicit: the probability in~\eqref{eq:id} conditions on the covariate paths and on the outcomes outside the window, which is why its right side varies with them. Second, the rank condition is imposed on the population actually selected, including the shared lagged outcome at $s-1$. On that population the second component of $Z_{ijs}$ reduces to $y_{j,s+2} - y_{i,s+2}$, and a rank condition imposed without the shared lag would not suffice: if $X_{i,s} - X_{i,s+1}$ were a function of $y_{i,s-1}$ alone, the first component would vanish on every selected pair while having positive variance on the larger population. Third, $\mathcal{O}_{ijs}$ can be refined by any further event in $\sigma(\mathcal{F}_{ijs}, w_i, w_j)$ without affecting~\eqref{eq:id}. This covers the aggregated distance used by the estimator, which averages period-specific distances over periods at which the pair shares the lagged outcome. For a discordant pair the set of such periods is the same under both orderings of the swapped outcomes: periods $s+1$ and $s+2$ are excluded under both, because the lagged outcomes differ there, and every other period's lag lies outside the window. In part (b), the same refinement replaces $\delta_{ijs}^2 \le h$ by a bound on the aggregated distance, since $\delta_{ijs}^2 \le T\bar{\delta}_{ij}^2$ whenever period $s$ contributes to the aggregate. Finally, the theorem places no restriction on the agent effects and does not require agents with the same network type to have the same value of $\lambda$: the comparison uses only the changes of the social influence term, which are common within network types by Assumption~\ref{ass:persistence}(i).

\begin{lemma}[Minimum number of time periods]\label{lem:mintime}
The identification strategy requires $T \geq 3$ observed periods (plus the initial condition $y_{i0}$), since it uses three consecutive periods $(s, s+1, s+2)$. The condition on period $s+2$ is needed because swapping outcomes at $(s, s+1)$ changes the lag entering the model at $s+2$. With additional time periods, several transition periods can be used, providing additional identifying variation.
When $T = 2$ and $\alpha \neq 0$, the within-agent conditional likelihood based on the transition event $y_{i1} + y_{i2} = 1$ depends on the unknown function $\lambda$ through a non-log-linear term that cannot be eliminated by the conditioning used in this paper. Consequently, the parameters $(\beta, \alpha)$ are not identified \emph{by this strategy} at $T = 2$. The lemma does not rule out identification through other functionals of the data. (When $\alpha = 0$ is known a priori, the offending term equals one, and $\beta$ is identified at $T = 2$ by the cross-agent comparison of Lemma~\ref{lem:identity} applied to the two-period window, which needs neither period $s+2$ nor Assumption~\ref{ass:persistence}(ii). The impossibility is specific to the dynamic parameter.)
\end{lemma}
\begin{proof}
With $T = 2$, the only paths satisfying $y_{i1} + y_{i2} = 1$ are $A_1 = (y_{i0}, 0, 1)$ and $B_1 = (y_{i0}, 1, 0)$. Define $c_{i1} = X_{i1}'\beta + \alpha y_{i0} + \lambda(w_{i1}) + \mu_i$ and $c_{i2} = X_{i2}'\beta + \lambda(w_{i2}) + \mu_i$. Conditional on $\chi_i$, the likelihoods are:
\begin{align*}
\Prob(y_i \in A_1 \mid \chi_i) &= \bigl[1 - F(c_{i1})\bigr] \cdot F(c_{i2}) = \frac{e^{c_{i2}}}{(1 + e^{c_{i1}})(1 + e^{c_{i2}})}, \\
\Prob(y_i \in B_1 \mid \chi_i) &= F(c_{i1}) \cdot \bigl[1 - F(c_{i2} + \alpha)\bigr] = \frac{e^{c_{i1}}}{(1 + e^{c_{i1}})(1 + e^{c_{i2} + \alpha})},
\end{align*}
where the index at period $2$ in path $B_1$ is $c_{i2} + \alpha$ because the lag $y_{i1} = 1$ adds $\alpha$ to the linear index. The conditional probability is:
\[
\Prob(y_i \in A_1 \mid y_{i1} + y_{i2} = 1, \chi_i) = \frac{e^{c_{i2}}(1 + e^{c_{i2} + \alpha})}{e^{c_{i2}}(1 + e^{c_{i2} + \alpha}) + e^{c_{i1}}(1 + e^{c_{i2}})}.
\]
When $\alpha \neq 0$, this does not reduce to the logistic form $F(V)$ for any index $V$ that is free of $\lambda$ and $\mu_i$: the term $(1 + e^{c_{i2} + \alpha})/(1 + e^{c_{i2}})$, which depends on $\lambda(w_{i2})$ and $\mu_i$ through $c_{i2}$, enters the conditional probability in a non-log-linear way. This is the same obstacle that \citet{honore2000panel} face with $T = 2$. (At $\alpha = 0$ the term equals one and the expression reduces to $F(c_{i2} - c_{i1})$, in which $\mu_i$ has already canceled and the remaining term $\lambda(w_{i2}) - \lambda(w_{i1}) = d_1(w_{i1})$ is removed by matching on network type as in Lemma~\ref{lem:identity}, consistent with the caveat in the statement.) When $T \geq 3$, period $s + 2$ contributes the reciprocal factor $(1 + e^{c})/(1 + e^{c + \alpha})$ under the stability condition, canceling this term exactly (cf.\ the proof of Lemma~\ref{lem:ratio}). With $T = 2$, no such period exists, and the nuisance function $\lambda$ remains in the conditional likelihood.
\end{proof}

The identification strategy is closely related to \citet{honore2000panel}, whose path-swapping argument eliminates the individual fixed effect $\mu_i$ and, with it, the level of $\lambda$. The difference is the use of network-type matching to handle the change in $\lambda(w_{it})$, a time-varying component of the unobserved heterogeneity that within-agent techniques alone cannot eliminate.

Given identification of $(\beta, \alpha)$, what can be learned about the social influence function is its movement, not its level, and only on network classes that the observed network can recover. For an agent $j$ and a transition period $s$, write $\mathcal{T}_{js} = \{y_{js} + y_{j,s+1} = 1\}$ for the event that the agent switches between $s$ and $s+1$, and $Z_{js} = \big(X_{js} - X_{j,s+1},\ y_{j,s-1} - y_{j,s+2}\big)'$ for the within-agent regressor, so that the pair regressor of Theorem~\ref{thm:identification} is $Z_{ijs} = Z_{is} - Z_{js}$. Fix a lag stratum $y$ and a reference agent $i$ with $y_{i,s-1} = y$. The agents that link like $i$ at period $s$ are identified through the codegree distance, not through the full link profile: by Proposition~\ref{prop:bridge}(a), $\delta_{ijs}$ is identified in the growing-network experiment for every $j$ with $y_{j,s-1} = y$, whereas the profile $f_{w_{js};\,y,\cdot}$ can differ from that of $i$ at partner types that are absent from the population, where the network carries no information. Define the observable class of $i$ and its neighborhoods,
\[
\mathcal{M}_{is} = \big\{j \neq i : y_{j,s-1} = y,\ \delta_{ijs} = 0\big\}, \qquad
\mathcal{M}_{is}(h) = \big\{j \neq i : y_{j,s-1} = y,\ \delta_{ijs}^2 \le h\big\},
\]
and the stable-switch events $\mathcal{E}_{js} = \mathcal{T}_{js} \cap \{X_{j,s+2} = X_{j,s+1}\}$ and $\mathcal{E}_{js}(h) = \mathcal{T}_{js} \cap \{\|X_{j,s+2} - X_{j,s+1}\|^2 \le h\}$. Probabilities and expectations below are taken over a second agent $j$, which is independent of $i$ by Assumption~\ref{ass:sampling}(i), given the primitives $\chi_i$ of the reference agent. They depend on $i$ only through $(w_{is}, y)$, and, because membership in $\mathcal{M}_{is}$ or $\mathcal{M}_{is}(h)$ is decided by the identified distance $\delta_{ijs}$, they are functionals of the joint law of the observed data and of that distance.

\begin{corollary}[Identification of the common movement]\label{cor:lambda}
Suppose Assumptions~\ref{ass:sampling} and~\ref{ass:persistence} hold and $(\beta, \alpha)$ is known. Fix $s \in \mathcal{S}$, a lag stratum $y$, and a reference agent $i$ with $y_{i,s-1} = y$.
\begin{enumerate}[label=(\alph*)]
\item (Finite support.) Suppose Assumption~\ref{ass:discrete}(i)--(ii) holds, in which case the continuity requirement in Assumption~\ref{ass:persistence}(i) is not needed, and $\Prob(j \in \mathcal{M}_{is},\ \mathcal{E}_{js} \mid \chi_i) > 0$. Then $d_s(w_{is})$ is the unique solution $d$ of
\begin{equation}\label{eq:dident}
\E\big[\ind\{y_j \in B_s\} - F(Z_{js}'\theta - d) \,\big|\, j \in \mathcal{M}_{is},\ \mathcal{E}_{js},\ \chi_i\big] = 0,
\end{equation}
and, for every value $(x, \bar y)$ of $(X_j, y_{j,-s})$ that has positive probability given $\{j \in \mathcal{M}_{is}\} \cap \mathcal{E}_{js}$, $d_s(w_{is}) = z'\theta - F^{-1}\big(\Prob(y_j \in B_s \mid j \in \mathcal{M}_{is},\ \mathcal{E}_{js},\ X_j = x,\ y_{j,-s} = \bar y,\ \chi_i)\big)$, where $z$ is the value of $Z_{js}$ at $(x, \bar y)$.
\item (Continuous types.) Suppose Assumption~\ref{ass:smooth} holds and $\Prob(j \in \mathcal{M}_{is}(h),\ \mathcal{E}_{js}(h) \mid \chi_i) > 0$ for every $h > 0$. Then, for every $h > 0$, the equation
\begin{equation}\label{eq:dident_h}
\E\big[\ind\{y_j \in B_s\} - F(Z_{js}'\theta - d) \,\big|\, j \in \mathcal{M}_{is}(h),\ \mathcal{E}_{js}(h),\ \chi_i\big] = 0
\end{equation}
has a unique solution $d_h$, and $|d_h - d_s(w_{is})| \le \varpi_1(h)$, with $\varpi_1$ as in Theorem~\ref{thm:identification}(b). Hence $d_s(w_{is}) = \lim_{h \downarrow 0} d_h$ is identified.
\end{enumerate}
No recovery claim is made for reference agents whose class does not satisfy the stated overlap condition, for which equations~\eqref{eq:dident} and~\eqref{eq:dident_h} have no observable counterpart. A supported type whose covariates never remain stable between $s+1$ and $s+2$ is an example. The level of $\lambda$ is not identified: the outcome distribution is unchanged if $\lambda$ is replaced by $\lambda + c$ and $\mu_i$ by $\mu_i - c$ for any constant $c$, and, when $w_{it} = w_i$ is time-invariant, only the sum $\lambda(w_i) + \mu_i$ enters the outcomes and no decomposition of it is identified.
\end{corollary}

\begin{proof}
(a) The proof has two steps. The first justifies conditioning on the observable class. By Proposition~\ref{prop:bridge}(a), the distance $\delta_{ijs}$ is identified, so the conditional law of $(y_j, X_j)$ given $\{j \in \mathcal{M}_{is}\} \cap \mathcal{E}_{js}$ is a functional of the data and of that distance. By Proposition~\ref{prop:bridge}(c), $j \in \mathcal{M}_{is}$ if and only if $w_{js}$ is equivalent on the support to $w_{is}$ in stratum $y$, and Assumption~\ref{ass:discrete}(ii) then gives $d_s(w_{js}) = d_s(w_{is})$: the movement is common on the observable class. The second step is the within-agent inversion. Condition on $\chi_i$, $\chi_j$, $y_{j,-s}$, and $\mathcal{E}_{js}$. Since $(\chi_j, \varepsilon_j)$ is independent of $(\chi_i, \varepsilon_i)$, the conditional probability of $B_s$ depends only on $(\chi_j, y_{j,-s})$. On $\mathcal{T}_{js}$ the path of $j$ is $A_s$ or $B_s$, and the stability conditions of Lemma~\ref{lem:ratio} hold by $X_{j,s+2} = X_{j,s+1}$ and Assumption~\ref{ass:persistence}(ii), so Lemma~\ref{lem:ratio} and Assumption~\ref{ass:persistence}(i) give $\Prob(y_j \in B_s \mid \chi_i, \chi_j, y_{j,-s}, \mathcal{E}_{js}) = F\big(Z_{js}'\theta - d_s(w_{js})\big)$, which does not depend on $\mu_j$. On $\{j \in \mathcal{M}_{is}\}$ this equals $F(Z_{js}'\theta - d_s(w_{is}))$. The event $\{j \in \mathcal{M}_{is}\}$ is determined by $(w_{is}, w_{js}, y_{j,s-1})$ and $Z_{js}$ is a function of $(X_j, y_{j,-s})$, so iterated expectations over the coarser conditioning sets give~\eqref{eq:dident} at $d = d_s(w_{is})$ and the pointwise formula, whose conditioning event has positive probability. The left side of~\eqref{eq:dident} is strictly increasing in $d$, with derivative $\E[F'(Z_{js}'\theta - d) \mid \cdot\,] > 0$, so the solution is unique. (b) is proved in Appendix~\ref{sec:proof_thm1b}. The non-identification statements follow because $(\lambda + c, \mu_i - c)$ generates the same outcomes as $(\lambda, \mu_i)$ and because $\mu_i$ may be an arbitrary function of $w_i$ under Assumption~\ref{ass:sampling}(i).
\end{proof}

\label{loc:corcounter}The observable-class premise cannot be dropped. Suppose that the latent index takes only the values 0 and 1, each with positive probability, that each agent's index is reflected between waves 1 and 2, $w_{i2} = 1 - w_{i1}$ with $w_{i0} = w_{i1}$, and is constant afterwards, so that $\mathcal{S} = \{1\}$ with $T = 3$, that $\lambda(u) = u$, and that links are lag-independent with $f(u, v) = \tfrac12 + \tfrac14\big[uv(1-v) + vu(1-u)\big]$. The movements are $d_1(0) = 1$ and $d_1(1) = -1$, and $d_2 \equiv 0$. The full link profiles of the two indices differ, since $f(0, \cdot) \equiv \tfrac12$ while $f(1, v) = \tfrac12 + \tfrac14 v(1-v)$, and because $u \mapsto f(u, \cdot)$ is continuous and injective on the compact set $[0,1]$, its inverse is uniformly continuous, so the movement $d_1(u) = 1 - 2u$ satisfies the continuity condition of Assumption~\ref{ass:persistence}(i). Explicitly, $f(u, v) - f(u', v) = \tfrac14 v(u - u')(2 - u - u' - v)$ gives $\|f_{u;\,y,\cdot} - f_{u';\,y,\cdot}\|_2 \ge |u - u'|/\sqrt{640}$. But every partner in the population has index 0 or 1, and between such partners the link probability is always $\tfrac12$, so the codegree distance between the two indices is zero. Conditioning on the full profile would separate agents that move by $+1$ from agents that move by $-1$, but the network alone cannot distinguish these full-profile classes, so codegree-based conditioning cannot separate agents whose latent terms move by different amounts in this example. The example does not show that no other observed information could: the stated setup does not exclude, for instance, an observed initial covariate that reveals the index. Part (a) excludes the example through Assumption~\ref{ass:discrete}(ii), which requires the movement to be constant on classes of indices that are equivalent on the support, and part (b) excludes it through the full-support condition of Assumption~\ref{ass:smooth}.

The corollary is the within-agent counterpart of Theorem~\ref{thm:identification}: among switchers with stable covariates in an observable network class, the odds of switching down rather than up reveal the common movement of the social influence term, which is what the within-agent estimator of \citet{honore2000panel} omits. The corollary also explains why inverting the conditional choice probability given the covariates, the lag, and the network type does not recover $\lambda$ here: that probability averages $F(X_{it}'\beta + \alpha y_{i,t-1} + \lambda(w_{it}) + \mu_i)$ over the conditional distribution of $\mu_i$, and its inversion does not return $\lambda$. Appendix~\ref{sec:appendix_lambda} records these points and describes an estimator of $d_s$, for which no asymptotic theory is claimed.

\section{Estimation}\label{sec:estimation}

This section develops a feasible estimator that implements the identification strategy of Section~\ref{sec:identification}. The main challenge is that the conditioning events, network-type equivalence ($\rho_{ijt} = 0$) and covariate stability ($X_{i,s+2} = X_{i,s+1}$), hold with probability zero in continuous distributions, and the network distance $\rho_{ijt}$ is not directly observable.

\subsection{Estimating Network Proximity}

Following \citet{auerbach2022} and \citet{gueyap2026}, I use codegree information to construct a feasible measure of network similarity. In the static setting, the codegree function directly proxies the type distance. In the dynamic setting, agent $i$'s type at period $t$ consists of the index of unobserved social characteristics $w_{it}$ and the observed lagged outcome $y_{it-1}$. Since the lagged outcome is observed, the codegree can be conditioned on it to isolate the unobserved component.

Define the conditional codegree function for two agents with social characteristics $(w_{it}, w_{jt})$ and lagged outcomes $(y_{it-1}, y_{jt-1})$ as:
\begin{equation*}
p(w_{it}, w_{jt}; y_{it-1}, y_{jt-1}) = \int f(w_{it}, u;\, y_{it-1}, y) \, f(w_{jt}, u;\, y_{jt-1}, y) \, dG_t(u, y),
\end{equation*}
where $G_t$ denotes the joint population distribution of $(w_{lt},y_{lt-1})$ at period $t$. This distribution may change over time. The period dependence of the codegree function $p$ is left implicit in the notation. This function measures the expected fraction of common neighbors between an agent of type $(w_{it}, y_{it-1})$ and an agent of type $(w_{jt}, y_{jt-1})$: the integration runs over all possible third agents with type $u \in [0,1]$ and lagged outcome $y \in \{0,1\}$.

For each agent $i$ at period $t$, define the \emph{conditional individual codegree function} as:
\begin{equation*}
p_{w_{it};\, y_{it-1},y_{jt-1}}(\cdot) \equiv p(w_{it}, \cdot; y_{it-1}, y_{jt-1}) : [0, 1]\to [0, 1].
\end{equation*}
This is the codegree profile of agent $i$: for each potential partner type $w_{jt}$ and partner lagged outcome $y_{jt-1}$, it gives the expected fraction of common neighbors between agents $i$ and $j$.

The \emph{conditional codegree distance} for agents $i$ and $j$ with $y_{it-1} = y_{jt-1}$ at period $t$ is the $L^2(G_t)$ distance between their codegree profiles, with the partner's type and lagged outcome $(s, y')$ integrated against their joint population distribution $G_t$:
\begin{align*}
\delta_{ijt} = \left\|p_{w_{it};\, y_{it-1},\cdot} - p_{w_{jt};\, y_{jt-1},\cdot}\right\|_{L^2(G_t)}
= \left[\int \big(p_{w_{it};\, y_{it-1},y'}(s) - p_{w_{jt};\, y_{jt-1},y'}(s)\big)^2 \, dG_t(s, y')\right]^{1/2}.
\end{align*}
The $G_t$-weighting is deliberate: it is exactly the metric in which the empirical codegree statistic of equation~\eqref{eq:codegree_est} below converges (Lemma~\ref{lem:codegree_consist}(a) establishes $\hat{\delta}_{ijt}^2 \to_p \delta_{ijt}^{2} = \E_k[(p_{ikt} - p_{jkt})^2]$, an average over third agents drawn from the population, which is the display above), so a single object is both the population target and the probability limit. An equal-stratum-weight variant, summing over $y' \in \{0,1\}$ with Lebesgue $ds$, has the same zero set and defines an equivalent topology when $G_t$ has full support and $f$ is continuous (Assumption~\ref{ass:smooth}), but not in general: when the population has atoms or lacks support at some partner types, the two metrics can differ in what they distinguish. The variant is in any case not the limit of the sample statistic, and every assumption stated in terms of the codegree distance uses the $G_t$-weighted metric. Since $y_{it-1} = y_{jt-1}$, the distance depends only on the type difference $w_{it}$ versus $w_{jt}$. Equivalently, substituting the definition of $p$:
\[
\delta_{ijt} = \left[\!\int \left(\!\int \big[f(w_{it}, u;\, y_{it-1}, y) - f(w_{jt}, u;\, y_{it-1}, y)\big] \, f(s, u;\, y', y) \, dG_t(u, y)\right)^2 dG_t(s, y')\right]^{1/2}.
\]
Since the codegree function is obtained by integrating the individual link function against $f(s, u;\, y', y)\, dG_t(u,y)$, conditionally network-type equivalent agents ($\rho_{ijt} = 0$) share the same codegree profile when $f$ is continuous, so $\rho_{ijt} = 0$ then implies $\delta_{ijt} = 0$ (Lemma~\ref{lem:codegree_consist}(b)). Without continuity, the implication can fail at atoms of the type distribution, since $\rho$ integrates against Lebesgue measure and $\delta$ against $G_t$. The codegree distance therefore provides a feasible proxy for network-type equivalence: under the continuity and full-support conditions of Assumption~\ref{ass:smooth}, Lemma~\ref{lem:codegree_consist}(c) establishes the converse separation, so that pairs selected by $\delta_{ijt} \approx 0$ have $\rho_{ijt} \approx 0$ and hence, by Assumption~\ref{ass:persistence}(i), approximately equal changes of the social influence term between $t$ and $t+1$, which is what the identification argument requires.

The codegree distance can be consistently estimated from the observed adjacency matrices. For each period $t$, define the empirical codegree distance:
\begin{equation}\label{eq:codegree_est}
\hat{\delta}_{ijt} = \left[\frac{1}{n}\sum_{k=1}^n \left(\frac{1}{n}\sum_{l=1}^n D_{klt}(D_{ilt} - D_{jlt})\right)^2\right]^{1/2}.
\end{equation}
When $y_{it-1} = y_{jt-1} = y$, the difference $D_{ilt} - D_{jlt}$ is driven exclusively by the type difference $w_{it}$ versus $w_{jt}$: for any third agent $l$ with $y_{lt-1} = y'$, $\E[D_{ilt} - D_{jlt} \mid w, y] = f(w_{it}, w_{lt};\, y, y') - f(w_{jt}, w_{lt};\, y, y') = f_{w_{it};\, y, y'}(w_{lt}) - f_{w_{jt};\, y, y'}(w_{lt})$, which is zero when $\rho_{ijt} = 0$. When $y_{it-1} \neq y_{jt-1}$, the conditional individual link functions of the two agents are evaluated at different own lagged outcomes, so the codegree difference conflates type differences with the lagged-outcome difference and cannot reliably proxy the type distance. Accordingly, the estimator restricts attention to pairs sharing the same lagged outcome at the relevant period.

When $f$ does not depend on lagged outcomes, the codegree difference between two agents is driven by their type difference whatever their lagged outcomes, so the shared-lag restriction is unnecessary for interpreting the distance in that special case, and the codegree reduces to the codegree of the static case. The estimator defined below nevertheless retains the restriction, and equation~\eqref{eq:codegree_agg} still excludes the periods at which the lagged outcomes differ.

For a discordant pair, the stratification discards the distance at $s+1$: $y_{is} \neq y_{js}$ by construction, so at $s+1$ (and again at $s+2$) the codegree difference conflates the type difference with the lagged-outcome difference and cannot proxy the type distance. This is why Theorem~\ref{thm:identification} certifies equivalence only at period $s$, where the pair shares the lagged outcome, and lets Assumption~\ref{ass:persistence}(i) deliver the equal change of the social influence term between $s$ and $s+1$. Other periods at which the pair shares the lagged outcome can be added to the certification: they do not change the identity of Theorem~\ref{thm:identification} (third remark after the theorem), and they reduce the noise in the estimated distance.

Accordingly, I define the \emph{aggregated lag-stratified codegree distance}
\begin{equation}\label{eq:codegree_agg}
\hat{\bar{\delta}}_{ij}^2 \;=\; \frac{\sum_{t=1}^{T} \ind\{y_{i,t-1} = y_{j,t-1}\}\, \hat{\delta}_{ijt}^2}{\max\!\big(1, \sum_{t=1}^{T} \ind\{y_{i,t-1} = y_{j,t-1}\}\big)},
\end{equation}
with the convention $\hat{\bar{\delta}}_{ij}^2 = +\infty$ when the two agents never share a lagged outcome. Each period contributes only when the two agents share the same lagged outcome, so that every term isolates the latent type difference, and the average over periods requires the pair to display similar linking behavior throughout the panel rather than at a single wave. A small value of $\hat{\bar{\delta}}_{ij}^2$ is therefore evidence that $\rho_{ijt} \approx 0$ at every informative period, including period $s$, which is the observable event of Theorem~\ref{thm:identification}. With $T = 3$ and $s = 1$, the only informative period for a discordant pair is $t = 1$, and the aggregated distance reduces to $\hat{\delta}_{ij1}^2$.

\subsection{Feasible Estimator}

The feasible estimator is a kernel-weighted conditional maximum likelihood estimator. For each transition period $s \in \mathcal{S}$, define the set of ``discordant'' agent pairs:
\[
\mathcal{P}_s = \left\{(i,j) : i < j, \; y_{i,s-1} = y_{j,s-1}, \; y_{is} + y_{i,s+1} = 1, \; y_{js} + y_{j,s+1} = 1, \; y_{is} \neq y_{js}\right\}.
\]
These are pairs where both agents share the same lagged outcome at period $s-1$ and transition between 0 and 1 across periods $(s, s+1)$, but in different orders. The condition $y_{i,s-1} = y_{j,s-1}$ ensures that the empirical codegree distance $\hat{\delta}_{ijs}$ is a valid proxy for network-type equivalence at period $s$.

For each such pair and period, define the conditioning variable:
\[
v_{ijs}(b,a) = (\Delta_s X_i - \Delta_s X_j)'b + a(y_{i,s-1} - y_{i,s+2} - y_{j,s-1} + y_{j,s+2}),
\]
where $\Delta_s X_i = X_{is} - X_{i,s+1}$.

The sample objective function is:
\begin{equation}\label{eq:objective}
\Omega_n(b, a) = \sum_{s \in \mathcal{S}} \sum_{(i,j) \in \mathcal{P}_s} K_1\!\left(\frac{\hat{\bar{\delta}}_{ij}^2}{h_1}\right) K_2\!\left(\frac{\|X_{i,s+2} - X_{i,s+1}\|^2}{h_2}\right) K_2\!\left(\frac{\|X_{j,s+2} - X_{j,s+1}\|^2}{h_2}\right) \cdot m_{ijs}(b, a),
\end{equation}
where:
\begin{equation*}
m_{ijs}(b, a) = \ind(y_{is} = 0) \log F\left[v_{jis}(b,a)\right]  + \ind(y_{is} = 1) \log F\left[v_{ijs}(b,a)\right].
\end{equation*}

The kernel $K_1$ weights agent pairs by their aggregated codegree similarity. By Assumption~\ref{ass:persistence}(i) and Proposition~\ref{prop:bridge}(b)--(c), pairs with small aggregated codegree distance have social influence terms that change by approximately the same amount between $s$ and $s+1$, which is the cross-agent condition of Lemma~\ref{lem:identity}. The kernel $K_2$ weights observations by the local stability of covariates between periods $s+1$ and $s+2$. It approximates the condition $X_{i,s+2} = X_{i,s+1}$ of Lemma~\ref{lem:ratio}, while Assumption~\ref{ass:persistence}(ii) supplies the within-agent stability condition $\lambda(w_{i,s+2}) = \lambda(w_{i,s+1})$. The bandwidths $h_1$ and $h_2$ are sequences converging to zero as $n \to \infty$, except in finite-type designs with discrete covariates, where fixed bandwidths small enough relative to the separation gaps select exact matches (Assumption~\ref{ass:discrete}(v) in Section~\ref{sec:asymptotics}).

The estimator is:
\begin{equation}\label{eq:estimator}
(\hat{\beta}, \hat{\alpha}) = \arg\max_{(b,a)} \; \Omega_n(b, a).
\end{equation}

Throughout the rest of the paper I write $\theta = (\beta', \alpha)'$ for the parameter vector and $\hat{\theta} = (\hat{\beta}', \hat{\alpha})'$ for the estimator in~\eqref{eq:estimator}, so that $\Theta$ is the parameter space for $\theta$ and $(b, a)$ denotes a generic candidate value.

The asymptotic results maximize over a compact parameter space $\Theta$ (Assumption~\ref{ass:regularity}), where a maximizer always exists. The simulations and the application instead seek an unconstrained maximum and record a separated likelihood, which has none, as a failed fit. Separation in a finite sample does not conflict with the theory, which concerns the constrained maximizer as $n$ grows. The Firth-penalized fits of Section~\ref{sec:empirical} maximize a different criterion and are not covered by Theorems~\ref{thm:consistency} and~\ref{thm:normality}.

The kernel $K_2$ is necessitated by the dynamic structure, and it serves only the covariate condition of Lemma~\ref{lem:ratio}: the condition $\lambda(w_{i,s+2}) = \lambda(w_{i,s+1})$ comes from Assumption~\ref{ass:persistence}(ii), and it holds automatically when $w_{it} = w_i$ is time-invariant. Because $K_2$ localizes the squared change $\|X_{i,s+2} - X_{i,s+1}\|^2$, the radius of the stability window in covariate units is $\sqrt{h_2}$, and with several covariates the squared norm adds changes measured in different units, so the covariates should be put on a common scale before $h_2$ is chosen. With integer-valued covariates, as in the simulations, $h_2 < 1$ selects exact stability.

\label{loc:bandwidthremark}The empirical codegree distances concentrate around their population values, whose scale is design dependent and can be numerically small (the smallest positive population distance in the baseline design of Section~\ref{sec:montecarlo} is $3.3 \times 10^{-3}$), so a bandwidth $h_1$ chosen as a fixed function of $n$ alone can exceed the largest distance in the sample, in which case every pair receives weight and no matching occurs. In sparse networks the same formula can leave too few pairs with positive weight. A scale-adaptive choice is therefore useful as a sensitivity diagnostic: set $h_1$ to a low quantile, say the fifth or tenth percentile, of the empirical distribution of $\{\hat{\bar{\delta}}_{ij}^2\}_{i<j}$. The quantile targets a share of the candidate distances, subject to ties. It does not guarantee a share of pairs that survive the discordance, shared-lag, and covariate-stability restrictions, which is why the diagnostics below should be reported with it. It is a diagnostic device rather than part of the theory: a fixed quantile of the estimated distances converges to a positive threshold in continuous designs, so it does not shrink, and in finite-type designs a quantile that falls inside the atom of exact matches is driven by estimation noise. Theorems~\ref{thm:consistency} and~\ref{thm:normality} are stated for deterministic bandwidth sequences, and the main simulations of Section~\ref{sec:montecarlo} use such sequences, calibrated with knowledge of the design that is not available in practice. In the sparse design of that section, the quantile rule selects mostly wrong matches.

\label{loc:h1rule}With $T = 3$, the following rule uses only the data, and in finite-type designs Theorem~\ref{thm:normality} and Proposition~\ref{prop:class} continue to hold with it (Corollary~\ref{cor:bandwidth}). Let $S_n$ be the largest $\hat{\delta}_{ij1}^2$ over the pairs $i < j \le \lceil \sqrt{n}\,\rceil$ with $y_{i0} = y_{j0}$, with $S_n = 0$ if there is none, and set
\begin{equation}\label{eq:h1rule}
\hat{h}_1 = 0.1\,\max\{0.001,\ S_n\}\,(n/500)^{-1/4}.
\end{equation}
In a finite-type design, $S_n$ converges in probability to the largest population distance between two types with the same initial outcome, so $\hat{h}_1$ scales with the distances of the design and shrinks at the rate $n^{-1/4}$. The floor keeps it positive. The rule uses neither the types nor the separation gap, and its constants are not optimized. It depends only on the first-wave network and the initial outcomes, which preserves the exactness of the class likelihood in Proposition~\ref{prop:class}(b), whereas a bandwidth chosen with later networks or with the orderings need not. Nothing ensures that $\hat{h}_1$ is below half the separation gap in a given sample, and Section~\ref{sec:mc_datadriven} shows when it is.

Three diagnostics should always be reported alongside the estimates: the total number of pairs receiving positive weight; the kernel-weighted effective sample size $(\sum_p w_p)^2/\sum_p w_p^2$ over the weighted pairs $p$, which is a measure of the concentration of the weights (it equals the number of weighted pairs when the weights are equal and is smaller otherwise) and not a count of independent agents, since pairs sharing an agent are dependent; and, over the set of eligible agents, those that switch between $s$ and $s+1$ and have covariates within the stability window, the number of agents that contribute at least one weighted pair, the number left with no partner, and the minimum number of weighted partners among the contributing agents, $\min_{i \in \mathcal{C}} \sum_{j \ne i} \ind\{\text{pair } (i,j) \text{ receives positive weight}\}$, where $\mathcal{C}$ is the set of contributing agents. The unmatched count says whether any eligible agent is excluded, and the minimum says how thinly the agents who do enter are matched. Agents who do not switch have no partners by construction and belong to neither count. Table~\ref{tab:mc_diag} reports all of these.

\subsection{Conditioning on Network Classes}\label{sec:class_estimator}

The pairwise estimator uses the network one pair at a time: each weighted pair contributes the conditional probability of Lemma~\ref{lem:identity}, whether or not its two agents also belong to a larger group of agents with the same network type. When types take finitely many values, the network certifies more than pairs. Among the agents who switch between $s$ and $s+1$ and have stable covariates, those who share the lagged outcome at $s-1$ and are at zero codegree distance from one another form a class whose members share the movement $d_s$, and, by the within-agent identity in the proof of Corollary~\ref{cor:lambda}, the ordering of each member's switch is a logit with the same intercept for the whole class. A logit with a group-specific intercept admits a sufficient statistic for that intercept, the number of members with a given outcome, and conditioning on it removes the intercept, as conditioning on the number of periods with a positive outcome removes the individual effect in the conditional logit of \citet{andersen1970} and \citet{chamberlain1980}. Here the class plays the role of the individual, and its members the role of the periods. The resulting likelihood is the conditional logistic regression for matched sets with several cases of \citet{breslow1980}, with network classes as the matched sets.

Formally, fix $s \in \mathcal{S}$ and let
\[
\mathcal{E}_s = \big\{j : y_{js} + y_{j,s+1} = 1,\ X_{j,s+2} = X_{j,s+1}\big\}
\]
be the set of eligible agents, who switch between $s$ and $s+1$ and have stable covariates between $s+1$ and $s+2$. With discrete covariates, this is the set that the covariate kernel selects at any bandwidth below $\bar\epsilon_X$ (Assumption~\ref{ass:discrete}(iv)). For $j \in \mathcal{E}_s$, let $b_{js} = \ind\{y_{js} = 1\}$ indicate the ordering $B_s$, and let $Z_{js}$ be the within-agent regressor defined before Corollary~\ref{cor:lambda}. The feasible classes are the connected components of the graph on $\mathcal{E}_s$ that joins $i$ and $j$ when
\begin{equation}\label{eq:class_graph}
y_{i,s-1} = y_{j,s-1} \quad \text{and} \quad \hat{\delta}_{ijs}^2 < h_1,
\end{equation}
with the network bandwidth of the pairwise estimator. The distance is that of period $s$: two members of a class may share the lagged outcome at $s+1$, and the networks of periods $s+1$ and $s+2$ depend on the outcomes whose ordering is modeled, so those periods cannot enter the construction. Other periods at which the two agents share the lagged outcome, before $s$ or after $s+2$, could be averaged in, as in~\eqref{eq:codegree_agg}, and with $T = 3$ there are none. For a class $C$ with $N_C = \sum_{j \in C} b_{js}$ members of ordering $B_s$, the class-conditional log likelihood is
\begin{equation}\label{eq:class_ll}
\ell_C(\theta) = \sum_{j \in C} b_{js}\, Z_{js}'\theta - \log \sum_{S \subseteq C:\ |S| = N_C} \exp\Big(\sum_{j \in S} Z_{js}'\theta\Big),
\end{equation}
the logarithm of the conditional probability of the observed orderings given $N_C$ when they are independent logits with a common intercept. A class in which every member has the same ordering contributes zero. The class-conditional estimator is
\begin{equation}\label{eq:class_estimator}
\hat\theta^{\mathrm{C}} = \argmax_{\theta \in \Theta}\ \sum_{s \in \mathcal{S}}\ \sum_{C} \ell_C(\theta),
\end{equation}
where the inner sum runs over the feasible classes at $s$. The sum inside the logarithm in~\eqref{eq:class_ll} is the elementary symmetric polynomial of order $N_C$ in the values $\exp(Z_{js}'\theta)$, $j \in C$. This polynomial and its first two derivatives are computed by a recursion over the members of the class \citep{gail1981}, as in standard software for conditional logistic regression, at a cost proportional to $|C|\,N_C$ rather than to the number of subsets. The objective is concave, since each $\ell_C$ is the log likelihood of an exponential family in $\theta$.

Three features distinguish the two estimators. First, for a class with one member of each ordering, \eqref{eq:class_ll} equals $\log F\big((Z_{js} - Z_{is})'\theta\big)$, with $j$ the member of ordering $B_s$, which is the pair's contribution to the pairwise objective~\eqref{eq:objective} at unit weight. For larger classes, the pairwise estimator replaces the conditional probability of the configuration of the whole class by the product of the conditional probabilities of its discordant pairs, which makes it a pairwise, or composite, version of the class likelihood, in the way that a conditional logit built from pairs of periods is a pairwise version of the conditional logit of \citet{chamberlain1980}. Pairs that share an agent carry overlapping information, whereas the class likelihood uses each class's sufficient statistic once, and Remark~\ref{rem:class_eff} argues that, in finite-type designs with $T = 3$, the class estimator is at least as efficient asymptotically. Second, the class construction is transitive and therefore more sensitive to an erroneous edge: a single pair of agents of different types whose estimated distance falls below the bandwidth merges their two classes, whereas in the pairwise estimator it affects one pair. A feasible class that splits a population class, by contrast, leaves the likelihood correctly specified and costs only efficiency (Proposition~\ref{prop:class}(b)). Third, the class estimator needs classes: exact covariate stability and a network whose types take finitely many values, separated by a gap that the bandwidth respects. In continuous designs the connected components of the thresholded graph chain together agents whose types, and movements, differ by much more than the bandwidth, so the construction does not deliver groups with a common movement, whereas the pairwise weights bound the distance of every compared pair. The pairwise estimator is therefore the one that the consistency theory of Section~\ref{sec:asymptotics} covers in general designs, and the class estimator is an alternative for finite-type designs, where it is more precise in the simulations of Section~\ref{sec:montecarlo}. In practice the diagnostics guide the choice: a visible gap in the estimated distances and a partition that is stable across bandwidths favor the class estimator, without certifying that its classes are correct, since a merge of two classes can survive a change of bandwidth. Otherwise, the pairwise estimator is the safer choice. Estimating the class intercepts together with $\theta$, by a logit with one dummy per class, is a third procedure. It is first-order equivalent to the class-conditional estimator in finite-type designs, as the proof of Proposition~\ref{prop:class} shows, but it estimates one nuisance parameter per class from that class's members, and its small-sample bias is larger (Table~\ref{tab:mc_joint}).

The diagnostics recommended for the pairwise estimator should be reported for the class estimator as well, together with the number of feasible classes, their sizes, the number of eligible agents in classes that contain both orderings, and the sensitivity of the partition to the bandwidth, since the merging of two classes changes the likelihood.

\section{Asymptotic Properties}\label{sec:asymptotics}

This section establishes the large-sample properties of the estimator $(\hat{\beta}, \hat{\alpha})$ defined in~\eqref{eq:estimator}. Throughout, $T$ is treated as fixed and asymptotics are studied as $n \to \infty$. The results are of two kinds. Consistency holds in general designs, including continuous network types and covariates, with shrinking bandwidths. Asymptotic normality at the root-$n$ rate, centered at the true value, is established for finite-type designs, in which network types take finitely many values and exact matches occur with positive probability. For continuous designs, the section explains what a distribution theory would require and why it is not established here. Section~\ref{sec:class_theory} gives the corresponding results for the class-conditional estimator of Section~\ref{sec:class_estimator} in finite-type designs, and Corollary~\ref{cor:bandwidth} extends Theorem~\ref{thm:normality} and Proposition~\ref{prop:class} to the data-driven bandwidth~\eqref{eq:h1rule}. Proofs are collected in Appendix~\ref{sec:appendix}.

\subsection{Sampling as the Sample Grows}

The estimation theory allows the idiosyncratic part of the movement of the latent index to be small but nonzero. Because the maximum over a growing sample of a variable with a fixed law cannot shrink to zero unless the variable is zero, this requires the law of the primitives to depend on $n$.

\begin{assumption}[Drifting idiosyncratic movement]\label{ass:drift}
For each $n$, the primitives are drawn under a law $\Prob_n$ (a triangular array) that satisfies Assumption~\ref{ass:sampling}, with the functions $\lambda$ and $f$, the distribution of the shocks, and the parameters $(\beta, \alpha)$ not depending on $n$. Under $\Prob_n$,
\[
\lambda(w_{i,t+1}) - \lambda(w_{it}) = d_t(w_{it}) + e_{it,n}, \qquad |e_{it,n}| \le \zeta_n \ \text{almost surely}, \quad t \le T-1,
\]
for a deterministic sequence $\zeta_n \to 0$, where the functions $d_t$ do not depend on $n$, satisfy the uniform-continuity requirement of Assumption~\ref{ass:persistence}(i), and satisfy $d_{s+1} \equiv 0$ for $s \in \mathcal{S}$. For each $t$, the distribution $G_{t,n}$ of $(w_{it}, y_{i,t-1})$ under $\Prob_n$ converges weakly to a limit $G_t$.
\end{assumption}

Assumption~\ref{ass:drift} is an asymptotic device of the drifting-parameter kind. It describes a latent index whose idiosyncratic movement between adjacent waves is small relative to the dispersion of the index across agents, and it reduces to Assumption~\ref{ass:persistence} under a fixed law when $\zeta_n \equiv 0$. Like the identification results, the estimation theory places no restriction on the agent effects, whose law under $\Prob_n$ may be arbitrary. The common component $d_t$ is unrestricted in size. Because the idiosyncratic component changes the law of the latent index, the population codegree distances also depend on $n$. The weak convergence of $G_{t,n}$ makes them converge uniformly, since $f$ is continuous and bounded (Assumption~\ref{ass:smooth} below). In the finite-type theorem of Section~\ref{sec:finitetype}, the law does not depend on $n$ and $e_{it,n} \equiv 0$.

\subsection{Kernels, Bandwidths, and Regularity}

Throughout this section, $\bar{\delta}_{ij}^{2}$ denotes the population aggregated codegree distance, obtained from~\eqref{eq:codegree_agg} by replacing each estimated per-period distance $\hat{\delta}_{ijt}^2$ with its population counterpart $\delta_{ijt}^{2}$ from Lemma~\ref{lem:codegree_consist}. The hatted symbol $\hat{\bar{\delta}}_{ij}^2$ is the sample statistic. Under Assumption~\ref{ass:drift}, the population distances are computed under $\Prob_n$. Let $\ind_{ijs}$ be the indicator that the unordered pair $\{i,j\}$ belongs to $\mathcal{P}_s$ (it collects the shared lagged outcome at $s-1$, both switch conditions, and discordance), write $K_2^{(a)} = K_2(\|X_{a,s+2}-X_{a,s+1}\|^2/h_2)$ for $a \in \{i,j\}$, and define the kernel weight at the population distances, its enlarged-window counterpart, and the expected weight mass,
\[
W_{ijs} = K_1\!\left(\frac{\bar{\delta}_{ij}^{2}}{h_1}\right)K_2^{(i)}K_2^{(j)}, \qquad
\bar{W}_{ijs} = \ind\{\bar{\delta}_{ij}^{2} < 2h_1\}\,K_2^{(i)}K_2^{(j)}, \qquad
\omega_n = \sum_{s \in \mathcal{S}} \E_n\big[\ind_{ijs} W_{ijs}\big],
\]
where $\E_n$ is the expectation under $\Prob_n$ and $(i,j)$ is a generic pair of distinct agents. The quantity $n\,\omega_n$ is, up to the normalization of the kernels, the expected number of effective partners of an agent.

\begin{assumption}[Kernel and bandwidth conditions]\label{ass:kernel}
\begin{enumerate}[label=(\roman*)]
\item The kernel functions $K_1$ and $K_2$ are non-negative, have support $[0,1)$, satisfy $K_\ell(0) > 0$, and are Lipschitz on $[0,\infty)$.
\item $h_1 \to 0$, $h_2 \to 0$, $n h_1^2/\log n \to \infty$, and $n\,\omega_n \to \infty$.
\end{enumerate}
\end{assumption}

The Lipschitz condition on $[0,\infty)$ requires the kernels to vanish continuously at the boundary of their support. The Epanechnikov-type kernel $K(x) = \tfrac{3}{4}(1 - x^2)\ind(0 \le x < 1)$ used in Sections~\ref{sec:montecarlo} and~\ref{sec:empirical} satisfies it. No differentiability is required, because the theory never expands the kernel weights beyond first order. The condition $n h_1^2/\log n \to \infty$ makes the uniform estimation error of the codegree distances, of order $(\log n/n)^{1/2}$ by Lemma~\ref{lem:codegree_consist}(d), small relative to the network bandwidth. The condition $n\,\omega_n \to \infty$ requires the expected number of effective partners to grow. It combines the network match probability and the covariate-stability probability, and it reflects the cost of covariate-stability smoothing familiar from \citet{honore2000panel}: when covariates are continuously distributed, the stability probability vanishes with $h_2$. Because the codegree profile is an unknown function, the network match probability cannot be approximated by a polynomial in $h_1$ of known order, so the condition is stated directly. In practice it is monitored through the diagnostics recommended in Section~\ref{sec:estimation}, in particular the number of pairs with positive weight and the per-agent minimum match count.

Define the localized criterion and the localized information matrix of the actual selected population,
\begin{equation}\label{eq:pseudo}
\begin{aligned}
M_n(b,a) &= \frac{1}{\omega_n}\sum_{s \in \mathcal{S}} \E_n\big[\ind_{ijs} W_{ijs}\, m_{ijs}(b,a)\big],\\
\Sigma_n &= \frac{1}{\omega_n}\sum_{s \in \mathcal{S}} \E_n\big[\ind_{ijs} W_{ijs}\, F(Z_{ijs}'\theta)\big(1 - F(Z_{ijs}'\theta)\big)\, Z_{ijs}Z_{ijs}'\big].
\end{aligned}
\end{equation}

\begin{assumption}[Regularity conditions]\label{ass:regularity}
\begin{enumerate}[label=(\roman*)]
\item $\E_n[\|X_{it}\|^4]$ is bounded uniformly in $n$ for all $t$.
\item The parameter space $\Theta$ is compact and convex and contains $\theta = (\beta', \alpha)'$ in its interior.
\item (Moments in the matching windows.) $\sup_n \omega_n^{-1} \sum_{s \in \mathcal{S}} \E_n\big[\ind_{ijs} \bar{W}_{ijs}(1 + \|Z_{ijs}\|)^4\big] < \infty$.
\item (Rank on the selected population.) $\liminf_{n} \lambda_{\min}(\Sigma_n) > 0$, where $\lambda_{\min}$ denotes the smallest eigenvalue.
\end{enumerate}
\end{assumption}

Part (iii) bounds the moments of the regressor within the matching windows, and, because $\bar{W}_{ijs}$ uses the window of twice the network bandwidth, it also requires the mass of pairs just outside the window to be comparable to the mass inside it. An unconditional moment condition does not deliver these requirements, because the windows shrink. Part (iv) is stated for the population that the estimator actually selects: the shared lagged outcome at $s-1$, discordance, the covariate localization, and the network localization are all built into $\ind_{ijs}W_{ijs}$. A rank condition imposed on a larger population need not survive the additional restrictions: if $X_{i,s} - X_{i,s+1}$ were a function of $y_{i,s-1}$ alone, the covariate component of $Z_{ijs}$ would vanish on every selected pair, since selected pairs share $y_{i,s-1}$, and the objective would be flat in $\beta$ at every sample size, even though the second moment of $Z_{ijs}$ could be positive definite on the larger population. The assumption does not require the localized criterion $M_n$ to converge, nor does it equate its behavior with separately defined conditional expectations, which could differ across relative bandwidths at a probability-zero event. The proof shows that the correctly specified comparison criterion $M_n^*$, which is $M_n$ with the index perturbation of the inexact matches removed, is maximized at the truth with a curvature bounded away from zero uniformly in $n$ and that $M_n$ is uniformly close to $M_n^*$ as the bandwidths and the drift vanish. The argument works along subsequences. At fixed bandwidths $M_n$ itself is generally maximized at a pseudo-true value that differs from $\theta$ (Remark~\ref{rem:pseudotrue}).

\begin{assumption}[Separation of network types]\label{ass:smooth}
The link formation function $f$ is continuous on $[0,1]^2 \times \{0,1\}^2$ and not constant, and for each $t$ the distribution $G_t$ of $(w_{it}, y_{i,t-1})$ (under Assumption~\ref{ass:drift}, its limit) has full support on $[0,1] \times \{0,1\}$.
\end{assumption}

Continuity makes the population codegree distance a continuous function of the latent indices, which the kernel localization argument uses, and non-constancy rules out a link function that is the same for every agent. Non-constancy does not by itself guarantee that the network is informative about the latent index (a link function that depends only on the observed lags is not constant but carries no information about $w_{it}$), which is the role of the full-support condition and of Lemma~\ref{lem:codegree_consist}(c), while identifying variation in the outcomes is the role of the rank condition. The full-support condition is what makes the codegree distance separate agents that are not network-type equivalent. The codegree compares two agents through their links to third agents, and it can only see differences in link functions at partner types that are present in the population. For a continuous link function, the separation is then automatic on the support: writing $g(\tau, y') = f(u, \tau;\, y, y') - f(v, \tau;\, y, y')$ and $(\mathcal{K} g)(\cdot) = \int f(\cdot, \tau;\, \cdot, y')\, g(\tau, y')\, dG_t(\tau, y')$ for the codegree transform, one has the identity
\[
\int g^2\, dG_t = (\mathcal{K} g)(u, y) - (\mathcal{K} g)(v, y),
\]
so if the codegree distance vanishes, which means that $\mathcal{K} g = 0$ almost everywhere and hence, by continuity, everywhere on the support of $G_t$, then $g = 0$ on that support (Lemma~\ref{lem:codegree_consist}(c)). Full support converts agreement on the support into agreement of the link functions in the metric~\eqref{eq:netdist} used in Assumption~\ref{ass:persistence}(i), which integrates over all partner characteristics and both lag strata. Without full support, two agents whose link functions differ only at partner types or lag strata that are absent from the population cannot be distinguished by the codegree, and the relevant equivalence is agreement on the support. Finite-type designs replace this assumption by the directly verifiable separation condition of Assumption~\ref{ass:discrete}, which does not require continuity of $f$ or full support.

\subsection{Consistency}

\begin{theorem}[Consistency]\label{thm:consistency}
Under Assumptions~\ref{ass:drift}--\ref{ass:smooth}, the estimator $(\hat{\beta}, \hat{\alpha})$ defined in~\eqref{eq:estimator} is consistent: $(\hat{\beta}, \hat{\alpha}) \plim (\beta, \alpha)$.
\end{theorem}

Assumption~\ref{ass:drift} contains Assumption~\ref{ass:sampling} for each $n$ and Assumption~\ref{ass:persistence} up to the vanishing idiosyncratic component. The proof, given in Appendix~\ref{sec:appendix}, exploits the concavity of the logistic log-likelihood. Each contribution $m_{ijs}(b,a)$ is concave in $(b,a)$ and the kernel weights are non-negative, so the sample objective is a random concave function on the convex set $\Theta$, and the convexity lemma for extremum estimators \citep[Theorem 2.7]{newey1994}, applied along subsequences, reduces the argument to pointwise convergence of the normalized objective to correctly specified criteria $M_n^*$ that are maximized at $(\beta, \alpha)$ with non-vanishing curvature. The localized criterion $M_n$ that the estimator targets is uniformly close to $M_n^*$ as the bandwidths shrink. Three steps establish this. First, the estimated distances can be replaced by the population distances, because their uniform estimation error is small relative to the bandwidth. Second, the normalized objective at the population distances concentrates around $M_n(b,a)$, by a variance bound for a second-order U-statistic in the i.i.d.\ agent primitives, controlled by $n\omega_n \to \infty$. Third, on the pairs with positive weight, the true conditional probability of the discordant configuration is a logit in $Z_{ijs}'\theta$ up to an index perturbation bounded by three terms: the difference between the two agents' common movements $d_s$, which vanishes with $h_1$ by Assumptions~\ref{ass:persistence}(i) and~\ref{ass:smooth}; the covariate movement inside the stability window, of order $h_2^{1/2}$; and the idiosyncratic movement, of order $\zeta_n$. The localized criterion is therefore asymptotically that of a correctly specified logit on the selected population, whose information matrix is bounded away from singularity, so it is maximized at the truth with a curvature that does not vanish as $n$ grows.

\begin{remark}[Fixed bandwidths and the pseudo-true value]\label{rem:pseudotrue}
If $(h_1, h_2)$ are held fixed as $n \to \infty$ under a fixed law, the same argument shows that $\hat{\theta}$ converges in probability to $\theta_h = \argmax_{(b,a) \in \Theta} M_h(b,a)$, the maximizer of the localized criterion~\eqref{eq:pseudo} at those bandwidths, provided that maximizer is unique. The pseudo-true value $\theta_h$ generally differs from $\theta$ in continuous designs, because matches at a positive bandwidth are inexact, and it converges to $\theta$ as the bandwidths shrink under the conditions above. Inexact matches do not always cause bias: when $w_{it} = w_i$ is time-invariant and the selected agents' covariates are exactly stable, the remainder in the proof of Lemma~\ref{lem:ratio} and the common movement vanish for every pair, whatever their types, so $\theta_h = \theta$ at any fixed network bandwidth for which the selected rank condition holds. Fixed-quantile implementations in continuous designs estimate $\theta_h$ at the corresponding limiting threshold.
\end{remark}

\subsection{Asymptotic Normality in Finite-Type Designs}\label{sec:finitetype}

The distribution theory covers designs in which exact matches occur with positive probability. Network types take finitely many values, as in stochastic block specifications with one latent value per block, and the covariates have mass points at exact stability. Survey data with covariates recorded on discrete scales fit this case naturally.

Write $\delta_t(u, v;\, y)$ for the population codegree distance at period $t$ between agents with latent indices $u$ and $v$ and common lagged outcome $y$, and say that two supported indices $u, v$ are \emph{equivalent on the support} at period $t$ in stratum $y$ if $f(u, \tau;\, y, y') = f(v, \tau;\, y, y')$ for every $(\tau, y')$ in the support of $G_t$.

\begin{assumption}[Finite-type design]\label{ass:discrete}
\begin{enumerate}[label=(\roman*)]
\item For each $t$, $w_{it}$ takes values in a finite set $\mathcal{W}_t \subset [0,1]$, each value with positive probability, so that the latent path $w_i$ has finite support.
\item (Equivalence on the support.) For each $t$ and $y$, if two supported indices $u, v \in \mathcal{W}_t$ are equivalent on the support at period $t$ in stratum $y$, then $d_t(u) = d_t(v)$.
\item (Exact dynamics.) The law of the primitives does not depend on $n$, so that Assumption~\ref{ass:persistence} holds exactly and there is no idiosyncratic movement.
\item For each $s \in \mathcal{S}$, $\Prob(X_{i,s+2} = X_{i,s+1}) > 0$, and there is $\bar\epsilon_X > 0$ with $\Prob(0 < \|X_{i,s+2} - X_{i,s+1}\|^2 < \bar\epsilon_X) = 0$.
\item (Bandwidths.) $h_2 \le \bar\epsilon_X$ is fixed; $h_1 \le \bar\epsilon/(2(T-2))$ for all $n$ large, where $\bar\epsilon$ is the smallest positive value of $\delta_t(u, v;\, y)^2$ over $t$, $y$, and supported $u, v$; and $\sqrt{n}\,h_1/\log n \to \infty$.
\item (Rank.) With $\ind^{\mathrm{ex}}_{ijs} = \ind_{ijs}\ind\{\bar{\delta}_{ij} = 0,\ X_{i,s+2} = X_{i,s+1},\ X_{j,s+2} = X_{j,s+1}\}$, the matrix
\[
\Sigma_{\mathrm{ex}} = \sum_{s \in \mathcal{S}} \E\big[\ind^{\mathrm{ex}}_{ijs}\, F(Z_{ijs}'\theta)\big(1 - F(Z_{ijs}'\theta)\big)\, Z_{ijs}Z_{ijs}'\big]
\]
is positive definite.
\end{enumerate}
\end{assumption}

Because the bandwidths in part (v) select exactly the pairs with $\ind^{\mathrm{ex}}_{ijs} = 1$, the matrix $\Sigma_n$ of Assumption~\ref{ass:regularity} equals $\Sigma_{\mathrm{ex}}$ divided by $\sum_{s}\Prob(\ind^{\mathrm{ex}}_{ijs} = 1)$ in this design, so part (vi) is the finite-type form of the rank condition in Assumption~\ref{ass:regularity}(iv). Part (ii) replaces the continuity requirement in Assumption~\ref{ass:persistence}(i) and Assumption~\ref{ass:smooth}. It is stated on the support because the codegree can only compare agents through partners who are present in the population. By the identity displayed after Assumption~\ref{ass:smooth}, which requires no continuity when the support is finite, the population codegree distance between two supported indices vanishes exactly when they are equivalent on the support (Lemma~\ref{lem:codegree_consist}(c)). Since there are finitely many supported indices, the positive distances are bounded below by $\bar\epsilon > 0$, and this separation gap can be computed from the block matrix of link probabilities and the type probabilities. No continuity of $f$ across block boundaries is needed, and only the common movement $d_t$, not $\lambda$ itself, needs to be constant on classes of supported types. Because periods $s+1$ and $s+2$ never enter the aggregated distance of a discordant pair, that distance averages at most $T - 2$ per-period distances, so a pair that is not equivalent at some contributing period has $\bar{\delta}_{ij}^2 \ge \bar\epsilon/(T-2)$, which is why part (v) bounds $h_1$ by $\bar\epsilon/(2(T-2))$. With $T = 3$, the bound is $\bar\epsilon/2$. Parts (ii) and (iii) together make the exact-match event reproduce the conditions of Lemma~\ref{lem:identity}: for a pair that shares the lagged outcome at $s-1$ and has $\bar{\delta}_{ij} = 0$, the two agents are equivalent on the support at $s$, hence have equal changes of the social influence term between $s$ and $s+1$ by part (ii), and both are stable between $s+1$ and $s+2$ by Assumption~\ref{ass:persistence}(ii). Part (v) allows a fixed network bandwidth below the stated fraction of the separation gap, or one that shrinks slowly. Unlike an empirical quantile, such a bandwidth does not fall inside the atom of exact matches.

Let $\mathcal{Z}_i = (y_{i0}, \ldots, y_{iT}, X_i, w_i, \mu_i)$ denote agent $i$'s complete node data, including the agent effect, and let
\[
S_{ijs}(\theta) = \big\{\ind(y_{is} = 1) - F(Z_{ijs}'\theta)\big\} Z_{ijs}
\]
denote the score of $m_{ijs}$ on the pairs in $\mathcal{P}_s$. The score is symmetric in $(i,j)$.

\begin{theorem}[Asymptotic normality, finite-type designs]\label{thm:normality}
Under Assumption~\ref{ass:sampling}, Assumption~\ref{ass:persistence} with the continuity requirement in part (i) replaced by Assumption~\ref{ass:discrete}(ii), and Assumptions~\ref{ass:kernel}(i), \ref{ass:regularity}(i)--(ii), and~\ref{ass:discrete},
\begin{equation*}
\sqrt{n}\begin{pmatrix} \hat{\beta} - \beta \\ \hat{\alpha} - \alpha \end{pmatrix} \dlim \N\left(0, \; 4\Sigma_{\mathrm{ex}}^{-1} V_{\mathrm{ex}} \Sigma_{\mathrm{ex}}^{-1}\right),
\qquad
V_{\mathrm{ex}} = \operatorname{Var}(\psi_i), \qquad
\psi_i = \E\Big[\sum_{s \in \mathcal{S}} \ind^{\mathrm{ex}}_{ijs}\, S_{ijs}(\theta) \,\Big|\, \mathcal{Z}_i\Big],
\end{equation*}
where the expectation defining $\psi_i$ is over an independent agent $j$.
\end{theorem}

Three features of the finite-type regime deliver the result, and all three are established in the proof. First, with probability approaching one, the estimated weights are positive exactly on the pairs of the exact-match event, and on those pairs they differ from $K_1(0)K_2(0)^2$ by at most $O_p(\log n/(n h_1))$ uniformly, which is $o_p(n^{-1/2})$ under Assumption~\ref{ass:discrete}(v). The estimation of the codegree distances therefore does not affect the limiting distribution. Second, the estimator is thereby asymptotically equivalent to an oracle conditional logit whose objective is a second-order U-process in the i.i.d.\ node data $\mathcal{Z}_i$. The link shocks enter only through the estimated weights and drop out. Third, by Lemma~\ref{lem:identity} and Assumption~\ref{ass:discrete}(ii)--(iii), the exact-match score has mean zero given the conditioning information of that lemma, so the expected oracle criterion is maximized at the truth at every sample size and no bias term appears. The factor 4 in the variance is the familiar one for second-order U-statistics, and $\psi_i$ sums the contributions of all transition periods before the variance is taken, so $V_{\mathrm{ex}}$ retains cross-period covariances.

The asymptotic variance is estimated by plug-in. Let $\hat{W}_{ijs}$ denote the estimated weight of pair $(i,j)$ at $s$, $\bar{K} = K_1(0)K_2(0)^2$, and
\begin{align*}
\hat{\Sigma} &= \binom{n}{2}^{-1}\sum_{s \in \mathcal{S}}\sum_{i<j}\frac{\hat{W}_{ijs}}{\bar{K}}\,\ind_{ijs}\, F(Z_{ijs}'\hat{\theta})\big(1 - F(Z_{ijs}'\hat{\theta})\big) Z_{ijs}Z_{ijs}',\\
\hat{\psi}_i &= \frac{1}{n-1}\sum_{j \neq i}\sum_{s \in \mathcal{S}}\frac{\hat{W}_{ijs}}{\bar{K}}\,\ind_{ijs}\, S_{ijs}(\hat{\theta}),
\end{align*}
with $\hat{V}$ the sample variance of $\hat{\psi}_i$. Then $4\hat{\Sigma}^{-1}\hat{V}\hat{\Sigma}^{-1}/n$ estimates the variance of $\hat{\theta}$, and consistency of $\hat{\Sigma}$ and $\hat{V}$ for $\Sigma_{\mathrm{ex}}$ and $V_{\mathrm{ex}}$ follows from the oracle equivalence established in the proof of Theorem~\ref{thm:normality} and the law of large numbers for U-statistics. This estimator coincides, up to finite-sample normalization, with the sandwich that clusters the pair-level scores by agent. When observations are grouped, the agent-level terms $\hat{\psi}_i$ can be summed within groups before forming the variance, yielding a cluster-robust version. That version is not covered by Theorem~\ref{thm:normality}, which assumes independent agents.

\subsection{The Class-Conditional Estimator in Finite-Type Designs}\label{sec:class_theory}

The class-conditional estimator is analyzed in the finite-type regime of Assumption~\ref{ass:discrete}, with a single transition period, $\mathcal{S} = \{s\}$, as in the designs of Section~\ref{sec:montecarlo}. The case of several transition periods is discussed after Remark~\ref{rem:class_eff}. Equivalence on the support at period $s$ in stratum $y$ is an equivalence relation on $\mathcal{W}_s$. The population class of an eligible agent $j$ is $c_j = (y_{j,s-1}, [w_{js}])$, where $[w_{js}]$ is the set of supported indices equivalent to $w_{js}$ on the support in stratum $y_{j,s-1}$. There are finitely many classes, and, by Assumption~\ref{ass:discrete}(ii), $d_s$ takes a single value $d_c$ on each class $c$. Let $e_j = \ind\{j \in \mathcal{E}_s\}$, $p_j = F(Z_{js}'\theta - d_{c_j})$, and $q_j = p_j(1 - p_j)$, and define, over a single agent,
\begin{gather*}
\bar Z_c = \frac{\E\big[e_j \ind\{c_j = c\}\, q_j Z_{js}\big]}{\E\big[e_j \ind\{c_j = c\}\, q_j\big]}, \qquad
\tilde S_j = e_j\,(b_{js} - p_j)\,(Z_{js} - \bar Z_{c_j}),\\
\mathcal{I}_{\mathrm{C}} = \E\big[e_j\, q_j\,(Z_{js} - \bar Z_{c_j})(Z_{js} - \bar Z_{c_j})'\big],
\end{gather*}
with $\bar Z_c$ defined for the classes that eligible agents occupy with positive probability. The matrix $\mathcal{I}_{\mathrm{C}}$ is the information about $\theta$ in a logit for the eligible agents' orderings with one free intercept per class, after the intercepts are profiled out, and it is positive definite when the regressor varies within classes.

\begin{proposition}[Class-conditional estimator]\label{prop:class}
Suppose that Assumption~\ref{ass:sampling}, Assumption~\ref{ass:persistence} with the continuity requirement in part (i) replaced by Assumption~\ref{ass:discrete}(ii), and Assumptions~\ref{ass:regularity}(i)--(ii) and~\ref{ass:discrete}(i)--(v) hold, with $\mathcal{S} = \{s\}$, and that $\mathcal{I}_{\mathrm{C}}$ is positive definite.
\begin{enumerate}[label=(\alph*)]
\item (Class recovery.) With probability approaching one, the feasible classes at $s$ are exactly the sets $\{j \in \mathcal{E}_s : c_j = c\}$.
\item (Exact conditional likelihood.) Let $\mathcal{G}_s$ be the $\sigma$-field generated by the primitives $\{\chi_j\}_{j \le n}$, the outcomes $\{y_{j,-s}\}_{j \le n}$ outside periods $s$ and $s+1$, the networks $D_t$ with $t \notin \{s+1, s+2\}$, and the switch events $\{\mathcal{T}_{js}\}_{j \le n}$. Given $\mathcal{G}_s$, the orderings $b_{js}$, $j \in \mathcal{E}_s$, are independent, with $\Prob(b_{js} = 1 \mid \mathcal{G}_s) = F(Z_{js}'\theta - d_s(w_{js}))$. Consequently, for every $\mathcal{G}_s$-measurable set $C \subseteq \mathcal{E}_s$ whose members have the same movement $d_s(w_{js})$, $\exp\{\ell_C(\theta)\}$ is the conditional probability of the observed orderings in $C$ given $\mathcal{G}_s$ and $N_C$. This holds, at every sample size, for every feasible class that does not merge agents with different movements, and in particular for every feasible class on the event of part (a).
\item (Asymptotic distribution.) $\sqrt{n}\,(\hat\theta^{\mathrm{C}} - \theta) = \mathcal{I}_{\mathrm{C}}^{-1}\, n^{-1/2}\sum_{j=1}^n \tilde S_j + o_p(1) \dlim \N\big(0, \mathcal{I}_{\mathrm{C}}^{-1}\big)$, and $n^{-1}$ times the negative Hessian of the objective in~\eqref{eq:class_estimator} at $\hat\theta^{\mathrm{C}}$ converges in probability to $\mathcal{I}_{\mathrm{C}}$.
\end{enumerate}
\end{proposition}

Part (a) is the class-recovery step: the per-period codegree distance is estimated uniformly at the rate $(\log n/n)^{1/2}$ (Proposition~\ref{prop:bridge}(a)), and its positive population values are at least $\bar\epsilon \ge 2h_1$, so the thresholded graph eventually joins exactly the pairs of eligible agents in the same population class, whose connected components are the classes because equivalence is transitive. Part (b) makes precise that the class likelihood is exact rather than asymptotic: it requires only that no feasible class merges agents whose movements differ, and splitting a population class is harmless. Part (c) combines the two. Its proof in Appendix~\ref{sec:proof_class} compares the class-conditional log likelihood with the profile log likelihood of the logit with one intercept per class. The two differ by the logarithm of the probability that a class has its observed number of switches of ordering $B_s$, which lies between $-\log(|C| + 1)$ and zero, because that number is a mode of its distribution when the intercept is profiled out \citep{darroch1964}, and whose gradient is bounded in probability, by the approximation of \citet{hajek1964} for the moments of independent Bernoulli variables conditioned on their sum (Lemma~\ref{lem:rejective}). The class estimator is therefore asymptotically equivalent to maximum likelihood with free class intercepts, which is why the profiled information $\mathcal{I}_{\mathrm{C}}$ appears. With a single transition period each eligible agent enters one class, so the inverse Hessian is a valid variance estimate, and no clustering is needed.

\begin{remark}[Efficiency relative to the pairwise estimator when $T = 3$]\label{rem:class_eff}
Suppose $T = 3$, so that $\mathcal{S} = \{1\}$, and the conditions of Proposition~\ref{prop:class} and Theorem~\ref{thm:normality} hold. Then $4\Sigma_{\mathrm{ex}}^{-1}V_{\mathrm{ex}}\Sigma_{\mathrm{ex}}^{-1} - \mathcal{I}_{\mathrm{C}}^{-1}$ is positive semidefinite, so the class estimator is at least as efficient asymptotically as the pairwise estimator. Both estimators use only the eligible agents' orderings, regressors, and population classes, which they recover with probability approaching one, and both are consistent in the model in which, given the regressors and classes, the orderings are independent logits with one intercept per class, with the intercepts and the law of the regressors unrestricted. In that model $\mathcal{I}_{\mathrm{C}}$ is the efficient information for $\theta$, and $\mathcal{I}_{\mathrm{C}}^{-1}\tilde S_j$ is the efficient influence function, which the class estimator attains by Proposition~\ref{prop:class}(c). The pairwise estimator is asymptotically linear with influence function $2\Sigma_{\mathrm{ex}}^{-1}\psi_j$, and its estimating equation has mean zero at every law in that model, so its influence function is a gradient and the estimator is regular. By the convolution theorem \citep[Chapter~25]{vandervaart1998}, its asymptotic variance is at least the efficient one. Equality holds only when the two influence functions coincide, and the simulations of Section~\ref{sec:montecarlo} measure the difference in finite samples. The comparison concerns this reduced model of the orderings, not the full model of outcomes and links, and no efficiency claim is made for $T > 3$, where the pairwise weights also use pair-specific distances at other periods. The argument is given in outline rather than as a formal proof.
\end{remark}

With several transition periods, the class log likelihoods of all periods can be summed in~\eqref{eq:class_estimator}, and an agent who switches in several windows enters a class in each. Proposition~\ref{prop:class} does not cover this case, which would require recovering the classes in all windows jointly and an expansion that accounts for the dependence of an agent's orderings across windows. A natural variance estimator is $\mathcal{I}_{\mathrm{C}}^{-1}V_{\mathrm{C}}\mathcal{I}_{\mathrm{C}}^{-1}$, with $\mathcal{I}_{\mathrm{C}}$ summed over periods and $V_{\mathrm{C}}$ estimated by clustering by agent the terms $(b_{js} - \hat p_{js})(Z_{js} - \hat{\bar Z}_{C})$, where $\hat p_{js} = F(Z_{js}'\hat\theta^{\mathrm{C}} - \hat\gamma_C)$, $\hat\gamma_C$ solves $\sum_{j \in C} F(Z_{js}'\hat\theta^{\mathrm{C}} - \gamma) = N_C$, and $\hat{\bar Z}_C$ is the $\hat p(1 - \hat p)$-weighted mean of $Z_{js}$ in the class, but its validity is not established here. The inverse Hessian alone would not account for that dependence.

\begin{corollary}[Data-driven network bandwidth]\label{cor:bandwidth}
Let $T = 3$, so that $\mathcal{S} = \{1\}$, and suppose that the conditions of Theorem~\ref{thm:normality} hold, except that in Assumption~\ref{ass:discrete}(v) the network bandwidth is a random variable $\hat{h}_1$ with $c_- n^{-1/4} \le \hat{h}_1 \le c_+ n^{-1/4}$ almost surely, for constants $0 < c_- \le c_+ < \infty$. Then the conclusion of Theorem~\ref{thm:normality} and the consistency of its variance estimator hold with $h_1 = \hat{h}_1$, and so, if $\mathcal{I}_{\mathrm{C}}$ is positive definite, do parts (a) and (c) of Proposition~\ref{prop:class}. If $\hat{h}_1$ is also a function of $D_1$ and $(y_{10}, \ldots, y_{n0})$, part (b) holds as well. The rule~\eqref{eq:h1rule} satisfies these conditions, with $c_- = 10^{-4} \cdot 500^{1/4}$ and $c_+ = 0.1 \cdot 500^{1/4}$.
\end{corollary}

The proof (Appendix~\ref{sec:proof_bandwidth}) uses only the envelope: the bounds under which the estimated weights select the exact matches hold uniformly over bandwidths between $c_- n^{-1/4}$ and $c_+ n^{-1/4}$, so they hold at $\hat{h}_1$ whatever its dependence on the data. Measurability matters only for the exactness of the class likelihood at every sample size: $D_1$ and the initial outcomes are $\mathcal{G}_1$-measurable, so the feasible classes remain so. The envelope is loose for the rule~\eqref{eq:h1rule}, whose $S_n$ is about 0.007 in the baseline design of Section~\ref{sec:montecarlo}, far below its bound of one.

\subsection{Continuous Designs}\label{sec:continuous}

When network types or covariates vary continuously, matches are inexact at every positive bandwidth. Theorem~\ref{thm:consistency} establishes consistency with shrinking bandwidths, and Remark~\ref{rem:pseudotrue} describes the limit at fixed bandwidths. The paper does not establish a limiting distribution in these designs, and I explain why, because the reasons also determine how the simulation results for continuous designs should be read.

First, the estimated weights of different pairs share network rows: the link $D_{ij}$ enters both agents' neighborhoods, and each neighborhood depends on overlapping sets of third agents. Per-agent influence terms built from an agent's realized links are therefore not independent across agents, and a central limit theorem requires an expansion in the independent node primitives and dyadic shocks, with the dependence of each component controlled.

Second, the estimation error of the codegree distance has two sources: link noise given the latent types, and sampling variation in the distribution of the auxiliary agents' latent types. The second does not vanish when link noise is averaged out. For example, with $f(u,v) = uv$ and $w$ uniform, the population distance is $\delta_{ij}^2 = (w_i - w_j)^2 m_2^3$ with $m_2 = \E[w^2]$, whereas replacing links by their probabilities in the statistic yields, up to repeated-index terms, $(w_i - w_j)^2 M_{2n}^3$ with $M_{2n} = n^{-1}\sum_l w_l^2$, whose fluctuation includes $3(w_i - w_j)^2 m_2^2(M_{2n} - m_2)$, of order $n^{-1/2}$. A first-stage correction must account for this term.

Third, the variance normalization depends on how the probability of being matched varies across agents. A localized estimator's leading variance contains a factor of the form $\E[r_h(w_i)^2]/(\E[r_h(w_i)])^2$, where $r_h(w)$ is an agent's conditional match mass, and this factor can diverge as $h_1 \to 0$ \citep[discussion of Proposition C2]{auerbach2021}. For example, with $T = 3$, $w_{it} = w_i$ uniform, $\lambda = 0$, $\mu_i = 0$, $\beta = \alpha = 0$, exactly stable covariates, and the smooth, separating link function $f(u,v) = \tfrac14 + \tfrac14(u^4 + v^4)$, the squared codegree distance is proportional to $(u^4 - v^4)^2$, the match mass behaves like $b^{1/4}$ for $w \lesssim b^{1/4}$ and like $b/w^3$ otherwise with $b = h_1^{1/2}$, and the factor is of order $h_1^{-1/8}$. The leading variance of the oracle estimator is then of order $1/(n h_1^{1/8})$, not $1/n$. There is therefore no universal convergence rate in continuous designs, and a valid theory must use a variance normalization that accommodates heterogeneous match masses, possibly after equalizing agents' match counts as \citet{auerbach2021} suggests.

Finally, consider bias correction. If the pseudo-true value admitted an expansion
\[
\theta_h = \theta + \sum_{l=1}^{L_1} C^{(1)}_l\, h_1^{l/\vartheta_1} + \sum_{l=1}^{L_2} C^{(2)}_l\, h_2^{l/\vartheta_2} + R(h_1, h_2)
\]
with known exponents, the combination
\begin{equation}\label{eq:jackknife}
\bar{\theta} = \sum_{l=1}^{M+1} a_l\, \hat{\theta}_{c_l}, \qquad \sum_l a_l = 1, \qquad \sum_l a_l\, c_l^{e_m} = 0 \ \text{ for } m = 1, \ldots, M,
\end{equation}
of estimators computed at bandwidths $(c_l h_1, c_l h_2)$, where $\{e_1, \ldots, e_M\}$ is the merged set of distinct exponents, would remove every listed term: the weights solve a generalized Vandermonde system and are unique \citep{honore1997}. The exponents are not derived here for any continuous design, symmetry of the kernel window can remove first-order terms, and inference based on $\bar{\theta}$ would in any case require the joint central limit theorem that the three preceding points show to be missing. The generalized jackknife is therefore recorded as an algebraic device, not as an inference result.

For these reasons, in continuous designs the paper reports point estimates together with the finite-type standard errors for comparison only. Confidence intervals based on them are not justified by the theory, and the simulations of Section~\ref{sec:montecarlo} report their coverage as a diagnostic.

\section{Monte Carlo Simulations}\label{sec:montecarlo}

This section studies the finite-sample behavior of the proposed estimator in designs that satisfy the assumptions of the theorems and in designs that deliberately violate them. The simulations address five questions. First, in the finite-type regime of Theorem~\ref{thm:normality}, with deterministic bandwidths, is the estimator centered at the truth, do the estimated standard errors match the sampling variability, and do the intervals have their nominal coverage? Second, what does estimating the matches cost relative to an oracle that uses the true network types? Third, how does the estimator compare with within-agent conditioning \citep{honore2000panel}? Fourth, where does the method lose ground or fail: when the heterogeneity is time invariant, so that within-agent conditioning is valid and more precise, with movement that is not common to network types, with sparse or weakly informative networks, and with covariates in link formation? Fifth, in the finite-type designs, how much does the class-conditional estimator of Section~\ref{sec:class_estimator}, which uses the same comparisons class by class rather than pair by pair, gain over the pairwise estimator, and does the gain come from the network classes or only from allowing a common movement? A continuous design, in which only consistency is established, and the data-driven bandwidth~\eqref{eq:h1rule} complete the section.

\subsection{Data Generating Processes}

In every design, $T = 3$, the true parameters are $\beta_0 = 1$ and $\alpha_0 = 0.5$, the social influence function is $\lambda(w) = (4w^3 - 3)/2$, the agent effect is $\mu_i \sim \N(0,1)$, independent of the latent index and absent from the link function, the outcome shocks are standard logistic, the link shocks are uniform and independent across dyads and waves, and the initial outcome is generated from the model without a lag, $y_{i0} = \ind\{X_{i0}\beta_0 + \lambda(w_{i0}) + \mu_i - \varepsilon_{i0} \ge 0\}$. The covariate depends on both unobserved terms, so that the naive estimator is inconsistent for two reasons, and the agent effect makes any comparison that requires equal levels of the unobserved term invalid.

\paragraph{Finite-type designs.} The latent index takes four values, $w \in \{0.15, 0.40, 0.65, 0.90\}$, labeled types $k = 0, 1, 2, 3$ and drawn with equal probability. Types are constant at $t = 0, 1$. At $t = 2$, every agent moves by the common cyclic permutation $k \mapsto k + 1 \pmod 4$, and types are constant again at $t = 3$. The type distribution is therefore the same in every period, the movement of $\lambda$ between the transition waves 1 and 2 is common to agents of the same type ($d_1 \neq 0$), and there is no movement after the transition ($d_2 = 0$), so Assumption~\ref{ass:persistence} holds with $\mathcal{S} = \{1\}$. Since the four types have distinct link profiles, every network-type equivalence class is a single type, and the induced change $d_1 = \lambda \circ g_1 - \lambda$, with $g_1$ the common map defined by the permutation, is trivially constant on each class. The covariate is integer valued, $X_{it} = \operatorname{round}(2w_{it} + 0.5\mu_i + \upsilon_{it})$ with $\upsilon_{it} \sim \N(0,1)$, so it is correlated with the latent index and with the agent effect, and exact covariate stability has positive probability. Links form with probability $f(w_{it}, w_{jt}) = P_{k_{it} k_{jt}}$, where
\[
P = \begin{pmatrix}
0.06 & 0.66 & 0.65 & 0.33\\
0.66 & 0.59 & 0.11 & 0.15\\
0.65 & 0.11 & 0.49 & 0.61\\
0.33 & 0.15 & 0.61 & 0.05
\end{pmatrix}.
\]
The four types have distinct codegree profiles: the smallest population squared codegree distance between different types, computed exactly from $P$ and the equal type shares, is $\bar\epsilon = 3.336 \times 10^{-3}$, so the separation gap of Assumption~\ref{ass:discrete} is positive. The matrix is chosen so that the average latent index of an agent's neighbors carries little information about the agent's own index, which is the case in which peer-average controls cannot proxy for the latent characteristic. This baseline design, which I call \emph{sbm4}, satisfies Assumptions~\ref{ass:sampling}, \ref{ass:persistence}, and~\ref{ass:discrete}. Three variants keep these assumptions and change one feature of the design.
\begin{itemize}[leftmargin=1.4em,itemsep=0.15em]
\item \emph{sbm4dyn}: link probabilities are multiplied by $1 + 0.4\,y_{i,t-1}y_{j,t-1}$, so that links depend on lagged outcomes and the lag stratification of the codegree distance binds. The separation gap, computed over $(\text{type}, y_{i0})$ cells with the population cell shares, is $3.852 \times 10^{-3}$.
\item \emph{fixed}: types do not move, so the heterogeneity is time invariant and within-agent conditioning is valid.
\item \emph{weak}: the block matrix is shrunk halfway toward its mean, which divides the separation gap by ten ($\bar\epsilon = 3.281 \times 10^{-4}$). The design satisfies the assumptions, but the gap is close to the sampling noise of the estimated distances at the sample sizes considered.
\end{itemize}
Three further variants leave the assumptions, each in one respect.
\begin{itemize}[leftmargin=1.4em,itemsep=0.15em]
\item \emph{idio}: at $t = 2$, a randomly chosen 30\% of agents draw a new type uniformly instead of following the permutation, so the movement is not common to network types and Assumption~\ref{ass:persistence}(i) fails.
\item \emph{xlink}: link probabilities are multiplied by $1.5$ when $X_{it} = X_{jt}$ and by $0.75$ otherwise, so that observed covariates enter the link function, which the network model of equation~\eqref{eq:network} excludes. Agents of the same type with different covariates then have different codegree profiles. The separation gap between different types, computed over $(\text{type}, X_{i1}, y_{i0})$ cells, is $8.120 \times 10^{-4}$.
\item \emph{sparse}: link probabilities are rescaled so that the expected degree is 8 at every sample size, outside the dense-graphon asymptotics of Section~\ref{sec:asymptotics}. This design also changes the bandwidth rule, as explained in Section~\ref{sec:mc_bandwidths}.
\end{itemize}

\paragraph{Continuous design.} The latent index is $w_{i0} = w_{i1} \sim U[0,1]$ and $w_{i2} = w_{i3} = w_{i1}^{1/2}$, so the movement between the transition waves is common to agents with the same index and there is no movement afterward. The distribution of the index differs across periods, which Assumption~\ref{ass:sampling}(iii) allows. Links follow the steep beta model $f(u, v) = F(8(u + v - 1))$, with $F$ the logistic distribution function. The link function is continuous and strictly increasing in each argument and therefore separates types. The covariate is continuous, $X_{it} = (3w_{it}^3 + 0.5\mu_i + \upsilon_{it} + 1)/3$. This design, called \emph{cont}, is in the regime of Theorem~\ref{thm:consistency} only.

\subsection{Estimators and Bandwidths}\label{sec:mc_bandwidths}

Six estimators are compared (five in the continuous design, which omits the controls estimator), and four class-based estimators are added. The \emph{naive} estimator is a pooled dynamic logit with an intercept, ignoring $\lambda$. The \emph{controls} estimator adds the averages of neighbors' covariates and lagged outcomes. The \emph{HK} estimator is the within-agent conditional logit of \citet{honore2000panel} on the transition event $y_{i1} + y_{i2} = 1$, with regressors $(X_{i1} - X_{i2}, y_{i0} - y_{i3})$ and the covariate-stability weight. The \emph{proposed} estimator is~\eqref{eq:estimator} with $\mathcal{S} = \{1\}$. The \emph{oracle} estimator is the same pairwise estimator computed with population matches: in the finite-type designs, the true type match; in the continuous design, the population codegree distances $\delta_1(w_{i1}, w_{j1})^2$, evaluated numerically by Gauss--Legendre quadrature from the link function and the uniform distribution of the index, so that neither the link noise nor the sampling variation of the other agents' latent indices enters the oracle's first stage. It isolates the cost of estimating the matches. The \emph{infeasible} estimator is a pooled dynamic logit that includes the true unobserved term $\lambda(w_{it}) + \mu_i$ as a regressor. \label{loc:mc_classest}The \emph{class} estimator is~\eqref{eq:class_estimator} with $\mathcal{S} = \{1\}$ and the feasible classes of~\eqref{eq:class_graph}, computed from the same estimated distances, the same bandwidth $h_1$, and the same covariate-stability selection as the proposed estimator. The \emph{class oracle} uses the population classes, the true type at $t = 1$ crossed with $y_{i0}$. The \emph{joint} estimator estimates one intercept per feasible class together with $(\beta, \alpha)$ instead of conditioning the intercepts out. The \emph{one-class} estimator, a within-agent conditional logit with a common time effect conditioned out, uses a single class per lag stratum, so it allows one common movement of $\lambda$ for all eligible agents that share $y_{i0}$ and uses no network information. In the continuous design the class estimator uses the connected components of the thresholded graph, which are not classes of equal movement there, with the covariate-stability window used as an indicator. The class-based estimators are computed by a separate program that regenerates each saved sample from its seed. Before computing them, it reproduces the proposed estimator from the regenerated sample, which it does in all 13{,}500 replications, to a largest absolute difference of $7 \times 10^{-13}$ in the estimates and standard errors.

The bandwidths are deterministic sequences with the rates that the theory requires. In the finite-type designs, $h_1 = 0.5\,\bar\epsilon\,(n/500)^{-1/4}$, where $\bar\epsilon$ is the population separation gap of the design, a fixed constant computed once from the block matrix and the population shares of the $(\text{type}, y_{i0})$ or $(\text{type}, X_{i1}, y_{i0})$ cells ($3.336 \times 10^{-3}$ in the baseline, fixed, and idiosyncratic designs, and the values given above for the lag-dependent, weak, and covariate-link designs). The cell shares are obtained by numerical integration of the model for the initial outcome and the covariate, and the code that computes them accompanies the paper. With $T = 3$, only period 1 enters the aggregated distance of a discordant pair, so this sequence satisfies Assumption~\ref{ass:discrete}(v). The calibration uses the population link probabilities and is therefore infeasible in practice. It is the sequence for which Theorem~\ref{thm:normality} is stated, and Section~\ref{sec:mc_datadriven} replaces it by the data-driven rule~\eqref{eq:h1rule}. A variant of the rule that recomputes $\bar\epsilon$ in each replication from the realized type shares, which makes the bandwidth random, gives practically the same results in the baseline design: the set of pairs with positive weight is identical in 91 to 100 percent of the replications, the median absolute change in $\hat\beta$ is at most 0.0015, the bias for $\beta$ changes by at most 0.001, and its coverage is unchanged to three decimals. Since the covariate is integer valued, $h_2 = 0.5$ selects exact stability, for HK as well as for the pairwise estimators. The exception for $h_1$ is the sparse design, in which the separation gap is far below the sampling noise of the estimated distances, so no bandwidth below the gap can select matches. There, $h_1$ is set to the 0.10 quantile of the estimated distances among pairs sharing the initial outcome, a practical rule that the theory does not cover. In the continuous design, $h_1 = c\,(n/500)^{-1/5}$, where $c = 5.13 \times 10^{-4}$ is the population 0.10 quantile of the squared codegree distance between two independent uniform indices at $t = 1$, computed by numerical integration (quadrature for the codegree integrals and an $8000 \times 8000$ midpoint grid with Richardson extrapolation for the quantile, accurate to about 0.2 percent), and $h_2 = n^{-1/5}/5$, for HK as well. All pairwise estimators use the kernel $K(x) = \tfrac34(1 - x^2)\ind(0 \le x < 1)$.

Standard errors are the MLE standard errors for the pooled logits, a sandwich for HK, and, for the proposed and oracle estimators, the sandwich that clusters the pair-level scores by agent, which is the plug-in estimator of Theorem~\ref{thm:normality} in the finite-type designs. For the class-conditional estimators, they are the inverse of the negative Hessian of the class-conditional log likelihood, as in Proposition~\ref{prop:class}(c), and for the joint estimator the inverse information of the joint logit. In the continuous design no distribution theory is established, and the same standard errors are reported only to show how far they are from the sampling variability. Each design is run with $R = 500$ replications. At that number, the Monte Carlo standard error of a coverage rate near 95\% is about one percentage point, and Table~\ref{tab:mc_success} reports it cell by cell. A fit counts as failed when fewer than three pairs (for the within-agent estimator, fewer than three agents) receive positive weight, when the Hessian is singular, when the Newton iterations do not converge within 100 steps (with a tolerance of $10^{-9}$ on the step), or when a coefficient exceeds 50 in absolute value. \label{loc:mc_failures}The last two are symptoms of a likelihood without a maximizer, but they do not establish it. For every failed fit of the pairwise and within-agent estimators whose selected sample is usable, that is, every failure due to a singular Hessian, nonconvergence, or divergence, a linear program therefore checks whether the weighted likelihood is completely separated, in which case a direction orders every positive-weight pair correctly, or quasi-completely separated, in which case a direction orders every pair weakly and some strictly \citep{albert1984}. In either case, the maximum likelihood estimate does not exist. Fits that fail because fewer than three pairs or agents are selected are classified separately and are not checked. For the class-based estimators, the same classification applies with fewer than three eligible agents in classes that contain both orderings, or a design without full rank, in the role of too few pairs, and the separation check is applied to the logit with one dummy per class. The tables report the number of failed fits of each kind and, for the optimizer failures, the number certified as separated in one of the two senses. All statistics are computed over the replications with an estimate, and Table~\ref{tab:mc_success} reports both the coverage conditional on a returned interval and the share of all attempted replications that return a covering interval.
\subsection{Results}

Tables~\ref{tab:mc_sbm4_beta} and~\ref{tab:mc_sbm4_alpha} report the baseline design, Tables~\ref{tab:mc_diag} and~\ref{tab:mc_success} the matching, inference, and success diagnostics of the proposed estimator, Table~\ref{tab:mc_class} the class-conditional estimator, and Table~\ref{tab:mc_bench_2000} the benchmark designs at $n = 2000$. Appendix~\ref{sec:appendix_figures} reports the other designs and sample sizes.

\begin{table}[!ht]
\centering
\caption{Monte Carlo results for $\beta$: baseline finite-type design}\label{tab:mc_sbm4_beta}
\smallskip
\footnotesize
\setlength{\tabcolsep}{4.5pt}
\begin{tabular}{l cccccccc}
\toprule
$n$ & Naive & Controls & HK & Proposed & Oracle & Class & Class orc. & Infeas. \\
\midrule
\multicolumn{9}{l}{\textit{Bias}} \\[2pt]
500 & 0.247 & 0.242 & 0.347 & 0.240 & 0.235 & 0.149 & 0.128 & 0.009 \\
1000 & 0.245 & 0.239 & 0.296 & 0.151 & 0.151 & 0.085 & 0.084 & 0.008 \\
2000 & 0.240 & 0.233 & 0.258 & 0.070 & 0.069 & 0.042 & 0.042 & 0.001 \\
3000 & 0.240 & 0.233 & 0.228 & 0.026 & 0.026 & 0.009 & 0.009 & 0.001 \\
\addlinespace[4pt]
\multicolumn{9}{l}{\textit{Standard deviation}} \\[2pt]
500 & 0.069 & 0.069 & 0.450 & 0.633 & 0.628 & 0.544 & 0.533 & 0.076 \\
1000 & 0.053 & 0.053 & 0.316 & 0.433 & 0.434 & 0.372 & 0.372 & 0.056 \\
2000 & 0.034 & 0.034 & 0.191 & 0.253 & 0.253 & 0.219 & 0.219 & 0.036 \\
3000 & 0.030 & 0.030 & 0.154 & 0.199 & 0.199 & 0.174 & 0.174 & 0.033 \\
\addlinespace[4pt]
\multicolumn{9}{l}{\textit{Root mean squared error}} \\[2pt]
500 & 0.256 & 0.252 & 0.568 & 0.677 & 0.670 & 0.563 & 0.548 & 0.076 \\
1000 & 0.251 & 0.245 & 0.433 & 0.458 & 0.459 & 0.381 & 0.381 & 0.056 \\
2000 & 0.242 & 0.236 & 0.321 & 0.262 & 0.262 & 0.222 & 0.222 & 0.036 \\
3000 & 0.242 & 0.235 & 0.275 & 0.201 & 0.201 & 0.174 & 0.174 & 0.033 \\
\addlinespace[4pt]
\multicolumn{9}{l}{\textit{Coverage of nominal 95\% interval}} \\[2pt]
500 & 0.050 & 0.058 & 0.957 & 0.968 & 0.968 & 0.963 & 0.963 & 0.934 \\
1000 & 0.000 & 0.004 & 0.876 & 0.950 & 0.950 & 0.974 & 0.974 & 0.944 \\
2000 & 0.000 & 0.000 & 0.772 & 0.944 & 0.944 & 0.960 & 0.960 & 0.970 \\
3000 & 0.000 & 0.000 & 0.754 & 0.942 & 0.944 & 0.950 & 0.950 & 0.936 \\
\addlinespace[4pt]
\multicolumn{9}{l}{\textit{Replications without an estimate}} \\[2pt]
500 & 0 & 0 & 12 & 69 & 69 & 65 & 69 & 0 \\
1000 & 0 & 0 & 0 & 0 & 0 & 0 & 0 & 0 \\
2000 & 0 & 0 & 0 & 0 & 0 & 0 & 0 & 0 \\
3000 & 0 & 0 & 0 & 0 & 0 & 0 & 0 & 0 \\
\bottomrule
\end{tabular}
\vspace{4pt}\par\raggedright\footnotesize \textit{Notes:} $R = 500$ replications; $\beta_0 = 1$, $\alpha_0 = 0.5$. Naive: pooled dynamic logit. Controls: pooled dynamic logit with neighbors' average covariate and lagged outcome. HK: within-agent conditional logit of \citet{honore2000panel}. Proposed: estimator~\eqref{eq:estimator}. Oracle: the same pairwise estimator with population matches. Class: the class-conditional estimator~\eqref{eq:class_estimator}, with classes formed from the same estimated distances and bandwidth. Class orc.: the same likelihood with the population classes. Infeas.: pooled dynamic logit including the true unobserved term $\lambda(w_{it}) + \mu_i$. Bandwidths: $h_1 = 0.5\,\bar\epsilon\,(n/500)^{-1/4}$, with $\bar\epsilon$ the population separation gap of the design, a fixed constant computed from the block matrix and the population type shares, and exact covariate stability ($h_2 = 0.5$ on the integer covariate). Coverage: share of replications in which the nominal 95\% Wald interval contains the true value. Standard errors are MLE for the pooled logits, a sandwich for HK, the sandwich clustering pair-level scores by agent for the pairwise estimators, and the inverse of the negative Hessian of the class-conditional log likelihood for the class estimators. Statistics are computed over the replications with an estimate, and failures are defined in Section~\ref{sec:mc_bandwidths}.
\end{table}
\begin{table}[!ht]
\centering
\caption{Monte Carlo results for $\alpha$: baseline finite-type design}\label{tab:mc_sbm4_alpha}
\smallskip
\footnotesize
\setlength{\tabcolsep}{4.5pt}
\begin{tabular}{l cccccccc}
\toprule
$n$ & Naive & Controls & HK & Proposed & Oracle & Class & Class orc. & Infeas. \\
\midrule
\multicolumn{9}{l}{\textit{Bias}} \\[2pt]
500 & 0.701 & 0.723 & 0.373 & $-$0.076 & $-$0.075 & $-$0.044 & $-$0.050 & $-$0.016 \\
1000 & 0.713 & 0.737 & 0.377 & 0.000 & 0.000 & 0.003 & 0.002 & 0.004 \\
2000 & 0.708 & 0.732 & 0.374 & 0.016 & 0.016 & $-$0.005 & $-$0.005 & $-$0.004 \\
3000 & 0.709 & 0.734 & 0.371 & 0.027 & 0.027 & 0.024 & 0.024 & $-$0.005 \\
\addlinespace[4pt]
\multicolumn{9}{l}{\textit{Standard deviation}} \\[2pt]
500 & 0.143 & 0.145 & 0.736 & 1.299 & 1.296 & 1.164 & 1.131 & 0.151 \\
1000 & 0.100 & 0.103 & 0.485 & 0.873 & 0.872 & 0.758 & 0.758 & 0.108 \\
2000 & 0.070 & 0.073 & 0.330 & 0.543 & 0.543 & 0.478 & 0.478 & 0.073 \\
3000 & 0.057 & 0.060 & 0.267 & 0.418 & 0.418 & 0.380 & 0.380 & 0.060 \\
\addlinespace[4pt]
\multicolumn{9}{l}{\textit{Root mean squared error}} \\[2pt]
500 & 0.715 & 0.737 & 0.825 & 1.300 & 1.297 & 1.163 & 1.131 & 0.152 \\
1000 & 0.720 & 0.744 & 0.614 & 0.872 & 0.872 & 0.757 & 0.757 & 0.108 \\
2000 & 0.712 & 0.736 & 0.499 & 0.543 & 0.542 & 0.477 & 0.477 & 0.073 \\
3000 & 0.712 & 0.736 & 0.457 & 0.419 & 0.419 & 0.380 & 0.380 & 0.061 \\
\addlinespace[4pt]
\multicolumn{9}{l}{\textit{Coverage of nominal 95\% interval}} \\[2pt]
500 & 0.002 & 0.000 & 0.932 & 0.961 & 0.961 & 0.991 & 0.991 & 0.938 \\
1000 & 0.000 & 0.000 & 0.898 & 0.938 & 0.940 & 0.964 & 0.964 & 0.932 \\
2000 & 0.000 & 0.000 & 0.802 & 0.942 & 0.942 & 0.952 & 0.952 & 0.946 \\
3000 & 0.000 & 0.000 & 0.690 & 0.944 & 0.944 & 0.954 & 0.954 & 0.956 \\
\addlinespace[4pt]
\multicolumn{9}{l}{\textit{Replications without an estimate}} \\[2pt]
500 & 0 & 0 & 12 & 69 & 69 & 65 & 69 & 0 \\
1000 & 0 & 0 & 0 & 0 & 0 & 0 & 0 & 0 \\
2000 & 0 & 0 & 0 & 0 & 0 & 0 & 0 & 0 \\
3000 & 0 & 0 & 0 & 0 & 0 & 0 & 0 & 0 \\
\bottomrule
\end{tabular}
\vspace{4pt}\par\raggedright\footnotesize \textit{Notes:} $R = 500$ replications; $\beta_0 = 1$, $\alpha_0 = 0.5$. Naive: pooled dynamic logit. Controls: pooled dynamic logit with neighbors' average covariate and lagged outcome. HK: within-agent conditional logit of \citet{honore2000panel}. Proposed: estimator~\eqref{eq:estimator}. Oracle: the same pairwise estimator with population matches. Class: the class-conditional estimator~\eqref{eq:class_estimator}, with classes formed from the same estimated distances and bandwidth. Class orc.: the same likelihood with the population classes. Infeas.: pooled dynamic logit including the true unobserved term $\lambda(w_{it}) + \mu_i$. Bandwidths: $h_1 = 0.5\,\bar\epsilon\,(n/500)^{-1/4}$, with $\bar\epsilon$ the population separation gap of the design, a fixed constant computed from the block matrix and the population type shares, and exact covariate stability ($h_2 = 0.5$ on the integer covariate). Coverage: share of replications in which the nominal 95\% Wald interval contains the true value. Standard errors are MLE for the pooled logits, a sandwich for HK, the sandwich clustering pair-level scores by agent for the pairwise estimators, and the inverse of the negative Hessian of the class-conditional log likelihood for the class estimators. Statistics are computed over the replications with an estimate, and failures are defined in Section~\ref{sec:mc_bandwidths}.
\end{table}

\paragraph{The finite-type regime.} The baseline design is the regime of Theorem~\ref{thm:normality}. From $n = 1000$ onward, the estimated weights select exactly the same-type pairs (precision and recall equal one, against 0.996 at $n = 500$), and the proposed estimator agrees with the oracle to within 0.001 in bias, standard deviation, and RMSE, as the oracle equivalence in the proof of the theorem implies. Its bias for $\beta$ falls from 0.240 at $n = 500$ to 0.026 at $n = 3000$, and its bias for $\alpha$ is below 0.03 in absolute value from $n = 1000$ onward. The oracle has the same small-sample bias, so the bias comes from a conditional logit that rests on few agents, not from the estimated matches. \label{loc:mc_agents}At $n = 500$, about 50 agents are eligible and about 36 have a partner (Table~\ref{tab:mc_diag}). The information grows with the number of contributing agents, not with the number of pairs, because pairs that share an agent are dependent. At that sample size 69 fits fail, all on separated likelihoods. Conditional on a returned interval, coverage is between 0.94 and 0.97 for both parameters at every sample size, but at $n = 500$ only 83 percent of all attempts return an interval that covers $\beta_0$ (Table~\ref{tab:mc_success}). The clustered standard errors are slightly too small: the ratio of the Monte Carlo standard deviation to the mean standard error is between 1.02 and 1.16.

\begin{table}[!ht]
\centering
\caption{Baseline and lag-dependent designs: matching, contributing-agent, and inference diagnostics for the proposed estimator}\label{tab:mc_diag}
\smallskip
\footnotesize
\setlength{\tabcolsep}{3.2pt}
\resizebox{\textwidth}{!}{%
\begin{tabular}{ll rr rrrr cc cccccc}
\toprule
 & & & & \multicolumn{4}{c}{Agents} & \multicolumn{2}{c}{Matching} & \multicolumn{6}{c}{Proposed estimator} \\
\cmidrule(lr){5-8} \cmidrule(lr){9-10} \cmidrule(lr){11-16}
Design & $n$ & Pairs & ESS & Elig. & Contrib. & Unm. & Min. & Prec. & Recall & Par. & Med.\ bias & SD & Mean SE & SD/SE & Coverage \\
\midrule
sbm4 & 500 & 48 & 48 & 50 & 36 & 14 & 1.0 & 1.00 & 1.00 & $\beta$ & 0.150 & 0.633 & 0.614 & 1.03 & 0.968 \\
 & & & & & & & & &  & $\alpha$ & $-$0.008 & 1.299 & 1.247 & 1.04 & 0.961 \\
\addlinespace[2pt]
sbm4 & 1000 & 190 & 190 & 99 & 86 & 14 & 1.1 & 1.00 & 1.00 & $\beta$ & 0.086 & 0.433 & 0.375 & 1.16 & 0.950 \\
 & & & & & & & & &  & $\alpha$ & 0.034 & 0.873 & 0.804 & 1.09 & 0.938 \\
\addlinespace[2pt]
sbm4 & 2000 & 772 & 772 & 200 & 191 & 9 & 1.6 & 1.00 & 1.00 & $\beta$ & 0.047 & 0.253 & 0.236 & 1.07 & 0.944 \\
 & & & & & & & & &  & $\alpha$ & 0.027 & 0.543 & 0.524 & 1.04 & 0.942 \\
\addlinespace[2pt]
sbm4 & 3000 & 1733 & 1733 & 299 & 294 & 5 & 2.1 & 1.00 & 1.00 & $\beta$ & 0.017 & 0.199 & 0.187 & 1.07 & 0.942 \\
 & & & & & & & & &  & $\alpha$ & 0.041 & 0.418 & 0.411 & 1.02 & 0.944 \\
\addlinespace[2pt]
sbm4dyn & 500 & 48 & 48 & 50 & 36 & 14 & 1.0 & 1.00 & 1.00 & $\beta$ & 0.151 & 0.632 & 0.608 & 1.04 & 0.965 \\
 & & & & & & & & &  & $\alpha$ & $-$0.019 & 1.427 & 1.279 & 1.12 & 0.941 \\
\addlinespace[2pt]
sbm4dyn & 1000 & 192 & 192 & 100 & 86 & 14 & 1.2 & 1.00 & 1.00 & $\beta$ & 0.069 & 0.411 & 0.373 & 1.10 & 0.950 \\
 & & & & & & & & &  & $\alpha$ & $-$0.030 & 0.917 & 0.809 & 1.13 & 0.942 \\
\addlinespace[2pt]
sbm4dyn & 2000 & 757 & 757 & 200 & 190 & 10 & 1.6 & 1.00 & 1.00 & $\beta$ & 0.052 & 0.262 & 0.240 & 1.09 & 0.948 \\
 & & & & & & & & &  & $\alpha$ & $-$0.005 & 0.551 & 0.521 & 1.06 & 0.938 \\
\addlinespace[2pt]
sbm4dyn & 3000 & 1734 & 1733 & 300 & 295 & 4 & 2.0 & 1.00 & 1.00 & $\beta$ & 0.042 & 0.187 & 0.190 & 0.98 & 0.950 \\
 & & & & & & & & &  & $\alpha$ & $-$0.008 & 0.415 & 0.411 & 1.01 & 0.946 \\
\addlinespace[2pt]
\bottomrule
\end{tabular}}
\vspace{4pt}\par\raggedright\footnotesize \textit{Notes:} Pairs: average number of pairs with positive weight. ESS: kernel-weighted effective sample size $(\sum_p w_p)^2/\sum_p w_p^2$. Elig.: average number of eligible agents, who switch between waves 1 and 2 and have a positive covariate-stability weight (the agents that can enter a discordant pair and the agents in the within-agent likelihood). Contrib.: average number of eligible agents with at least one positive-weight partner. Unm.: average number of eligible agents with no partner. Min.: average over replications of the smallest number of positive-weight partners among contributing agents. Pair counts and ESS measure the concentration of the weights, not the number of independent agents, because pairs that share an agent are dependent. Prec.: share of weighted pairs whose members have the same type. Recall: share of same-type discordant pairs with stable covariates that receive positive weight. Med.\ bias: median of $\hat\theta - \theta_0$. SD: Monte Carlo standard deviation. Mean SE: average of the clustered standard errors. The mean length of the 95\% interval is 3.92 times this number. Coverage of the nominal 95\% interval. Failed fits: sbm4 at $n = 500$: Proposed 69 (69 optimizer failures, all certified as separated); sbm4dyn at $n = 500$: Proposed 75 (75 optimizer failures, all certified as separated); and sbm4dyn at $n = 1000$: Proposed 1 (1 optimizer failure, certified as separated). 
\end{table}

\begin{table}[!ht]
\centering
\caption{Baseline and lag-dependent designs: success-qualified diagnostics for the proposed estimator}\label{tab:mc_success}
\smallskip
\footnotesize
\setlength{\tabcolsep}{3.2pt}
\resizebox{\textwidth}{!}{%
\begin{tabular}{ll r c rrr cc c cc c}
\toprule
Design & $n$ & Fits & Par. & Bias & 5th pct. & 95th pct. & SD & Mean SE & Length & Coverage & MC SE & Covering share \\
\midrule
Baseline & 500 & 431/500 & $\beta$ & 0.240 & $-$0.610 & 1.529 & 0.633 & 0.614 & 2.41 & 0.968 & 0.009 & 0.834 \\
 & &  & $\alpha$ & $-$0.076 & $-$2.231 & 1.947 & 1.299 & 1.247 & 4.89 & 0.961 & 0.009 & 0.828 \\
\addlinespace[2pt]
Baseline & 1000 & 500/500 & $\beta$ & 0.151 & $-$0.398 & 0.861 & 0.433 & 0.375 & 1.47 & 0.950 & 0.010 & 0.950 \\
 & &  & $\alpha$ & 0.000 & $-$1.476 & 1.366 & 0.873 & 0.804 & 3.15 & 0.938 & 0.011 & 0.938 \\
\addlinespace[2pt]
Baseline & 2000 & 500/500 & $\beta$ & 0.070 & $-$0.277 & 0.538 & 0.253 & 0.236 & 0.92 & 0.944 & 0.010 & 0.944 \\
 & &  & $\alpha$ & 0.016 & $-$0.844 & 0.948 & 0.543 & 0.524 & 2.05 & 0.942 & 0.010 & 0.942 \\
\addlinespace[2pt]
Baseline & 3000 & 500/500 & $\beta$ & 0.026 & $-$0.282 & 0.366 & 0.199 & 0.187 & 0.73 & 0.942 & 0.010 & 0.942 \\
 & &  & $\alpha$ & 0.027 & $-$0.647 & 0.734 & 0.418 & 0.411 & 1.61 & 0.944 & 0.010 & 0.944 \\
\addlinespace[2pt]
Lag-dependent links & 500 & 425/500 & $\beta$ & 0.234 & $-$0.650 & 1.412 & 0.632 & 0.608 & 2.38 & 0.965 & 0.009 & 0.820 \\
 & &  & $\alpha$ & $-$0.055 & $-$2.436 & 2.127 & 1.427 & 1.279 & 5.02 & 0.941 & 0.011 & 0.800 \\
\addlinespace[2pt]
Lag-dependent links & 1000 & 499/500 & $\beta$ & 0.143 & $-$0.399 & 0.857 & 0.411 & 0.373 & 1.46 & 0.950 & 0.010 & 0.948 \\
 & &  & $\alpha$ & 0.008 & $-$1.317 & 1.568 & 0.917 & 0.809 & 3.17 & 0.942 & 0.010 & 0.940 \\
\addlinespace[2pt]
Lag-dependent links & 2000 & 500/500 & $\beta$ & 0.068 & $-$0.311 & 0.482 & 0.262 & 0.240 & 0.94 & 0.948 & 0.010 & 0.948 \\
 & &  & $\alpha$ & 0.000 & $-$0.893 & 0.881 & 0.551 & 0.521 & 2.04 & 0.938 & 0.011 & 0.938 \\
\addlinespace[2pt]
Lag-dependent links & 3000 & 500/500 & $\beta$ & 0.043 & $-$0.258 & 0.359 & 0.187 & 0.190 & 0.74 & 0.950 & 0.010 & 0.950 \\
 & &  & $\alpha$ & 0.004 & $-$0.699 & 0.724 & 0.415 & 0.411 & 1.61 & 0.946 & 0.010 & 0.946 \\
\addlinespace[2pt]
\bottomrule
\end{tabular}}
\vspace{4pt}\par\raggedright\footnotesize \textit{Notes:} Fits: replications returning an estimate, out of 500 attempted. Bias, 5th and 95th percentiles of $\hat\theta - \theta_0$, Monte Carlo standard deviation (SD), and mean clustered standard error, computed over the successful fits. Length: mean length of the nominal 95\% Wald interval. Coverage: share of successful fits whose interval contains the true value, with its Monte Carlo standard error (MC SE). Covering share: share of all 500 attempts that return an interval containing the true value, which is not a coverage probability but an operational success rate. Failed fits: Baseline at $n = 500$: Proposed 69 (69 optimizer failures, all certified as separated); Lag-dependent links at $n = 500$: Proposed 75 (75 optimizer failures, all certified as separated); and Lag-dependent links at $n = 1000$: Proposed 1 (1 optimizer failure, certified as separated). 
\end{table}

\paragraph{Estimators that do not use the path comparison.} The naive and controls estimators are biased by 0.23 to 0.25 for $\beta$ and 0.70 to 0.74 for $\alpha$ at every sample size: a pooled logit attributes the omitted persistent effect to state dependence, and neighbors' averages barely help, because friends' characteristics carry little information about an agent's own type and none about the agent effect. The within-agent estimator removes the agent effect and the level of $\lambda$ but not the common movement of $\lambda$, so its bias, 0.23 to 0.35 for $\beta$ and about 0.37 for $\alpha$, does not vanish, and its coverage for $\beta$ falls from 0.96 at $n = 500$ to 0.75 at $n = 3000$. It still has the smaller RMSE for both parameters at $n \le 1000$ (for $\beta$, 0.568 against 0.677 at $n = 500$), whereas the proposed estimator has the smaller RMSE for $\beta$ from $n = 2000$ and for both parameters at $n = 3000$: removing a bias that does not shrink costs variance, which pays off only as $n$ grows. \label{loc:mc_common}Restricting each comparison to the replications in which both estimators succeed does not reverse any ranking (Table~\ref{tab:mc_common}).

\paragraph{Conditioning on network classes.}\label{loc:mc_class} In the baseline and lag-dependent designs, the class-conditional estimator has the smaller RMSE for both parameters at every sample size, by 6 to 20 percent, with coverage between 0.95 and 0.99 (Table~\ref{tab:mc_class}). The gain comes from the likelihood, not from the classification: from $n = 1000$ the feasible classes are the population classes in at least 99.6 percent of the replications, and with the population classes the gain over the pairwise oracle is 6 to 19 percent, because the class likelihood uses each class's sufficient statistic once, whereas the pairwise likelihood counts overlapping pairs. Conditioning matters: estimating the class intercepts instead raises the bias for $\beta$ at $n = 500$ from 0.149 to 0.485 (Table~\ref{tab:mc_joint}). The network classes matter: a single class per lag stratum, which uses no network information, is biased by 0.229 at $n = 3000$, as is the within-agent estimator (0.228), because the movements differ across types. In the weak-network design, the transitive class graph joins same-type agents whose direct distance exceeds the bandwidth. Its partitions split classes but never merge them, which keeps the class likelihood exact, and the RMSE for $\beta$ at $n = 2000$ falls from 0.362 to 0.246. Outside the assumptions, the class estimator is biased where the pairwise estimator is (Tables~\ref{tab:mc_bench_2000}, \ref{tab:mc_bench_1000}, and~\ref{tab:mc_class_bench}).

\begin{table}[!ht]
\centering
\caption{Pairwise and class-conditional estimators: baseline and lag-dependent designs}\label{tab:mc_class}
\smallskip
\footnotesize
\setlength{\tabcolsep}{3.2pt}
\resizebox{\textwidth}{!}{%
\begin{tabular}{ll rrr ccc ccc cc cc cc}
\toprule
 & & \multicolumn{3}{c}{Successful fits} & \multicolumn{3}{c}{RMSE for $\beta$} & \multicolumn{3}{c}{RMSE for $\alpha$} & \multicolumn{2}{c}{Class: coverage} & \multicolumn{2}{c}{Classes} & \multicolumn{2}{c}{Time (ms)} \\
\cmidrule(lr){3-5} \cmidrule(lr){6-8} \cmidrule(lr){9-11} \cmidrule(lr){12-13} \cmidrule(lr){14-15} \cmidrule(lr){16-17}
Design & $n$ & Prop. & Class & Common & Prop. & Class & Class orc. & Prop. & Class & Class orc. & $\beta$ & $\alpha$ & Exact & Used & Prop. & Class \\
\midrule
Baseline & 500 & 431 & 435 & 431 & 0.677 & 0.551 & 0.548 & 1.300 & 1.136 & 1.131 & 0.963 & 0.991 & 88.2 & 36 & 0.5 & 12.6 \\
 & 1000 & 500 & 500 & 500 & 0.458 & 0.381 & 0.381 & 0.872 & 0.757 & 0.757 & 0.974 & 0.964 & 99.6 & 86 & 0.6 & 14.0 \\
 & 2000 & 500 & 500 & 500 & 0.262 & 0.222 & 0.222 & 0.543 & 0.477 & 0.477 & 0.960 & 0.952 & 100.0 & 191 & 1.0 & 22.4 \\
 & 3000 & 500 & 500 & 500 & 0.201 & 0.174 & 0.174 & 0.419 & 0.380 & 0.380 & 0.950 & 0.954 & 100.0 & 294 & 1.4 & 34.3 \\
\addlinespace[2pt]
Lag-dependent links & 500 & 425 & 426 & 424 & 0.674 & 0.541 & 0.540 & 1.426 & 1.206 & 1.225 & 0.984 & 0.984 & 93.4 & 36 & 0.5 & 6.2 \\
 & 1000 & 499 & 499 & 499 & 0.435 & 0.354 & 0.354 & 0.917 & 0.791 & 0.791 & 0.968 & 0.952 & 99.6 & 86 & 0.5 & 20.9 \\
 & 2000 & 500 & 500 & 500 & 0.270 & 0.232 & 0.232 & 0.550 & 0.472 & 0.472 & 0.960 & 0.950 & 100.0 & 190 & 0.9 & 21.8 \\
 & 3000 & 500 & 500 & 500 & 0.191 & 0.172 & 0.172 & 0.415 & 0.391 & 0.391 & 0.962 & 0.950 & 100.0 & 295 & 1.4 & 35.0 \\
\addlinespace[2pt]
\bottomrule
\end{tabular}}
\vspace{4pt}\par\raggedright\footnotesize \textit{Notes:} $R = 500$ replications. Prop.: the pairwise estimator~\eqref{eq:estimator}. Class: the class-conditional estimator~\eqref{eq:class_estimator}, with feasible classes formed from the pairwise estimator's estimated distances and bandwidth $h_1$. Class orc.: the same likelihood with the population classes (type at $t = 1$ crossed with $y_{i0}$). Successful fits: replications in which the estimator returns finite estimates and standard errors. Common: replications in which the pairwise, class, and class-oracle estimators all succeed (in the continuous design, which has no class oracle, the first two). RMSE on the common replications. Class coverage: coverage of the nominal 95\% interval based on the inverse of the negative Hessian of the class-conditional log likelihood, over its own successful fits. Exact: percentage of replications in which the feasible classes coincide with the population classes among the eligible agents. Used: average number of eligible agents in classes that contain both orderings. Time: median computing time per replication, in milliseconds on one thread (two single-threaded jobs ran concurrently), of each estimator's own step (pairs and fit for the pairwise estimator, connected components and fit for the class estimator). The codegree distances among eligible agents, which both use, take a further 0.7 to 100 milliseconds between $n = 500$ and $n = 3000$. Failed fits are optimizer failures, all certified as separated, except one class-oracle fit with a rank-deficient design: Baseline at $n = 500$: Proposed 69, Class 65, and Class orc.\ 69 (68 optimizer failures and the rank-deficient fit); Lag-dependent links at $n = 500$: Proposed 75, Class 74, and Class orc.\ 76; and Lag-dependent links at $n = 1000$: 1 for each estimator.
\end{table}

\begin{table}[!ht]
\centering
\caption{Benchmark designs at $n = 2000$}\label{tab:mc_bench_2000}
\smallskip
\footnotesize
\setlength{\tabcolsep}{3.2pt}
\resizebox{\textwidth}{!}{%
\begin{tabular}{l ccc ccc ccc ccc }
\toprule
 & \multicolumn{3}{c}{HK} & \multicolumn{3}{c}{Proposed} & \multicolumn{3}{c}{Oracle} & \multicolumn{3}{c}{Class} \\
\cmidrule(lr){2-4} \cmidrule(lr){5-7} \cmidrule(lr){8-10} \cmidrule(lr){11-13}
Design & Bias & RMSE & Cov. & Bias & RMSE & Cov. & Bias & RMSE & Cov. & Bias & RMSE & Cov. \\
\midrule
\multicolumn{13}{l}{\textit{$\beta$ ($\beta_0 = 1$)}} \\[2pt]
\multicolumn{13}{l}{\quad\textit{Designs that satisfy the assumptions}} \\[1pt]
Baseline & 0.258 & 0.321 & 0.77 & 0.070 & 0.262 & 0.94 & 0.069 & 0.262 & 0.94 & 0.042 & 0.222 & 0.96 \\
Lag-dependent links & 0.264 & 0.331 & 0.81 & 0.068 & 0.270 & 0.95 & 0.068 & 0.271 & 0.95 & 0.047 & 0.232 & 0.96 \\
Fixed types & 0.030 & 0.195 & 0.94 & 0.046 & 0.223 & 0.96 & 0.046 & 0.223 & 0.96 & 0.032 & 0.207 & 0.96 \\
Weak network & 0.252 & 0.322 & 0.81 & 0.110 & 0.362 & 0.94 & 0.064 & 0.268 & 0.95 & 0.041 & 0.246 & 0.94 \\
\multicolumn{13}{l}{\quad\textit{Designs outside the assumptions}} \\[1pt]
Idiosyncratic movement & 0.245 & 0.317 & 0.79 & 0.131 & 0.287 & 0.93 & 0.131 & 0.287 & 0.93 & 0.107 & 0.253 & 0.94 \\
Covariates in links & 0.258 & 0.327 & 0.78 & 0.104 & 0.394 & 0.95 & 0.064 & 0.249 & 0.94 & 0.051 & 0.288 & 0.95 \\
Sparse network & 0.223 & 0.292 & 0.85 & 0.555 & 0.833 & 0.78 & 0.036 & 0.239 & 0.96 & 0.312 & 0.431 & 0.90 \\
\addlinespace[4pt]
\multicolumn{13}{l}{\textit{$\alpha$ ($\alpha_0 = 0.5$)}} \\[2pt]
\multicolumn{13}{l}{\quad\textit{Designs that satisfy the assumptions}} \\[1pt]
Baseline & 0.374 & 0.499 & 0.80 & 0.016 & 0.543 & 0.94 & 0.016 & 0.542 & 0.94 & $-$0.005 & 0.477 & 0.95 \\
Lag-dependent links & 0.400 & 0.508 & 0.79 & 0.000 & 0.550 & 0.94 & 0.000 & 0.550 & 0.94 & 0.001 & 0.472 & 0.95 \\
Fixed types & 0.011 & 0.321 & 0.93 & 0.023 & 0.454 & 0.95 & 0.023 & 0.454 & 0.95 & 0.023 & 0.433 & 0.95 \\
Weak network & 0.359 & 0.489 & 0.81 & $-$0.058 & 0.721 & 0.93 & $-$0.026 & 0.568 & 0.93 & $-$0.003 & 0.522 & 0.93 \\
\multicolumn{13}{l}{\quad\textit{Designs outside the assumptions}} \\[1pt]
Idiosyncratic movement & 0.352 & 0.482 & 0.79 & 0.122 & 0.534 & 0.95 & 0.122 & 0.533 & 0.95 & 0.144 & 0.493 & 0.94 \\
Covariates in links & 0.364 & 0.490 & 0.80 & $-$0.022 & 0.641 & 0.96 & $-$0.012 & 0.541 & 0.94 & $-$0.015 & 0.516 & 0.95 \\
Sparse network & 0.395 & 0.494 & 0.79 & 0.487 & 1.195 & 0.87 & 0.029 & 0.501 & 0.96 & 0.404 & 0.735 & 0.91 \\
\addlinespace[4pt]
\bottomrule
\end{tabular}}
\vspace{4pt}\par\raggedright\footnotesize \textit{Notes:} $R = 500$ replications. Bias, root mean squared error, and coverage of the nominal 95\% interval, computed over the replications with an estimate. HK: within-agent conditional logit of \citet{honore2000panel}. Proposed: estimator~\eqref{eq:estimator}. Oracle: the same pairwise estimator with population matches. Class: the class-conditional estimator~\eqref{eq:class_estimator} with classes formed from the same estimated distances and bandwidth. The sparse design uses the 0.10 quantile of the estimated distances as $h_1$, and all other finite-type designs use $h_1 = 0.5\,\bar\epsilon\,(n/500)^{-1/4}$ with $\bar\epsilon$ the population separation gap of the design (Section~\ref{sec:mc_bandwidths}). The covariate-link and sparse designs violate the maintained network model (covariates in the link function and bounded degrees, respectively), and the idiosyncratic design violates Assumption~\ref{ass:persistence}(i). Failed fits: Sparse network at $n = 2000$: Proposed 2 (2 optimizer failures, all certified as separated). 
\end{table}

\paragraph{Designs in which the assumptions hold.} With lag-dependent links, the lag-stratified distance still separates the types, and the results are those of the baseline design (Table~\ref{tab:mc_bench_2000}). When the heterogeneity does not move, $\lambda(w_i) + \mu_i$ is a fixed effect, the within-agent estimator is valid and more precise (RMSE for $\beta$ of 0.195, against 0.223 for the pairwise and 0.207 for the class estimator), and the network adds nothing. In the weak-network design, the separation gap is ten times smaller. At $n = 1000$, it is below the noise of the estimated distances, so no pair is selected (Table~\ref{tab:mc_bench_1000}), and recall rises to 0.45 at $n = 2000$ and 0.79 at $n = 3000$ (Table~\ref{tab:mc_diag_bench}). The theory needs a gap that is large relative to $(\log n/n)^{1/2}$.

\paragraph{Designs outside the assumptions.} When 30\% of the agents move idiosyncratically, Assumption~\ref{ass:persistence}(i) fails, and the bias of the proposed estimator rises to 0.131 for $\beta$ and 0.122 for $\alpha$, against 0.070 and 0.016 in the baseline design. Its coverage stays near nominal only because its standard errors are large. When covariates enter link formation, same-type agents with different covariates have different codegree profiles, so recall falls to about 0.32 with precision one (Table~\ref{tab:mc_diag_bench}), and the RMSE for $\beta$ is 0.394, against 0.249 for the oracle. In the sparse design, with expected degree 8, fewer than one in five weighted pairs have the same type, and the bias for $\beta$ is 0.555 with coverage 0.78, while the oracle is centered: with bounded degrees, the codegree distance carries little information about types, as in the application of Section~\ref{sec:empirical}.

\paragraph{The continuous design.} Only consistency is established here (Theorem~\ref{thm:consistency}). With the shrinking bandwidths, the bias of the proposed estimator for $\beta$ falls from 0.431 at $n = 500$ to between 0.03 and 0.06 at $n \ge 2000$, within 0.05 of the oracle with population distances at every sample size, while the within-agent estimator's bias stays between 0.30 and 0.44 (Tables~\ref{tab:mc_cont_beta} and~\ref{tab:mc_cont_alpha}). Coverage, conditional on a returned interval, is between 0.95 and 0.98, but without a distribution theory it describes this design only. The class estimator also has a smaller RMSE than the pairwise estimator, but the gain comes from pooling: its connected components are not classes of equal movement, and a single class without network information does as well or better (RMSE for $\beta$ of 0.867 against 1.108 at $n = 500$ and 0.334 against 0.330 at $n = 3000$, with the class values from Table~\ref{tab:mc_class_bench} and the single-class values not tabulated).

\paragraph{Summary.} When its assumptions hold, the proposed estimator agrees with the oracle, its bias shrinks with $n$, and its intervals have close to nominal coverage among successful fits. When these assumptions hold and the heterogeneity moves, it and the class-conditional estimator are the only feasible estimators compared whose biases vanish, at the cost of precision in small samples. In the baseline and lag-dependent designs, conditioning on network classes reduces the RMSE by 6 to 20 percent. The method loses precision when the network is weakly informative or covariates enter link formation, and it fails when the movement is not common to network types or the network is sparse.

\subsection{A Data-Driven Bandwidth}\label{sec:mc_datadriven}

Table~\ref{tab:mc_datadriven} replaces the calibrated bandwidth by the rule~\eqref{eq:h1rule}, with the same seeds, so the oracle rows reproduce Tables~\ref{tab:mc_sbm4_beta}, \ref{tab:mc_sbm4_alpha}, \ref{tab:mc_bench_2000}, and~\ref{tab:mc_bench_1000}. In the baseline design, $\hat{h}_1$ lies between 0.11 and 0.33 times the separation gap in every replication, so the selected pairs are all of the same type, and from $n = 1000$ almost all same-type pairs are selected (recall at least 0.99): the pairwise estimator is within 0.008 of the oracle in bias and RMSE, and the class partition is exact in every replication. At $n = 500$ it returns 426 estimates, against 431 with the calibration, with an RMSE for $\beta$ of 0.715, against 0.677. In the weak-network design, whose gap is ten times smaller, the noise of the estimated distances is of the order of the gap: at $n = 500$ and 1000 the rule selects only 28 and 39 percent of the exact matches, and the pairwise estimator returns 33 and 322 estimates (none with the calibration at $n = 1000$). At $n = 2000$, every replication returns one, with an RMSE for $\beta$ of 0.304, against 0.362 for the calibration. Across the three sample sizes, the class partitions merge two classes in 3.6 to 7.8 percent of the weak-design replications. The constant matters in small samples: halving it leaves 352 estimates in the baseline design at $n = 500$ and none in the weak design at $n \le 1000$, and doubling it admits false matches in the weak design (precision 0.66 at $n = 500$). The rule works where the gap is large relative to the noise of the estimated distances, which is what the theory requires.

\begin{table}[!ht]
\centering
\caption{Data-driven network bandwidth: baseline and weak-network designs}\label{tab:mc_datadriven}
\smallskip
\footnotesize
\setlength{\tabcolsep}{4.5pt}
\begin{tabular}{ll c ccc ccc c}
\toprule
 & & & \multicolumn{3}{c}{$\beta$} & \multicolumn{3}{c}{$\alpha$} & \\
\cmidrule(lr){4-6} \cmidrule(lr){7-9}
$n$ & Estimator & Fits & Bias & RMSE & Cov. & Bias & RMSE & Cov. & Merge \\
\midrule
\multicolumn{10}{l}{\textit{Baseline}} \\[2pt]
500 & Pairwise & 426 & 0.256 & 0.715 & 0.965 & $-$0.060 & 1.396 & 0.955 & -- \\
 & Class & 430 & 0.127 & 0.549 & 0.963 & $-$0.049 & 1.132 & 0.991 & 0.0 \\
 & Oracle & 431 & 0.235 & 0.670 & 0.968 & $-$0.075 & 1.297 & 0.961 & -- \\
1000 & Pairwise & 500 & 0.152 & 0.457 & 0.950 & 0.003 & 0.880 & 0.934 & -- \\
 & Class & 500 & 0.084 & 0.381 & 0.974 & 0.002 & 0.757 & 0.964 & 0.0 \\
 & Oracle & 500 & 0.151 & 0.459 & 0.950 & 0.000 & 0.872 & 0.940 & -- \\
2000 & Pairwise & 500 & 0.070 & 0.263 & 0.942 & 0.016 & 0.546 & 0.938 & -- \\
 & Class & 500 & 0.042 & 0.222 & 0.960 & $-$0.005 & 0.477 & 0.952 & 0.0 \\
 & Oracle & 500 & 0.069 & 0.262 & 0.944 & 0.016 & 0.542 & 0.942 & -- \\
\addlinespace[4pt]
\multicolumn{10}{l}{\textit{Weak network}} \\[2pt]
500 & Pairwise & 33 & $-$0.087 & 0.850 & 0.970 & $-$0.328 & 1.373 & 1.000 & -- \\
 & Class & 45 & 0.012 & 0.498 & 1.000 & 0.015 & 1.185 & 0.978 & 5.4 \\
 & Oracle & 425 & 0.238 & 0.665 & 0.967 & $-$0.020 & 1.393 & 0.951 & -- \\
1000 & Pairwise & 322 & 0.278 & 0.787 & 0.941 & 0.117 & 1.657 & 0.922 & -- \\
 & Class & 341 & 0.100 & 0.423 & 0.947 & 0.094 & 0.881 & 0.962 & 7.8 \\
 & Oracle & 500 & 0.133 & 0.451 & 0.950 & 0.018 & 0.849 & 0.948 & -- \\
2000 & Pairwise & 500 & 0.083 & 0.304 & 0.934 & $-$0.051 & 0.629 & 0.938 & -- \\
 & Class & 500 & 0.043 & 0.239 & 0.942 & $-$0.001 & 0.516 & 0.944 & 3.6 \\
 & Oracle & 500 & 0.064 & 0.268 & 0.946 & $-$0.026 & 0.568 & 0.930 & -- \\
\bottomrule
\end{tabular}
\vspace{4pt}\par\raggedright\footnotesize \textit{Notes:} $R = 500$ replications per cell. Pairwise and Class: estimators~\eqref{eq:estimator} and~\eqref{eq:class_estimator} with $h_1 = \hat{h}_1$ of~\eqref{eq:h1rule}. Oracle: the pairwise estimator with population matches. $h_2 = 0.5$ (exact covariate stability). Fits: replications with an estimate. Bias, RMSE, and coverage of the nominal 95\% interval are computed over them, with the standard errors of Section~\ref{sec:mc_bandwidths}. Merge: percentage of replications in which a feasible class contains agents of different population classes.
\end{table}

\section{Empirical Application}\label{sec:empirical}

This section takes the method to longitudinal data on adolescent substance use in a Scottish secondary school, where friendship networks and behavior are observed at three survey waves. The exercise is a feasibility study rather than an estimation of the identified model. The data fall short of the model's requirements in ways documented below: the initial outcome is not observed, the covariate-stability condition cannot be imposed, the network is sparse and its observation process is not the dyadic one of Assumption~\ref{ass:sampling}, and at every localized bandwidth examined, the conditional likelihood is completely separated. The conditional estimates reported below are therefore descriptive alternative fits, and no inferential claims are attached to them.

\subsection{Data}

The data come from the Teenage Friends and Lifestyle Study (TFLS), a longitudinal survey of adolescent health behavior and social networks conducted at a secondary school in Glasgow, Scotland \citep{michell1996, pearson2000}.\footnote{The data are publicly available at \url{https://www.stats.ox.ac.uk/~snijders/siena/Glasgow_data.htm}.} The study tracked 160 students from an older cohort across three waves of data collection: wave~1 in February 1995, wave~2 in January 1996, and wave~3 in January 1997. Of these, 129 were present at all three waves. After restricting to students with complete outcome and covariate data across all waves, the estimation sample contains $n = 96$ students and $n \times T = 288$ student-wave observations. The networks used below are the subgraphs induced by these 96 students: nominations to and from the other 64 students of the cohort are discarded. This complete-case construction changes the observed neighborhoods and the distribution of partners against which codegrees are computed, and the paper makes no correction for it. A missing-node or nonresponse model would be needed to assess its consequences.

\label{loc:coevolution}These data have been used extensively to study the coevolution of friendship networks and substance use with stochastic actor-oriented models, which separate the selection of friends on behavior from the influence of friends' behavior \citep{steglich2006, steglich2010}. Those models specify network and behavior dynamics jointly and parametrically, with actors making decisions in continuous time between waves, and their object of interest is peer influence. The model of this paper asks a different question: it contains no influence of friends' outcomes, and it uses the network only as information about a latent own characteristic that confounds the effects of the covariates and of past behavior. The two approaches rest on different assumptions and address different parameters. Studies using these models find that friends' smoking may influence a student's smoking \citep{steglich2010, schaefer2012}. If it does and links are homophilous on lagged smoking, the omitted influence is correlated with the lagged outcome and is not common within a network class, so the estimate of $\alpha$ absorbs part of it. The state dependence estimated here is state dependence in a model without contagion.

 The setting motivates the framework, although, as this section documents, the data fall short of several of its requirements. Adolescent smoking is thought to be state dependent, through habit formation, nicotine dependence, and reduced social stigma \citep{fletcher2010, nakajima2007}. Unobserved characteristics such as risk attitude, impulsivity, or attachment to conventional norms also evolve during adolescence and affect both substance use and friendship formation: more risk-tolerant students may be more likely both to smoke and to befriend other risk-taking peers. The proposed estimator targets this confounding pathway between network structure and smoking.

\subsection{Variables}\label{sec:variables}

Table~\ref{tab:descriptive} reports summary statistics for the estimation sample across the three survey waves.

\begin{table}[!ht]
\centering
\caption{Descriptive statistics: Glasgow TFLS sample.}
\label{tab:descriptive}
\smallskip
\small
\setlength{\tabcolsep}{5pt}
\begin{tabular}{l cc cc cc c}
\toprule
 & \multicolumn{2}{c}{Wave 1 (1995)} & \multicolumn{2}{c}{Wave 2 (1996)} & \multicolumn{2}{c}{Wave 3 (1997)} & Range \\
\cmidrule(lr){2-3} \cmidrule(lr){4-5} \cmidrule(lr){6-7} \cmidrule(lr){8-8}
 & Mean & SD & Mean & SD & Mean & SD & all waves \\
\midrule
\multicolumn{8}{l}{\textit{Panel A: Outcome and covariates}} \\[4pt]
Smoking & 0.417 & 0.496 & 0.521 & 0.502 & 0.708 & 0.457 & 0--1 \\[2pt]
Pocket money & 9.67 & 7.51 & 11.70 & 8.36 & 16.21 & 10.40 & 0--50 \\[2pt]
Romantic relationship & 0.302 & 0.462 & 0.229 & 0.423 & 0.323 & 0.470 & 0--1 \\[2pt]
Female & 0.458 & 0.501 & 0.458 & 0.501 & 0.458 & 0.501 & 0--1 \\[2pt]
Age & 13.34 & 0.30 & 14.34 & 0.30 & 15.34 & 0.30 & 12.8--16.3 \\[2pt]
Alcohol & 2.781 & 1.038 & 3.062 & 1.064 & 3.479 & 0.894 & 1--5 \\[2pt]
Leisure & 2.616 & 0.308 & 2.748 & 0.260 & 2.736 & 0.279 & 2--4 \\[2pt]
Smoking at home & 0.594 & 0.494 & 0.594 & 0.494 & 0.594 & 0.494 & 0--1 \\[4pt]
\multicolumn{8}{l}{\textit{Panel B: Friend network}} \\[4pt]
Density & 0.0408 & & 0.0390 & & 0.0373 & & \\[2pt]
Degree & 3.88 & 2.17 & 3.71 & 1.90 & 3.54 & 1.80 & 0--10 \\[2pt]
\midrule
Students & \multicolumn{7}{c}{$96$} \\
Waves & \multicolumn{7}{c}{$3$} \\
Observations ($n \times T$) & \multicolumn{7}{c}{$288$} \\
\bottomrule
\end{tabular}

\smallskip
\par\raggedright
{\footnotesize The range column reports the minimum and maximum over all three waves. Smoking equals one if the student reports any tobacco use (occasional or regular) or any cannabis use (tried once, occasional, or regular). Pocket money is the monthly amount in British pounds. Romantic relationship equals one if the student reports being in a relationship. Female equals one for girls. Age is recorded in years with decimal precision (computed from month of birth), measured at wave 1, and incremented by one year per wave. Alcohol is an ordinal scale: 1 (non-drinker), 2 (once or twice a year), 3 (once a month), 4 (once a week), and 5 (more than once a week). Leisure is the average across fifteen leisure activity items, each coded from 1 (most days) to 4 (less often or never), so a higher value means less frequent activity. Smoking at home equals one if the student reports that someone smokes at home. It is coded as the first non-missing report across the three waves and is therefore time-invariant by construction. Networks are symmetrized using the OR rule. Degree is the number of undirected links per student.}
\end{table}

The binary outcome $y_{it}$ equals one if student $i$ reports any tobacco use (occasional or regular) or any cannabis use (tried once, occasional, or regular) at wave $t$ and zero otherwise. The composite is hard to interpret in a dynamic model: the cannabis category ``tried once'' can record past initiation rather than current use, so a value of one can persist mechanically. The estimand is therefore persistence in the recorded composite status, and the exercise makes no claim to isolate addiction or habit formation. The tobacco-only outcome is reported as a robustness check in Table~\ref{tab:robust}. Prevalence in the estimation sample rises from $41.7\%$ at wave~1 to $52.1\%$ at wave~2 and $70.8\%$ at wave~3, consistent with rapid uptake during early adolescence.

The identification strategy requires $T \geq 3$ modeled periods plus an initial condition $y_{i0}$, which I set to zero for all students, using the three waves as $t = 1, 2, 3$. Students are 13 at wave~1, an age at which regular smoking is uncommon, but the assumption is substantive: $41.7\%$ already report some use at wave~1, so imposing universal non-use one year earlier attributes all wave-1 use to initiation within that year. Because $y_{i0}$ enters the state dependence regressor directly, misclassification of the initial condition can affect $\hat\alpha$. This is the classical initial conditions problem of \citet{heckman1981b}. Assumption~\ref{ass:sampling} treats the initial outcome as observed and independent of the later outcome shocks, and neither property can be checked here.

Four time-varying covariates enter the proposed estimator. \emph{Pocket money} is self-reported monthly disposable income in pounds, rising on average from $9.67$ at wave~1 to $16.21$ at wave~3, and captures the capacity to buy tobacco or cannabis. The survey's ``don't know'' code ($-1$) is treated as missing, while zero is a valid response. \emph{Romantic relationship} indicates being in a relationship, between $23\%$ and $32\%$ across waves, and captures social involvement and maturity. \emph{Alcohol} is an ordinal five-point scale from 1 (non-drinker) to 5 (more than once a week), rising from $2.78$ to $3.48$. Drinking and smoking co-occur strongly in adolescent populations and may reflect common risk factors. \emph{Leisure} is the mean reported frequency across the fifteen out-of-school leisure items in the questionnaire (sports, reading, cinema, and so on), not a validated measure of structured activity. Each item is coded from 1 (most days) to 4 (less often or never), so that a higher value means less frequent participation, and the mean is stable between $2.62$ and $2.75$ across waves. Pocket money and leisure take $47$ and $27$ distinct values in the estimation sample, with $40$ of the $288$ pocket-money observations non-integer, so both are effectively continuous rather than recorded on a coarse discrete scale.

Assumption~\ref{ass:sampling}(iv) requires the entire covariate path to be independent of the outcome shocks. This is implausible for several of these covariates during adolescence: alcohol use, romantic involvement, leisure activity, and pocket money can all respond to past substance use, and the latent characteristic may respond to it as well. Such feedback violates the strict exogeneity on which the path factorization of Lemma~\ref{lem:ratio} rests. It is a substantive identification concern, and a larger sample would not resolve it.

Time variation is required by the conditional likelihood approach: the pairwise differencing that eliminates $\lambda(w_{it})$ also differences the covariates, so any time-invariant characteristic drops out. The naive and controls specifications of Appendix~\ref{sec:appendix_glasgow}, which do not rely on that differencing, add three covariates the proposed estimator cannot use: \emph{female} ($45.8\%$ of the sample), \emph{age} (mean $13.34$ at wave~1, incremented by one year per wave, so it cancels exactly under differencing), and \emph{smoking at home}, an indicator for whether someone smokes at the student's home, coded from the first non-missing report.

\subsection{Network Construction}\label{sec:network_construction}

At each wave, students nominated up to six friends within the cohort, with nominations coded as 1 (best friend) or 2 (friend). The analysis uses the ``any friend'' network, which counts a link whenever at least one of the two students nominated the other at level 1 or 2. The directed adjacency matrix is symmetrized using the OR rule: $D_{ijt} = \max(D_{ijt}^{\mathrm{dir}}, D_{jit}^{\mathrm{dir}})$. The resulting undirected networks have densities of approximately 0.04 at each wave, with an average degree around 3.5 to 3.9 links per student.

The observation process departs from the network model of equation~\eqref{eq:network} in three ways. First, the model takes links to be undirected and formed from symmetric dyadic shocks that are conditionally independent across dyads (Assumption~\ref{ass:sampling}(ii)). The six-nomination cap violates this independence: one outgoing nomination constrains another, so the links of a student are dependent even given the latent characteristics, and the problem is not only that high-degree students' neighborhoods are truncated. Restricting attention to reciprocated nominations would not remove this dependence and would make the network sparser. Symmetrizing directed nominations by the OR rule does not by itself break dyadic independence, but it changes what a link records: an unreciprocated nomination, whether or not it reflects a friendship, becomes a link. Second, the asymptotic theory of Section~\ref{sec:asymptotics} is that of a fixed dense graphon, in which degrees grow in proportion to the number of students, and the bounded-degree and capped-nomination features of these data are not covered by that asymptotic experiment. Average degree here is below four in the any-friend network and below one in the best-friend network (0.81, 0.67, and 0.81 links per student across the three waves, so its codegree distances rest on very little linking information), so the approximation that underlies the codegree theory is poor, and a larger school with degrees that remain bounded would not by itself deliver more informative codegrees. Third, the networks are induced subgraphs of the complete-case sample, as noted above. None of these problems enters the unlocalized conditional columns of Table~\ref{tab:any_friend_network}, which do not use the network. All of them enter the localized estimates of Table~\ref{tab:bandwidth}, where the sparsity is part of why so few pairs fall inside the kernel window.

\subsection{Feasibility of the Estimator}\label{sec:localized}

The outcome equation takes the form
\[
y_{it} = \ind\{X_{it}'\beta + \alpha\, y_{it-1} + \lambda(w_{it}) + \mu_i - \varepsilon_{it} \geq 0\}, \quad t = 1, 2, 3,
\]
where $X_{it}$ contains pocket money, romantic relationship, alcohol, and leisure; $y_{it-1}$ captures state dependence in smoking; $\mu_i$ is a student fixed effect, which may stand for a physiological susceptibility to nicotine or for persistent features of the home environment; and $\lambda(w_{it})$ is an unknown function of the latent social characteristic $w_{it}$ that also governs friendship formation. The unobserved characteristic $w_{it}$ can be interpreted as an evolving index of risk attitude or social orientation, which affects both substance use and the tendency to befriend other risk-taking students.

Appendix~\ref{sec:appendix_glasgow} reports five descriptive fits on these data: a naive dynamic logit, a dynamic logit with peer-average controls, a within-agent conditional fit in the spirit of \citet{honore2000panel}, an unlocalized cross-agent pairwise fit, and an unlocalized class-conditional fit. \label{loc:departures}Two departures from the estimator of Section~\ref{sec:estimation} are forced by these data, and together they mean that the conditional fits do not estimate the identified model. First, the covariate-stability kernel $K_2$ is infeasible: none of the 22 switching students has stable covariates between waves 2 and 3, and any meaningful bandwidth leaves at most one weighted pair, so the conditional fits omit it. Without it, the log odds of equation~\eqref{eq:worked} retain the term $\Gamma(c_{i3}) - \Gamma(c_{i2})$, which is bounded in absolute value by $|c_{i3} - c_{i2}|$ but is not shown to be negligible in these data and which depends on the level of the individual effect and of $\lambda$ through $c_{i2}$. Second, the unlocalized pairwise fit sets $h_1 = \infty$, so that all 96 discordant lag-matched pairs receive the same weight. The largest aggregated codegree distance among these pairs is $1.8 \times 10^{-4}$, so a bandwidth that is large relative to it, such as $h_1 = n^{-1/5} \approx 0.40$, gives weights that differ from equal weights by less than one part in a million and the same fit, whereas bandwidths of the order of the distances localize. The localized fits, with quantile bandwidths, are those of Table~\ref{tab:bandwidth}. The unlocalized fit therefore does not use the network, is numerically identical across network definitions, and does not remove the difference between the two agents' common movements, which need not vanish without matching. At the same unlocalized bandwidth, the feasible classes of the class-conditional estimator of Section~\ref{sec:class_estimator} are the lag strata, and since $y_{i0} = 0$ for every student there is a single class of 22 switchers: the class-conditional fit is then the within-agent conditional logit with one common movement of $\lambda$ for all students, conditioned out, and it and the unlocalized pairwise fit are the full and the pairwise conditional likelihoods of the same class. With the stability kernel omitted, the $\Gamma$ remainder generally still depends on the level of the individual effect and of $\lambda$, so none of the conditional fits delivers the cancellation of Lemma~\ref{lem:ratio}, and differences between their coefficients are differences between objectives.

Table~\ref{tab:bandwidth} reports what happens when the codegree kernel actually binds: the bandwidth is set to the $q$-th quantile of the finite aggregated codegree distances, so that roughly a fraction $q$ of the discordant lag-matched pairs receives positive weight. This is the network-localized version of the pairwise fit, still omitting the $K_2$ condition of Section~\ref{sec:estimation}, and the table shows what it can and cannot do at $n = 96$. The central diagnostic finding is that at every localized bandwidth reported in the table (the 0.10, 0.20, and 0.30 quantiles with the any-friend network and the 0.10 quantile with the best-friend network, but not the 0.05 quantile, which leaves only five pairs), the weighted conditional likelihood is completely separated \citep{albert1984}: a unit-norm direction orders every weighted pair correctly (Panel A), and at the attempted maximum the share of pair-level fitted probabilities outside $[0.01, 0.99]$ is $1.000$ in every column, with maximal fitted indices between 38 and 67. Under complete separation the likelihood has no maximizer, its supremum is approached only as $\|\theta\| \to \infty$, the point at which a quasi-Newton optimizer stops is arbitrary, and the associated standard errors invert a numerically singular information matrix. The table therefore does not report unpenalized point estimates, which do not exist as maximizers. It reports the separation diagnostics together with the Firth-penalized estimates \citep{firth1993}, which provide finite estimates for the penalized specifications reported here. In the localized fits the penalized criterion is the kernel-weighted log likelihood plus one half of the log determinant of its weighted information, with the raw kernel weights $K(\hat{\bar{\delta}}_{ij}^2/h_1)$, not weights normalized to sum to one or to the number of pairs. Because the weight of the penalty relative to the fit depends on this scale, penalized coefficients from objectives with different weight scales, such as the localized and unlocalized fits, are not directly comparable.

\begin{table}[!ht]
\centering
\caption{Network-localized estimates: separation diagnostics and penalized fits.}
\label{tab:bandwidth}
\smallskip
\begin{tabular}{l cccc}
\toprule
 & \multicolumn{3}{c}{Any friend} & Best friend \\
\cmidrule(lr){2-4} \cmidrule(lr){5-5}
 & $q = 0.10$ & $q = 0.20$ & $q = 0.30$ & $q = 0.10$ \\
\midrule
\multicolumn{5}{l}{\textit{Panel A: Separation diagnostics (unpenalized likelihood)}} \\[2pt]
Share of fitted prob.\ outside $[0.01, 0.99]$ & $1.000$ & $1.000$ & $1.000$ & $1.000$ \\
Max $|$fitted index$|$ at optimizer stop & $38.2$ & $67.2$ & $51.7$ & $53.9$ \\
Max.\ min.\ signed margin ($\|b\|=1$) & $1.540$ & $0.566$ & $0.483$ & $0.729$ \\
MLE exists & No & No & No & No \\[4pt]
\multicolumn{5}{l}{\textit{Panel B: Firth-penalized estimates}} \\[2pt]
Pocket money & $\underset{(0.047)}{-0.005}$ & $\underset{(0.040)}{-0.009}$ & $\underset{(0.038)}{-0.004}$ & $\underset{(0.193)}{0.009}$ \\[6pt]
Romantic relationship & $\underset{(2.481)}{-0.658}$ & $\underset{(1.765)}{-0.388}$ & $\underset{(1.628)}{-0.786}$ & $\underset{(1.061)}{0.126}$ \\[6pt]
Alcohol & $\underset{(0.513)}{0.317}$ & $\underset{(0.413)}{0.625}$ & $\underset{(0.397)}{0.824}$ & $\underset{(0.600)}{0.847}$ \\[6pt]
Leisure & $\underset{(2.797)}{0.867}$ & $\underset{(1.975)}{0.820}$ & $\underset{(1.820)}{1.502}$ & $\underset{(3.188)}{0.179}$ \\[6pt]
Lagged smoking ($\hat\alpha$) & $\underset{(1.514)}{0.176}$ & $\underset{(0.996)}{-0.217}$ & $\underset{(0.980)}{-0.403}$ & $\underset{(1.184)}{0.876}$ \\[6pt]
\midrule
Pairs with positive weight & $9$ & $21$ & $26$ & $14$ \\
Pair-weight ESS & $7.0$ & $17.1$ & $23.0$ & $14.0$ \\
\bottomrule
\end{tabular}

\smallskip
{\footnotesize $h_1$: $q$-th quantile of the finite aggregated lag-stratified codegree distances. $K_2$ is omitted (see text). ESS $= (\sum_p w_p)^2/\sum_p w_p^2$ over the weighted pairs, not a count of independent students. Panel A: evaluated where the BFGS optimizer stops on the unpenalized likelihood. The signed-margin row is the largest minimum margin of a unit-norm direction over the pair rows, oriented so that the realized ordering is positive, and a positive value certifies complete separation (see text). Panel B: Firth-penalized estimates \citep{firth1993}, with the unpenalized sandwich at the penalized estimates in parentheses, for scale only. No stars: the effective samples are too small for asymptotic approximations.}
\end{table}

The table is reported in full because its message is part of the paper's argument. At the localized bandwidths of the table, the effective sample size drops to between 7 and 23, and the likelihoods are completely separated: with so few weighted pairs, some linear combination of the regressors perfectly predicts the discordant ordering, the MLE does not exist, and only penalized estimates are finite. The Firth-penalized estimates are correspondingly modest and unstable across bandwidths (only alcohol and leisure keep a consistent sign across all four columns, and the state dependence coefficient changes sign across them), with standard errors that make none of the localized estimates informative. The conclusion, backed by the separation diagnostics rather than by erratic point estimates, is that $n = 96$ with a sparse, capped network does not support the network-localized estimator at any of the bandwidths examined. \label{loc:classlocal}The class-conditional estimator of Section~\ref{sec:class_estimator} does no better at the same bandwidths: at each of them only one feasible class contains both orderings, with between 9 and 15 of the 22 switchers, and its likelihood with the class intercept is completely separated, with unit-norm margins between 0.20 and 0.34. The Glasgow data can be used to compute the unlocalized conditional fits of Appendix~\ref{sec:appendix_glasgow}, which omit both the covariate-stability condition and the network matching, but not to implement the codegree localization that the method requires. The application therefore documents the data requirements of the method rather than its performance, and nothing in it validates the procedure. Two of the four covariates (romantic relationship and alcohol) are recorded on coarse discrete scales, but pocket money and leisure are effectively continuous and, more to the point, no switching student keeps them unchanged between waves 2 and 3, so the atom of exact covariate stability that the finite-type regime of Theorem~\ref{thm:normality} requires is absent in these data, even setting the network aside. And a suitable dataset would combine at least four waves, so that the initial outcome is observed, covariates recorded on discrete scales or with enough stable observations to impose the covariate-stability condition, and a network observed without a nomination cap whose degrees are large enough for codegrees to be informative. Such data, collected at close and evenly spaced dates around a scheduled transition as described in Appendix~\ref{sec:appendix_window}, could address the obstacles documented here, but they would not by themselves validate the method: the overlap of network types among switchers, the common movement of Assumption~\ref{ass:persistence}, the strict exogeneity of the covariates, and the rank condition on the selected pairs would still have to be assessed, and none of them is a matter of sample size alone.


\section{Conclusion}\label{sec:conclusion}

This paper studies a dynamic logit model in which a binary outcome depends on covariates, the past outcome, a fixed effect, and the social influence function, an unknown function of a time-varying unobserved characteristic that also drives link formation. Covariate effects and state dependence are identified when the change of that function between two waves is common to agents of the same network type and is followed by a stable wave: conditioning on an agent's own switch removes the fixed effect, and comparing agents of the same network type removes the change.

The pairwise conditional likelihood estimator weights pairs of switchers by their estimated codegree distance and the stability of their covariates. It is consistent, including with continuous types, and root-$n$ asymptotically normal when the characteristic takes finitely many values and covariate changes have an atom at zero. With three waves, the same holds for a bandwidth computed from the first-wave network. In that case, a second estimator conditions on network classes of switchers. It is root-$n$ asymptotically normal with one transition period and, when $T = 3$, at least as efficient asymptotically as the pairwise estimator, by an argument sketched in the paper.

In simulations that satisfy the assumptions, the pairwise estimator agrees with an oracle that knows the true matches, its bias falls with the sample size, and its intervals have close to nominal coverage among successful fits, while pooled logits and the within-agent estimator remain biased when the characteristic moves. The class estimator reduces the root mean squared error of the pairwise estimator by 6 to 20 percent in the two main designs. The method fails when the movement is not common to network types or when the network is sparse. In the Glasgow Teenage Friends and Lifestyle Study, the network-localized likelihood is completely separated at every bandwidth examined, so the application documents the data the method requires rather than its performance.

Three extensions remain open: a distribution theory for continuous types, which must handle the dependence created by shared links, the sampling of latent types, and heterogeneous match probabilities; feedback from outcome shocks into later covariates and into the latent characteristic; and directed, mismeasured, or multiple networks.

\appendix
\section{Appendix}\label{sec:appendix}

This appendix collects the proofs of the main results. I begin with the proof of the within-agent likelihood ratio and a detailed derivation of the identification argument for $T = 3$, then establish a lemma on the empirical codegree distance and use it to complete the proofs of Proposition~\ref{prop:bridge}, Theorem~\ref{thm:identification}, and Corollary~\ref{cor:lambda}. I then prove consistency (Theorem~\ref{thm:consistency}) and asymptotic normality in finite-type designs (Theorem~\ref{thm:normality}), establish a lemma on conditional sums of Bernoulli variables, and prove Proposition~\ref{prop:class} and Corollary~\ref{cor:bandwidth}. What is and is not identified about $\lambda$ is discussed in Appendix~\ref{sec:appendix_lambda}, the extension to ordered outcomes in Appendix~\ref{sec:appendix_ordered}, the observation window of the smoking example in Appendix~\ref{sec:appendix_window}, and descriptive fits for the Glasgow data in Appendix~\ref{sec:appendix_glasgow}. The supplementary tables are in Appendix~\ref{sec:appendix_figures}.

\subsection{Proof of Lemma~\ref{lem:ratio}}\label{sec:proof_lemratio}

\begin{proof}
By Assumption~\ref{ass:sampling}(iv), the shocks $(\varepsilon_{i1}, \ldots, \varepsilon_{iT})$ are i.i.d.\ logistic and independent of $\chi_i = (y_{i0}, X_i, w_i, \mu_i)$. Since $y_{it} = \ind\{X_{it}'\beta + \alpha y_{i,t-1} + \lambda(w_{it}) + \mu_i - \varepsilon_{it} \ge 0\}$, conditional on $\chi_i$ the outcomes form a Markov chain with transition probabilities
\[
\Prob(y_{it} = 1 \mid y_{i0}, \ldots, y_{i,t-1}, X_i, w_i, \mu_i) = F\big(X_{it}'\beta + \alpha y_{i,t-1} + \lambda(w_{it}) + \mu_i\big),
\]
and conditioning on the entire covariate and latent paths, including their future values, does not alter these probabilities. Hence the probability of a path factorizes into per-period transition probabilities, with the initial outcome held fixed. Write $\pi_t(y \mid y')$ for the probability of outcome $y$ at $t$ given lag $y'$. For path $A_s$,
\begin{align*}
\Prob(y_i \in A_s \mid \chi_i) &= \left[\prod_{\substack{t=1 \\ t \neq s, s+1, s+2}}^T \pi_t(y_{it} \mid y_{i,t-1})\right]
\big(1 - F(X_{is}'\beta + \alpha y_{i,s-1} + \lambda(w_{is}) + \mu_i)\big)\\
&\quad \times F\big(X_{i,s+1}'\beta + \lambda(w_{i,s+1}) + \mu_i\big)\, \pi_{s+2}(y_{i,s+2} \mid 1),
\end{align*}
where the contribution at period $s$ reflects $y_{is} = 0$ with lag $y_{i,s-1}$, the contribution at $s+1$ reflects $y_{i,s+1} = 1$ with lag $0$, and the contribution at $s+2$ reflects the outcome $y_{i,s+2}$ with lag $1$. For path $B_s$,
\begin{align*}
\Prob(y_i \in B_s \mid \chi_i) &= \left[\prod_{\substack{t=1 \\ t \neq s, s+1, s+2}}^T \pi_t(y_{it} \mid y_{i,t-1})\right]
F\big(X_{is}'\beta + \alpha y_{i,s-1} + \lambda(w_{is}) + \mu_i\big)\\
&\quad \times \big(1 - F(X_{i,s+1}'\beta + \alpha + \lambda(w_{i,s+1}) + \mu_i)\big)\, \pi_{s+2}(y_{i,s+2} \mid 0).
\end{align*}
The products over $t \notin \{s, s+1, s+2\}$ coincide, because the two paths agree outside periods $s$ and $s+1$ and the lags entering those periods are the same. Let $\gamma_s = X_{is}'\beta + \alpha y_{i,s-1} + \lambda(w_{is}) + \mu_i$, $c = X_{i,s+1}'\beta + \lambda(w_{i,s+1}) + \mu_i$, and $c' = X_{i,s+2}'\beta + \lambda(w_{i,s+2}) + \mu_i$. Using $F(x)/(1 - F(x)) = e^x$, the ratio of the period-$s$ terms is $e^{\gamma_s}$, and the ratio of the terms at $s+1$ and $s+2$ is
\[
\frac{1 + e^{c}}{e^{c}(1 + e^{c+\alpha})}\cdot \frac{e^{y_{i,s+2}c'}\,(1 + e^{c'+\alpha})}{(1 + e^{c'})\, e^{y_{i,s+2}(c'+\alpha)}} = \exp\big[-c - \alpha y_{i,s+2} + \Gamma(c') - \Gamma(c)\big],
\]
with $\Gamma(x) = \log(1 + e^{x+\alpha}) - \log(1 + e^{x})$. Under $X_{i,s+2} = X_{i,s+1}$ and $\lambda(w_{i,s+2}) = \lambda(w_{i,s+1})$, $c' = c$ and the $\Gamma$ terms cancel. Combining and noting that $\gamma_s - c$ does not involve $\mu_i$,
\[
\frac{\Prob(y_i \in B_s \mid \chi_i)}{\Prob(y_i \in A_s \mid \chi_i)} = \exp\big[(X_{is} - X_{i,s+1})'\beta + \alpha(y_{i,s-1} - y_{i,s+2}) + \lambda(w_{is}) - \lambda(w_{i,s+1})\big]. \qedhere
\]
\end{proof}

The display before the last step also gives the exact log odds without any stability condition, which is equation~\eqref{eq:worked} for $T = 3$. Because $\Gamma'(x) = F(x + \alpha) - F(x) \in (-1, 1)$, the remainder satisfies $|\Gamma(c') - \Gamma(c)| \le |c' - c|$. This bound is used in the proof of Theorem~\ref{thm:consistency}.

\subsection{Derivation for \texorpdfstring{$T = 3$}{T = 3}}

For concreteness, I write out the argument of Lemma~\ref{lem:identity} and Theorem~\ref{thm:identification} for $T = 3$, so that $\mathcal{S} = \{1\}$. The paths are
\[
A = \{y_{i0}, 0, 1, y_{i3}\}, \qquad B = \{y_{i0}, 1, 0, y_{i3}\},
\]
with $y_{i0}$ and $y_{i3}$ held at their realized values. Write $c_{it} = X_{it}'\beta + \lambda(w_{it}) + \mu_i$. Conditional on $\chi_i$,
\begin{align*}
\Prob(y_i \in A \mid \chi_i) &= \frac{1}{1 + e^{c_{i1} + \alpha y_{i0}}} \cdot \frac{e^{c_{i2}}}{1 + e^{c_{i2}}} \cdot \frac{e^{y_{i3}(c_{i3} + \alpha)}}{1 + e^{c_{i3} + \alpha}},\\
\Prob(y_i \in B \mid \chi_i) &= \frac{e^{c_{i1} + \alpha y_{i0}}}{1 + e^{c_{i1} + \alpha y_{i0}}} \cdot \frac{1}{1 + e^{c_{i2} + \alpha}} \cdot \frac{e^{y_{i3}c_{i3}}}{1 + e^{c_{i3}}}.
\end{align*}
Taking the ratio and using $F(x)/(1 - F(x)) = e^x$ gives equation~\eqref{eq:worked}: the agent effect cancels in $c_{i1} - c_{i2} = (X_{i1} - X_{i2})'\beta + \lambda(w_{i1}) - \lambda(w_{i2})$ and enters only the bracket $\Gamma(c_{i3}) - \Gamma(c_{i2})$, through the levels. If $X_{i3} = X_{i2}$ and $\lambda(w_{i3}) = \lambda(w_{i2})$, then $c_{i3} = c_{i2}$ and Lemma~\ref{lem:ratio} gives
\[
\frac{\Prob(y_i \in B \mid \chi_i)}{\Prob(y_i \in A \mid \chi_i)} = \exp\big[(X_{i1} - X_{i2})'\beta + \alpha(y_{i0} - y_{i3}) + \lambda(w_{i1}) - \lambda(w_{i2})\big].
\]
Take a second agent $j$ with the same stability conditions and with $\lambda(w_{i1}) - \lambda(w_{i2}) = \lambda(w_{j1}) - \lambda(w_{j2})$, which is the event $\mathcal{C}_{ij1} \cap \mathcal{X}_{ij1}$. The levels $\lambda(w_{i1})$ and $\lambda(w_{j1})$ need not be equal, and neither need the agent effects $\mu_i$ and $\mu_j$. Since the two agents' paths are independent given their primitives, with $\Delta X_a = X_{a1} - X_{a2}$,
\[
\frac{\Prob(y_i \in B \mid \cdot)\, \Prob(y_j \in A \mid \cdot)}{\Prob(y_i \in A \mid \cdot)\, \Prob(y_j \in B \mid \cdot)} = \exp\big[(\Delta X_i - \Delta X_j)'\beta + \alpha(y_{i0} - y_{i3} - y_{j0} + y_{j3})\big],
\]
because the two changes of $\lambda$ cancel. Given the discordance event $\mathcal{D}_{ij1}$ and the outcomes at periods $0$ and $3$, only the configurations $A \times B$ and $B \times A$ are possible, so
\begin{align*}
\Prob\big(y_i \in A, y_j \in B \,\big|\, \mathcal{F}_{ij1}, w_i, w_j, \mathcal{D}_{ij1}\big) &= F\big[(\Delta X_j - \Delta X_i)'\beta + \alpha(y_{j0} - y_{j3} - y_{i0} + y_{i3})\big]\\
&= F\big(V_{ij1}(\beta,\alpha)\big).
\end{align*}
Under Assumption~\ref{ass:persistence}, the event $\mathcal{C}_{ij1}$ is implied by network-type equivalence at period 1 in the stratum of the shared initial outcome together with $d_2 \equiv 0$, which is the content of Theorem~\ref{thm:identification} for $T = 3$.

\subsection{The Empirical Codegree Distance}

Throughout the remainder of the appendix, fix a period $t$ and write $\xi_l = (w_{lt}, y_{l,t-1})$ for agent $l$'s type at $t$, with distribution $G_t$. By Assumption~\ref{ass:sampling}, the $\xi_l$ are i.i.d.\ across agents, and, conditional on $\xi = (\xi_1, \ldots, \xi_n)$, the links $D_{klt}$ with $k < l$ are independent with means $f_{klt} = f(w_{kt}, w_{lt};\, y_{k,t-1}, y_{l,t-1})$, because the link shocks are independent of the outcome shocks and of the primitives. The diagonal of the adjacency matrix is zero. For a pair $(i,j)$ with $y_{i,t-1} = y_{j,t-1} = y$, write $g_l = f_{ilt} - f_{jlt}$ and let $p_{ikt} = \int f(w_{it}, \tau;\, y, y')\, f(w_{kt}, \tau;\, y_{k,t-1}, y')\, dG_t(\tau, y')$ denote the population codegree, so that $\delta_{ijt}^2 = \E_k[(p_{ikt} - p_{jkt})^2]$, where $\E_k$ integrates over the type of an independent agent $k$.

\begin{lemma}[Empirical codegree distance]\label{lem:codegree_consist}
Suppose Assumption~\ref{ass:sampling} holds. For each period $t$:
\begin{enumerate}[label=(\alph*)]
\item For any pair $(i,j)$ with $y_{i,t-1} = y_{j,t-1}$, $\hat{\delta}_{ijt}^2 \plim \delta_{ijt}^{2}$.
\item If $f$ is continuous, $\rho_{ijt} = 0$ implies $\delta_{ijt} = 0$.
\item Suppose that $f$ is continuous or that $G_t$ has finite support, and let $(w_{it}, y)$ and $(w_{jt}, y)$ lie in the support of $G_t$. Then $\delta_{ijt} = 0$ if and only if $f(w_{it}, \tau;\, y, y') = f(w_{jt}, \tau;\, y, y')$ for every $(\tau, y')$ in the support of $G_t$. Under Assumption~\ref{ass:smooth}, $\delta_{ijt} = 0$ if and only if $\rho_{ijt} = 0$.
\item Uniformly over pairs,
\[
\max_{i \neq j}\, \big|\hat{\delta}_{ijt}^2 - \delta_{ijt}^{2}\big| \,\ind\{y_{i,t-1} = y_{j,t-1}\} = O_p\big((\log n/n)^{1/2}\big),
\]
and, for pairs that are equivalent on the support, $\max_{i \neq j}\, \hat{\delta}_{ijt}^2\, \ind\{g_l = 0 \text{ for all } l\} = O_p\big(\log n/n\big)$. The constants in these bounds do not depend on the law of the primitives, so the bounds hold uniformly along the triangular array of Assumption~\ref{ass:drift}.
\end{enumerate}
\end{lemma}

\begin{proof}
\textit{Decomposition.} Fix a pair $(i,j)$ with a shared lag and write $\hat{C}_k = \hat{c}_{ikt} - \hat{c}_{jkt} = n^{-1}\sum_{l} D_{klt}(D_{ilt} - D_{jlt})$, so that $\hat{\delta}_{ijt}^2 = n^{-1}\sum_k \hat{C}_k^2$. The two terms $k \in \{i,j\}$ share a dyad with the inner products and are bounded by one, so they contribute at most $2/n$. For $k \notin \{i,j\}$, the inner terms with $l \in \{i, j, k\}$ are bounded by one and contribute at most $2/n$ to $\hat{C}_k$. The remaining inner terms, $l \notin \{i,j,k\}$, involve the distinct dyads $(k,l)$, $(i,l)$, and $(j,l)$, so, conditional on $\xi$, they are independent across $l$ with means $f_{klt}\, g_l$. Hence
\[
\hat{C}_k = a_k + \nu_k + r_k, \qquad a_k = \frac{1}{n}\sum_{l \notin \{i,j,k\}} f_{klt}\, g_l, \qquad \nu_k = \frac{1}{n}\sum_{l \notin \{i,j,k\}} \big[D_{klt}(D_{ilt} - D_{jlt}) - f_{klt}\, g_l\big],
\]
with $|r_k| \le 2/n$, $|a_k| \le 1$, and $\E[\nu_k \mid \xi] = 0$, and
\[
\hat{\delta}_{ijt}^2 = \underbrace{\frac{1}{n}\sum_{k \notin \{i,j\}} a_k^2}_{\mathrm{I}} + \underbrace{\frac{2}{n}\sum_{k \notin \{i,j\}} a_k \nu_k}_{\mathrm{II}} + \underbrace{\frac{1}{n}\sum_{k \notin \{i,j\}} \nu_k^2}_{\mathrm{III}} + R_{ij}, \qquad |R_{ij}| \le \frac{C}{n}.
\]
The inner averages $\nu_k$ share observations across $k$, since each involves agent $i$'s and agent $j$'s links with every $l$. The argument below never uses independence across $k$, only bounds that hold simultaneously for all $k$. $\mathrm{I}$ carries the sampling variation of the latent types, and $\mathrm{II}$ and $\mathrm{III}$ carry the link noise.

\textit{Link noise.} Conditional on $\xi$, $\nu_k$ is an average of at most $n$ independent, centered variables with range at most $2$, so Hoeffding's inequality gives $\Prob(|\nu_k| > x \mid \xi) \le 2\exp(-n x^2/8)$ for $n$ large. A union bound over the at most $n^3$ triples $(i,j,k)$ with $x = (32 \log n/n)^{1/2}$ gives $\max_{i,j,k}|\nu_k| = O_p((\log n/n)^{1/2})$, with constants that do not depend on $\xi$ or on the law of the primitives. Therefore, uniformly over pairs, $|\mathrm{II}| \le 2\max_k|\nu_k| = O_p((\log n/n)^{1/2})$ and $0 \le \mathrm{III} \le \max_k \nu_k^2 = O_p(\log n/n)$.

\textit{Sampling of latent types.} Expanding the square,
\[
\mathrm{I} = \frac{1}{n^3}\sum_{k \notin \{i,j\}}\ \sum_{l, l' \notin \{i,j,k\}} f_{klt}\, f_{kl't}\, g_l\, g_{l'}.
\]
Conditional on $(\xi_i, \xi_j)$, this is, up to terms with a repeated index that contribute $O(1/n)$, a U-statistic of order three in the i.i.d.\ types of the other agents with a kernel bounded by one. Its mean is
\[
\E\big[f_{klt} f_{kl't} g_l g_{l'} \,\big|\, \xi_i, \xi_j\big] = \E_k\Big[\big(\E_l[f_{klt}\, g_l]\big)^2\Big] = \E_k\big[(p_{ikt} - p_{jkt})^2\big] = \delta_{ijt}^2,
\]
where the middle equality uses $\E_l[f_{klt}\, f_{ilt}] = p_{ikt}$, which is the definition of $p_{ikt}$. The joint symmetry of $f$ (Assumption~\ref{ass:sampling}(vi)) is used in the decomposition above, where it gives $D_{klt} = D_{lkt}$ the mean $f_{klt}$ also when $l < k$. Hoeffding's inequality for U-statistics \citep{hoeffding1963} gives, conditional on $(\xi_i, \xi_j)$, $\Prob(|\mathrm{I} - \delta_{ijt}^2| > x + C/n) \le 2\exp(-c\, n x^2)$ with universal constants, and a union bound over the at most $n^2$ pairs gives $\max_{i,j}|\mathrm{I} - \delta_{ijt}^2| = O_p((\log n/n)^{1/2})$. Combining the three terms proves the first bound in part (d), which implies part (a).

\textit{Exact matches.} If the two agents are equivalent on the support, then $g_l = 0$ for every agent $l$, because every agent's type lies in the support of $G_t$ almost surely. Then $a_k = 0$ for every $k$, so $\mathrm{I} = \mathrm{II} = 0$ exactly, and $\hat{\delta}_{ijt}^2 = \mathrm{III} + R_{ij} = O_p(\log n/n)$ uniformly over such pairs. The term $R_{ij}$ is non-negative up to $O(1/n)$ and generates an upward bias of order $1/n$, which is within the stated rate.

\textit{Part (b).} If $\rho_{ijt} = 0$, the link functions of the two agents coincide Lebesgue-almost everywhere, hence everywhere by continuity, so $f_{ilt} = f_{jlt}$ for every partner, $p_{ikt} = p_{jkt}$ for every $k$, and $\delta_{ijt} = 0$. Without continuity the implication can fail when $G_t$ has atoms, because the metric~\eqref{eq:netdist} ignores Lebesgue-null sets while the codegree does not.

\textit{Part (c).} Write $g(\tau, y') = f(w_{it}, \tau;\, y, y') - f(w_{jt}, \tau;\, y, y')$ and $(\mathcal{K} g)(\sigma, y'') = \int f(\sigma, \tau;\, y'', y')\, g(\tau, y')\, dG_t(\tau, y')$, so that $p_{ikt} - p_{jkt} = (\mathcal{K} g)(\xi_k)$ and $\delta_{ijt}^2 = \int (\mathcal{K} g)^2\, dG_t$. By definition of $g$,
\begin{align*}
\int g^2\, dG_t &= \int \big[f(w_{it}, \tau;\, y, y') - f(w_{jt}, \tau;\, y, y')\big]\, g(\tau, y')\, dG_t(\tau, y')\\
&= (\mathcal{K} g)(w_{it}, y) - (\mathcal{K} g)(w_{jt}, y).
\end{align*}
If $\delta_{ijt} = 0$, then $\mathcal{K} g = 0$ $G_t$-almost everywhere. When $G_t$ has finite support, this means $\mathcal{K} g = 0$ at every support point. When $f$ is continuous, $\mathcal{K} g$ is continuous and vanishes on the support. In either case $\mathcal{K} g$ vanishes at the two support points $(w_{it}, y)$ and $(w_{jt}, y)$, so $\int g^2\, dG_t = 0$, and $g = 0$ $G_t$-almost everywhere, hence at every support point under either condition. The converse is immediate. No definiteness property of the kernel $f$ is involved, but some regularity is: if $f$ is discontinuous and $G_t$ is continuous, $\mathcal{K} g$ can vanish almost everywhere without vanishing at the two evaluation points. What the codegree cannot see is a difference between link functions at partner types or lag strata outside the support of $G_t$: the metric~\eqref{eq:netdist} integrates over all $\tau \in [0,1]$ and both strata, whereas the codegree integrates against $G_t$. Under Assumption~\ref{ass:smooth}, $G_t$ has full support and $f$ is continuous, so agreement on the support is agreement everywhere, and $\delta_{ijt} = 0$ if and only if $\rho_{ijt} = 0$.
\end{proof}

The uniform bounds in part (d) extend to the aggregated distance $\hat{\bar{\delta}}_{ij}^2$ of equation~\eqref{eq:codegree_agg}, which averages at most $T$ period-specific distances over a set of periods that is observed and is the same for the sample statistic and its population counterpart.

The following consequence of part (c) is used in the proof of Theorem~\ref{thm:consistency}. Under Assumptions~\ref{ass:drift} and~\ref{ass:smooth}, write $\delta_{t,n}(u, v;\, y)$ for the population codegree distance under $\Prob_n$ and $\delta_t(u,v;\,y)$ for its limit under $G_t$. Because $f$ is uniformly continuous on its compact domain and $G_{t,n}$ converges weakly to $G_t$, $\delta_{t,n}^2 \to \delta_t^2$ uniformly on $[0,1]^2 \times \{0,1\}$. The functions $\delta_t^2$ and $\rho_t(u,v;\,y) = \|f_{u;\,y,\cdot} - f_{v;\,y,\cdot}\|_2$ are continuous on the compact set $[0,1]^2$ and, by part (c) and full support, have the same zero set. Hence
\begin{equation}\label{eq:sepmodulus}
\tau_n(h) := \sup\big\{\rho_t(u, v;\, y) : t \le T,\ y \in \{0,1\},\ \delta_{t,n}(u, v;\, y)^2 \le h\big\} \longrightarrow 0
\end{equation}
as $h \to 0$ and $n \to \infty$.

\subsection{Proofs of Proposition~\ref{prop:bridge}(b)--(c), Theorem~\ref{thm:identification}(b), and Corollary~\ref{cor:lambda}(b)}\label{sec:proof_thm1b}

\textit{Proposition~\ref{prop:bridge}(b).} Under Assumption~\ref{ass:smooth} with a fixed law, Lemma~\ref{lem:codegree_consist}(c) gives that $\delta_s(u,v;\,y) = 0$ if and only if $\|f_{u;\,y,\cdot} - f_{v;\,y,\cdot}\|_2 = 0$, and the argument leading to~\eqref{eq:sepmodulus}, with $G_{t,n} = G_t$ for every $n$, shows that $\tau(h) \to 0$ as $h \downarrow 0$: the functions $\delta_s^2$ and $\rho_s$ are continuous on the compact set $[0,1]^2$ for each $y$ and have the same zero set, so if $\tau(h)$ did not vanish there would be a sequence $(u_k, v_k, y_k)$ with $\delta_s(u_k, v_k;\, y_k)^2 \to 0$ and $\rho_s(u_k, v_k;\, y_k)$ bounded away from zero, and a limit point along a subsequence would have $\delta_s = 0$ and $\rho_s > 0$, a contradiction. The two consequences follow from the definition of $\varpi$ after Assumption~\ref{ass:persistence}, since $\|f_{u;\,y,\cdot} - f_{v;\,y,\cdot}\|_2 \le \tau(h)$ whenever $\delta_s(u,v;\,y)^2 \le h$.

\textit{Proposition~\ref{prop:bridge}(c).} Under Assumption~\ref{ass:discrete}(i), $G_s$ has finite support, so Lemma~\ref{lem:codegree_consist}(c) applies without continuity: $\delta_{ijs} = 0$ if and only if the two link functions agree at every supported partner type and lag stratum, which is equivalence on the support, and Assumption~\ref{ass:discrete}(ii) then gives $d_s(w_{is}) = d_s(w_{js})$. Since the support is finite, the set of positive values of $\delta_s^2$ over supported pairs is finite, so its minimum is positive, and it is at least the separation gap $\bar\epsilon$ of Assumption~\ref{ass:discrete}(v), which is the minimum over all periods.

\textit{Theorem~\ref{thm:identification}(b).} Condition on $(\mathcal{F}_{ijs}, w_i, w_j, \mu_i, \mu_j)$ and on $\mathcal{D}_{ijs}$. By the exact log-odds formula in the proof of Lemma~\ref{lem:ratio}, applied to both agents without any stability condition,
\[
\Prob\big(y_i \in A_s,\, y_j \in B_s \,\big|\, \cdot\big) = F\big(V_{ijs}(\beta,\alpha) + \Delta_{ijs}\big), \qquad
\Delta_{ijs} = \big[d_s(w_{is}) - d_s(w_{js})\big] + R_j - R_i,
\]
where $R_a = \Gamma(c_{a,s+2}) - \Gamma(c_{a,s+1})$ with $c_{at} = X_{at}'\beta + \lambda(w_{at}) + \mu_a$ and $\Gamma(c) = \log(1 + e^{c + \alpha}) - \log(1 + e^{c})$. Since $|\Gamma'| < 1$, $|R_a| \le |c_{a,s+2} - c_{a,s+1}|$, and by Assumption~\ref{ass:persistence}(ii), $c_{a,s+2} - c_{a,s+1} = (X_{a,s+2} - X_{a,s+1})'\beta$, in which the agent effect cancels. The value of $R_a$ depends on the level $c_{a,s+1}$, hence on $\mu_a$, but its bound does not. On $\mathcal{O}_{ijs}(h)$, $\|X_{a,s+2} - X_{a,s+1}\| \le h^{1/2}$ and, by Proposition~\ref{prop:bridge}(b), $|d_s(w_{is}) - d_s(w_{js})| \le \varpi(\tau(h))$, which gives~\eqref{eq:id_nbhd}. Because $|F(v + \Delta) - F(v)| \le |\Delta|/4$ and $F(V_{ijs}(\beta,\alpha))$ is $\mathcal{F}_{ijs}$-measurable, integrating over $(w_i, w_j, \mu_i, \mu_j)$ on $\mathcal{O}_{ijs}(h)$ gives the bound $\varpi_1(h)/4$ on the coarser conditional probability.

\label{loc:signfix}For the last statement, note that $m_{ijs}$ records the configuration $(y_i, y_j) \in B_s \times A_s$, the complement of the one in~\eqref{eq:id_nbhd} given the discordance event, and that $V_{ijs}(\beta, \alpha) = -Z_{ijs}'\theta$. Hence its conditional probability is $q_{ijs} = 1 - F(V_{ijs}(\beta, \alpha) + \Delta_{ijs}) = F(Z_{ijs}'\theta - \Delta_{ijs})$, with the same remainder $\Delta_{ijs}$ as in~\eqref{eq:id_nbhd}. Only $|\Delta_{ijs}| \le \varpi_1(h)$ is used below. Write $\ell(p; x) = p\log F(x) + (1-p)\log(1 - F(x))$. By iterated expectations, $M_h(b,a) = \E[\ell(q_{ijs};\, Z_{ijs}'(b,a)') \mid \mathcal{O}_{ijs}(h), \mathcal{D}_{ijs}]$, and $M_h$ is concave on $\Theta$. It is strictly concave for small $h$, because its Hessian is $-\E[F(Z_{ijs}'(b,a)')(1 - F(Z_{ijs}'(b,a)'))Z_{ijs}Z_{ijs}' \mid \cdot]$, which is negative definite whenever $\E[Z_{ijs}Z_{ijs}' \mid \mathcal{O}_{ijs}(h), \mathcal{D}_{ijs}]$ is positive definite, and the rank condition of the theorem implies the latter for all small $h$. Hence $\theta_h$ exists and is unique. Let $M_h^*$ be the same expression with $q_{ijs}$ replaced by $F(Z_{ijs}'\theta)$. The argument of Step~5 of the proof of Theorem~\ref{thm:consistency}, with $-\Delta_{ijs}$ in place of the remainder there (which is defined for the same configuration with the opposite sign convention, but only its absolute value is used), the indicator of $\mathcal{O}_{ijs}(h) \cap \mathcal{D}_{ijs}$, normalized by its probability, in place of $\ind_{ijs}W_{ijs}/\omega_n$, and $\varpi_1(h)$ in place of $\kappa_n$, gives $\sup_{\Theta}|M_h - M_h^*| \le C\varpi_1(h)$ by the moment condition and $M_h^*(\theta) - M_h^*(\theta + v) \ge \tfrac14 c_0 e^{-R\|v\|}\|v\|^2$ for all small $h$ by the rank condition. Therefore $\tfrac14 c_0 e^{-R\|\theta_h - \theta\|}\|\theta_h - \theta\|^2 \le M_h^*(\theta) - M_h^*(\theta_h) \le M_h(\theta) - M_h(\theta_h) + 2C\varpi_1(h) \le 2C\varpi_1(h)$, since $\theta_h$ maximizes $M_h$, and $\Theta$ is compact, so $\theta_h \to \theta$ as $h \downarrow 0$. \qed

\textit{Corollary~\ref{cor:lambda}(b).} Write $d^* = d_s(w_{is})$ and $\Psi_h(d)$ for the left side of~\eqref{eq:dident_h}. Condition first on $\chi_i$, $\chi_j$, $y_{j,-s}$, and $\mathcal{T}_{js}$. By the independence of agents $i$ and $j$, the conditional probability below depends only on $(\chi_j, y_{j,-s})$. By the exact log-odds formula in the proof of Lemma~\ref{lem:ratio}, applied without any stability condition, $\Prob(y_j \in B_s \mid \chi_i, \chi_j, y_{j,-s}, \mathcal{T}_{js}) = F\big(Z_{js}'\theta - d_s(w_{js}) + R_j\big)$ with $R_j = \Gamma(c_{j,s+2}) - \Gamma(c_{j,s+1})$, and, as in the proof of Theorem~\ref{thm:identification}(b), $|R_j| \le \|\beta\|\,\|X_{j,s+2} - X_{j,s+1}\| \le \|\beta\| h^{1/2}$ on $\mathcal{E}_{js}(h)$ by Assumption~\ref{ass:persistence}(ii). The remainder $R_j$ may depend on $\mu_j$, but its bound does not. On $\{j \in \mathcal{M}_{is}(h)\}$, Proposition~\ref{prop:bridge}(b) gives $|d_s(w_{js}) - d^*| \le \varpi(\tau(h))$. Hence the conditional probability equals $F(Z_{js}'\theta - d^* + e_j)$ with $|e_j| \le \varpi(\tau(h)) + \|\beta\|h^{1/2} \le \varpi_1(h)$, uniformly in the agent effect. Iterated expectations, which are valid because the event $\{j \in \mathcal{M}_{is}(h)\} \cap \mathcal{E}_{js}(h)$ is determined by $(w_{is}, w_{js}, X_j, y_{j,-s})$ and $\mathcal{T}_{js}$, give $\Psi_h(d) = \E\big[F(Z_{js}'\theta - d^* + e_j) - F(Z_{js}'\theta - d) \,\big|\, \cdot\,\big]$, where the conditioning is on $\{j \in \mathcal{M}_{is}(h)\} \cap \mathcal{E}_{js}(h)$ and $\chi_i$. The function $\Psi_h$ is continuous and strictly increasing, and by dominated convergence it tends to $\E[F(Z_{js}'\theta - d^* + e_j) \mid \cdot\,] > 0$ as $d \to \infty$ and to $\E[F(Z_{js}'\theta - d^* + e_j) \mid \cdot\,] - 1 < 0$ as $d \to -\infty$, so it has a unique zero $d_h$. \label{loc:bracket}Write $\epsilon_h = \varpi(\tau(h)) + \|\beta\|h^{1/2}$, so that $|e_j| \le \epsilon_h \le \varpi_1(h)$. Because $F$ is increasing, $F(Z_{js}'\theta - d^* - \epsilon_h) \le F(Z_{js}'\theta - d^* + e_j) \le F(Z_{js}'\theta - d^* + \epsilon_h)$ pointwise, and taking the same conditional expectations gives
\[
\Psi_h(d^* - \epsilon_h) \le 0 \le \Psi_h(d^* + \epsilon_h).
\]
Since $\Psi_h$ is continuous and strictly increasing, its unique zero lies between these endpoints, so $|d_h - d^*| \le \epsilon_h \le \varpi_1(h)$. The argument uses no moment of $Z_{js}$: because $F$ is bounded, $\Psi_h$ is well defined, continuous, and strictly increasing whatever the conditional law of $Z_{js}$. The conditions used above to obtain the probability representation and the overlap of the neighborhoods are still needed, and the estimation results of Section~\ref{sec:asymptotics} keep their own moment assumptions. Since $\varpi_1(h) \to 0$ by Proposition~\ref{prop:bridge}(b), $d_h \to d^*$. \qed
\subsection{Proof of Theorem~\ref{thm:consistency}}\label{sec:proof_consistency}

Let $N = n(n-1)/2$ and $Q_n(b,a) = \Omega_n(b,a)/(N\omega_n)$, which has the same maximizer as $\Omega_n$. Fix $(b,a) \in \Theta$. Since $-\log F(v) = \log(1 + e^{-v}) \le \log 2 + |v|$ and $\Theta$ is compact, $|m_{ijs}(b,a)| \le C(1 + \|Z_{ijs}\|)$ for a constant $C$ that does not depend on $(b,a)$ in $\Theta$. Note also that $W_{ijs} \le \|K_1\|_\infty \bar{W}_{ijs}$.

\textit{Step 1: Reduction by concavity.} Each $m_{ijs}(b,a)$ is a logistic log-likelihood in a linear index and is concave in $(b,a)$, and the kernel weights are non-negative, so $Q_n$ is a random concave function on the convex set $\Theta$, with $\theta$ in its interior. Let $M_n^*$ be the deterministic concave functions defined in Step~5. It suffices to show (a) that $Q_n(b,a) - M_n^*(b,a) \to_p 0$ for each $(b,a)$ and (b) that the $M_n^*$ are uniformly bounded on $\Theta$ and satisfy, for every $\eta > 0$, $\inf_{\|v\| \ge \eta,\, \theta + v \in \Theta}[M_n^*(\theta) - M_n^*(\theta + v)] \ge c(\eta) > 0$ for all $n$ large. Given any subsequence, uniform boundedness and a diagonal argument yield a further subsequence along which $M_n^*$ converges pointwise on a countable dense subset of $\Theta$. Since uniformly bounded concave functions are equicontinuous on compact subsets of the interior, the convergence extends to the interior, with a concave limit that inherits the curvature bound and is therefore uniquely maximized at $\theta$. Theorem~2.7 of \citet{newey1994}, whose argument uses only a compact convex neighborhood of $\theta$ in the interior, then gives $\hat{\theta} \to_p \theta$ along the further subsequence. Since every subsequence has such a further subsequence, $\hat{\theta} \to_p \theta$.

\textit{Step 2: Estimated distances.} Let $\tilde{Q}_n$ denote $Q_n$ with $\hat{\bar{\delta}}_{ij}^2$ replaced by $\bar{\delta}_{ij}^{2}$ in $K_1$, and let $u_n = \max_{i \neq j}|\hat{\bar{\delta}}_{ij}^2 - \bar{\delta}_{ij}^{2}|$ over pairs with a finite aggregated distance. By Lemma~\ref{lem:codegree_consist}(d), $u_n = O_p((\log n/n)^{1/2}) = o_p(h_1)$ by Assumption~\ref{ass:kernel}(ii). On the event $\{u_n \le h_1\}$, the two kernel values differ only if $\bar{\delta}_{ij}^{2} < 2h_1$, and by the Lipschitz property they differ by at most $L u_n/h_1$. Hence
\[
|Q_n(b,a) - \tilde{Q}_n(b,a)| \le \frac{L u_n}{h_1} \cdot \frac{C}{N\omega_n}\sum_{s \in \mathcal{S}}\sum_{i<j} \ind_{ijs}\bar{W}_{ijs}(1 + \|Z_{ijs}\|).
\]
The second factor has expectation bounded by Assumption~\ref{ass:regularity}(iii), so it is $O_p(1)$, and the right side is $o_p(1)$.

\textit{Step 3: Concentration.} The population weight $W_{ijs}$, the selection indicator, and $m_{ijs}$ are functions of the node data $\mathcal{Z}_i = (y_{i0}, \ldots, y_{iT}, X_i, w_i, \mu_i)$ and $\mathcal{Z}_j$ only, and the kernel $\phi(\mathcal{Z}_i, \mathcal{Z}_j) = \sum_{s \in \mathcal{S}} \ind_{ijs}W_{ijs}m_{ijs}(b,a)$ is symmetric, because $m_{ijs}$ takes the same value when the roles of $i$ and $j$ are exchanged on a discordant pair. Thus $\tilde{Q}_n(b,a)$ is a U-statistic of order two in i.i.d.\ data with mean $M_n(b,a)$. By the variance formula for U-statistics,
\[
\operatorname{Var}\big(\tilde{Q}_n(b,a)\big) \le \frac{1}{\omega_n^2}\left[\frac{4}{n}\,\E\big[\E(|\phi| \mid \mathcal{Z}_i)^2\big] + \frac{2}{N}\,\E[\phi^2]\right].
\]
By the bound on $|m_{ijs}|$ and Assumption~\ref{ass:regularity}(iii), $\E[\phi^2] \le C\sum_s \E[\ind_{ijs}\bar{W}_{ijs}(1 + \|Z_{ijs}\|)^2] \le C'\omega_n$. By the Cauchy--Schwarz inequality, $\E(|\phi| \mid \mathcal{Z}_i)^2 \le C\,\E\big[\sum_s \ind_{ijs}\bar{W}_{ijs}(1 + \|Z_{ijs}\|)^2 \,\big|\, \mathcal{Z}_i\big]\cdot \E\big[\sum_s \ind_{ijs}\bar{W}_{ijs} \,\big|\, \mathcal{Z}_i\big]$, and the second factor is bounded, so $\E[\E(|\phi| \mid \mathcal{Z}_i)^2] \le C''\omega_n$. Hence $\operatorname{Var}(\tilde{Q}_n(b,a)) \le C/(n\omega_n) \to 0$ by Assumption~\ref{ass:kernel}(ii), and $\tilde{Q}_n(b,a) - M_n(b,a) \to_p 0$. Together with Step 2, $Q_n(b,a) - M_n(b,a) \to_p 0$.

\textit{Step 4: The discordant configuration on the selected pairs.} Condition on $(\mathcal{F}_{ijs}, w_i, w_j, \mu_i, \mu_j)$ and on the discordance event. The weight $W_{ijs}$ and the selection indicator are measurable with respect to this information: the set of periods entering the aggregated distance excludes $s+1$ and $s+2$ under both orderings of the swapped outcomes (third remark after Theorem~\ref{thm:identification}), and the remaining periods' distances depend only on the latent indices and on lags outside the window. By the exact log-odds formula in the proof of Lemma~\ref{lem:ratio}, applied to both agents,
\begin{align*}
\Prob\big(y_{is} = 1 \,\big|\, \mathcal{F}_{ijs}, w_i, w_j, \mu_i, \mu_j, \mathcal{D}_{ijs}\big) &= F\big(Z_{ijs}'\theta + \Delta_{ijs}\big),\\
\Delta_{ijs} &= \big[\lambda(w_{is}) - \lambda(w_{i,s+1})\big] - \big[\lambda(w_{js}) - \lambda(w_{j,s+1})\big] + R_i - R_j,
\end{align*}
where $R_a = \Gamma(c_{a,s+2}) - \Gamma(c_{a,s+1})$ satisfies $|R_a| \le |c_{a,s+2} - c_{a,s+1}|$ with $c_{at} = X_{at}'\beta + \lambda(w_{at}) + \mu_a$. The remainders $R_a$, and hence $\Delta_{ijs}$, may depend on the agent effects through the levels $c_{a,s+1}$, but their bounds do not, because the agent effects cancel in the differences $c_{a,s+2} - c_{a,s+1}$. The bound below therefore holds uniformly in $(\mu_i, \mu_j)$, which is all that Step~5 uses. By Assumption~\ref{ass:drift}, $\lambda(w_{a,s+1}) - \lambda(w_{as}) = d_s(w_{as}) + e_{as,n}$ and, since $d_{s+1} \equiv 0$, $|\lambda(w_{a,s+2}) - \lambda(w_{a,s+1})| \le \zeta_n$. On a pair with $W_{ijs} > 0$, $\|X_{a,s+2} - X_{a,s+1}\| < h_2^{1/2}$ for $a \in \{i,j\}$, and $\delta_{s,n}(w_{is}, w_{js};\, y_{i,s-1})^2 \le T\bar{\delta}_{ij}^{2} < T h_1$, because period $s$ contributes to the aggregated distance. Hence, with $\varpi$ the modulus of continuity of $d_s$ defined after Assumption~\ref{ass:persistence} and $\tau_n$ from~\eqref{eq:sepmodulus},
\[
|\Delta_{ijs}| \le \varpi\big(\tau_n(T h_1)\big) + 2\|\beta\|\, h_2^{1/2} + 4\zeta_n =: \kappa_n \longrightarrow 0
\]
on every pair with positive weight.

\textit{Step 5: Curvature of the localized criterion.} Write $\ell(p; x) = p\log F(x) + (1 - p)\log(1 - F(x))$, which is linear in $p$. By Step 4 and iterated expectations over $(w_i, w_j, \mu_i, \mu_j)$, $M_n(b,a) = \omega_n^{-1}\sum_s \E_n[\ind_{ijs}W_{ijs}\,\ell(F(Z_{ijs}'\theta + \Delta_{ijs});\, Z_{ijs}'(b,a)')]$, where the expectation integrates over the agent effects on which $\Delta_{ijs}$ may depend. Let $M_n^*$ denote the same expression with $\Delta_{ijs} = 0$. It is a deterministic concave function whose comparison logit index $Z_{ijs}'\theta$ contains no individual-effect term, although its expectation still uses the law of the selected population, which can depend on the individual effects. Since $|\partial \ell/\partial p| = |x|$ and $|F(v + \Delta) - F(v)| \le |\Delta|/4$, with the bound on $|\Delta_{ijs}|$ holding uniformly in the agent effects,
\[
|M_n(b,a) - M_n^*(b,a)| \le \frac{\kappa_n}{4}\,\frac{1}{\omega_n}\sum_s \E_n\big[\ind_{ijs}W_{ijs}\,|Z_{ijs}'(b,a)'|\big] \longrightarrow 0,
\]
so $Q_n(b,a) - M_n^*(b,a) \to_p 0$ by Step 3, which is condition (a) of Step 1. The same bound and Assumption~\ref{ass:regularity}(iii) show that the $M_n^*$ are uniformly bounded on $\Theta$. $M_n^*$ is the expected log-likelihood of a correctly specified logit, so for $v \neq 0$,
\[
M_n^*(\theta) - M_n^*(\theta + v) = \frac{1}{\omega_n}\sum_s \E_n\big[\ind_{ijs}W_{ijs}\,\mathrm{KL}\big(Z_{ijs}'\theta,\ Z_{ijs}'(\theta + v)\big)\big],
\]
where $\mathrm{KL}(a, a + x) = \int_0^x (x - u)F'(a + u)\,du$ is the Kullback--Leibler divergence between Bernoulli distributions with logits $a$ and $a + x$. For the logistic distribution, $F'(a + u) \ge F'(a)e^{-|u|}$, so $\mathrm{KL}(a, a + x) \ge \tfrac12 F'(a)x^2 e^{-|x|}$. Restricting to $\|Z_{ijs}\| \le R$ and using $F' \le 1/4$,
\[
M_n^*(\theta) - M_n^*(\theta + v) \ge \frac{e^{-R\|v\|}}{2}\left[v'\Sigma_n v - \frac{\|v\|^2}{4R^2}\,\frac{1}{\omega_n}\sum_s \E_n\big[\ind_{ijs}W_{ijs}\|Z_{ijs}\|^4\big]\right].
\]
By Assumption~\ref{ass:regularity}(iii) the last average is bounded, and by part (iv) there is $c_0 > 0$ with $v'\Sigma_n v \ge c_0\|v\|^2$ for $n$ large. Choosing $R$ large enough, the right side is at least $\tfrac14 c_0 e^{-R\|v\|}\|v\|^2$ for all $n$ large. For $\|v\| \ge \eta$, concavity of $M_n^*$ along the ray from $\theta$ through $\theta + v$ gives $M_n^*(\theta) - M_n^*(\theta + v) \ge M_n^*(\theta) - M_n^*(\theta + \eta v/\|v\|) \ge \tfrac14 c_0 e^{-R\eta}\eta^2$, which is condition (b) of Step 1 and completes the proof. \qed

\subsection{Proof of Theorem~\ref{thm:normality}}\label{sec:proof_normality}

Let $\bar{K} = K_1(0)K_2(0)^2$, $Q_n(b,a) = \Omega_n(b,a)/(N\bar{K})$, and let
\[
Q_n^{o}(b,a) = \frac{1}{N}\sum_{s \in \mathcal{S}}\sum_{i<j}\ind^{\mathrm{ex}}_{ijs}\, m_{ijs}(b,a)
\]
be the oracle criterion that uses the exact-match indicator in place of the kernel weights. The law of the primitives is fixed, and all moments below are finite by Assumption~\ref{ass:regularity}(i), since $Z_{ijs}$ is a linear function of the two agents' covariates and of bounded outcomes.

\textit{Step 1: The estimated weights select the exact matches.} Let $\hat{W}_{ijs}$ be the estimated weight of pair $(i,j)$ at $s$. Because $h_2 \le \bar\epsilon_X$, $K_2(\|X_{a,s+2} - X_{a,s+1}\|^2/h_2) = K_2(0)\ind\{X_{a,s+2} = X_{a,s+1}\}$ almost surely. By Assumption~\ref{ass:discrete}(v) and the definition of $\bar\epsilon$, a pair in $\mathcal{P}_s$ with $\bar{\delta}_{ij}^2 > 0$ has $\bar{\delta}_{ij}^2 \ge \bar\epsilon/(T-2) \ge 2h_1$. By Lemma~\ref{lem:codegree_consist}(d), with probability approaching one all such pairs have $\hat{\bar{\delta}}_{ij}^2 > h_1$ and therefore receive zero weight. For pairs with $\bar{\delta}_{ij} = 0$, the two agents are equivalent on the support at every contributing period, so Lemma~\ref{lem:codegree_consist}(d) gives $\max \hat{\bar{\delta}}_{ij}^2 = O_p(\log n/n)$ and, by the Lipschitz property,
\[
\max_{i<j,\ s}\ \ind_{ijs}\,\big|\hat{W}_{ijs}/\bar{K} - \ind^{\mathrm{ex}}_{ijs}\big| = O_p\!\left(\frac{\log n}{n h_1}\right) = o_p(n^{-1/2})
\]
with probability approaching one, by Assumption~\ref{ass:discrete}(v).

\textit{Step 2: The oracle score is exactly centered.} On the event $\ind^{\mathrm{ex}}_{ijs} = 1$, the agents share the lag at $s - 1$ and have $\bar{\delta}_{ij} = 0$, so by Lemma~\ref{lem:codegree_consist}(c) they are equivalent on the support at $s$. By Assumption~\ref{ass:discrete}(ii), $d_s(w_{is}) = d_s(w_{js})$, so their social influence terms change by the same amount between $s$ and $s+1$. By Assumption~\ref{ass:persistence}(ii), both are stable between $s+1$ and $s+2$, and their covariates are exactly stable. The event is therefore contained in $\mathcal{C}_{ijs} \cap \mathcal{X}_{ijs}$, and it is measurable with respect to $(\mathcal{F}_{ijs}, w_i, w_j)$ and the discordance event, as in Step~4 of the proof of Theorem~\ref{thm:consistency}. Lemma~\ref{lem:identity} gives $\Prob(y_{is} = 1 \mid \mathcal{F}_{ijs}, w_i, w_j, \mathcal{D}_{ijs}) = F(Z_{ijs}'\theta)$ on this event, hence $\E[\ind^{\mathrm{ex}}_{ijs}S_{ijs}(\theta) \mid \mathcal{F}_{ijs}, w_i, w_j] = 0$, and in particular $\E[\ind^{\mathrm{ex}}_{ijs}S_{ijs}(\theta)] = 0$.

\textit{Step 3: Consistency.} By Step 1,
\[
|Q_n(b,a) - Q_n^{o}(b,a)| \le o_p(n^{-1/2})\, N^{-1}\sum_{s}\sum_{i<j}\ind^{\mathrm{ex}}_{ijs}|m_{ijs}(b,a)| = o_p(1).
\]
The oracle criterion is a U-statistic with integrable symmetric kernel, so $Q_n^{o}(b,a) \to_p Q^{o}(b,a) = \sum_s \E[\ind^{\mathrm{ex}}_{ijs}m_{ijs}(b,a)]$. The function $Q^{o}$ is concave with gradient $\sum_s\E[\ind^{\mathrm{ex}}_{ijs}S_{ijs}(\theta)] = 0$ at $\theta$ by Step 2, and its Hessian $-\sum_s \E[\ind^{\mathrm{ex}}_{ijs}F'(Z_{ijs}'b)Z_{ijs}Z_{ijs}']$, where from here on $b \in \Theta$ denotes a value of the whole parameter vector, is negative definite at every $b$, because $F' > 0$ and $\sum_s\E[\ind^{\mathrm{ex}}_{ijs}Z_{ijs}Z_{ijs}']$ is positive definite when $\Sigma_{\mathrm{ex}}$ is. Hence $\theta$ is the unique maximizer of $Q^{o}$, and, since $Q_n$ is concave, Theorem~2.7 of \citet{newey1994} gives $\hat{\theta} \to_p \theta$.

\textit{Step 4: Score.} Let $U_n = N^{-1}\sum_s\sum_{i<j}\ind^{\mathrm{ex}}_{ijs}S_{ijs}(\theta)$. By Step 1,
\[
\|\nabla Q_n(\theta) - U_n\| \le o_p(n^{-1/2})\, N^{-1}\sum_s\sum_{i<j}\ind^{\mathrm{ex}}_{ijs}\|S_{ijs}(\theta)\| = o_p(n^{-1/2}).
\]
The statistic $U_n$ is a U-statistic of order two in the i.i.d.\ node data with the kernel $k(\mathcal{Z}_i, \mathcal{Z}_j) = \sum_s \ind^{\mathrm{ex}}_{ijs}S_{ijs}(\theta)$, which is symmetric because $S_{ijs} = S_{jis}$ on discordant pairs. By Step 2 the kernel has mean zero, and it has finite second moment. Hoeffding's decomposition gives $\sqrt{n}\,U_n = 2n^{-1/2}\sum_i \psi_i + O_p(n^{-1/2})$ with $\psi_i = \E[k(\mathcal{Z}_i, \mathcal{Z}_j) \mid \mathcal{Z}_i]$, and the central limit theorem gives $\sqrt{n}\,U_n \dlim \N(0, 4V_{\mathrm{ex}})$. Note that $\psi_i$ sums over the transition periods before the conditional expectation is taken, so the variance retains covariances across periods.

\textit{Step 5: Hessian and conclusion.} The Hessian of $Q_n$ is
\[
H_n(b) = -N^{-1}\sum_s\sum_{i<j}(\hat{W}_{ijs}/\bar{K})\ind_{ijs}F'(Z_{ijs}'b)Z_{ijs}Z_{ijs}'.
\]
By Step 1, $H_n(b)$ differs from its oracle counterpart by $o_p(1)$ uniformly in $b$. Since $|F'(x) - F'(x')| \le |x - x'|$, the oracle Hessian is Lipschitz in $b$ with an integrable envelope proportional to $\|Z_{ijs}\|^3$, so the uniform law of large numbers for U-statistics gives $\sup_{\|b - \theta\| \le \eta}\|H_n(b) + \Sigma_{\mathrm{ex}}(b)\| \to_p 0$ for small $\eta > 0$, with $\Sigma_{\mathrm{ex}}(b) = \sum_s\E[\ind^{\mathrm{ex}}_{ijs}F'(Z_{ijs}'b)Z_{ijs}Z_{ijs}']$ continuous and $\Sigma_{\mathrm{ex}}(\theta) = \Sigma_{\mathrm{ex}}$. With probability approaching one, $\hat{\theta}$ is interior and satisfies the first-order condition, and
\[
0 = \nabla Q_n(\hat{\theta}) = \nabla Q_n(\theta) + \left[\int_0^1 H_n\big(\theta + u(\hat{\theta} - \theta)\big)\,du\right](\hat{\theta} - \theta).
\]
The integrated Hessian converges in probability to $-\Sigma_{\mathrm{ex}}$ because $\hat{\theta} \to_p \theta$, so $\sqrt{n}(\hat{\theta} - \theta) = \Sigma_{\mathrm{ex}}^{-1}\sqrt{n}\,U_n + o_p(1) \dlim \N(0, 4\Sigma_{\mathrm{ex}}^{-1}V_{\mathrm{ex}}\Sigma_{\mathrm{ex}}^{-1})$. \qed

\medskip
\noindent\textit{Variance estimation.}\label{loc:varest} Write $\tilde{\Sigma}(b)$ and $\tilde{\psi}_i(b)$ for the oracle versions of $\hat{\Sigma}$ and $\hat{\psi}_i$ evaluated at $b$, in which $\hat{W}_{ijs}/\bar{K}$ is replaced by $\ind^{\mathrm{ex}}_{ijs}$, and let $r_n = \max_{i<j,\,s}\ind_{ijs}|\hat{W}_{ijs}/\bar{K} - \ind^{\mathrm{ex}}_{ijs}|$. By Step 1, $r_n = O_p(\log n/(nh_1)) = o_p(1)$, and with probability approaching one the two sets of weights differ only on exact-match pairs. On that event, since $\|S_{ijs}(b)\| \le \|Z_{ijs}\|$ and $F(1 - F) \le \tfrac14$,
\begin{align*}
\|\hat{\psi}_i - \tilde{\psi}_i(\hat{\theta})\| &\le r_n A_i, \qquad A_i = \frac{1}{n-1}\sum_{j \ne i}\sum_{s \in \mathcal{S}}\ind^{\mathrm{ex}}_{ijs}\|Z_{ijs}\|,\\
\|\hat{\Sigma} - \tilde{\Sigma}(\hat{\theta})\| &\le \frac{r_n}{4}\binom{n}{2}^{-1}\sum_{s \in \mathcal{S}}\sum_{i<j}\ind^{\mathrm{ex}}_{ijs}\|Z_{ijs}\|^2,
\end{align*}
so that $n^{-1}\sum_i \|\hat{\psi}_i - \tilde{\psi}_i(\hat{\theta})\|^2 \le r_n^2\, n^{-1}\sum_i A_i^2$. By Jensen's inequality, $\E[A_i^2] \le \E\big[(\sum_s \ind^{\mathrm{ex}}_{ijs}\|Z_{ijs}\|)^2\big] \le |\mathcal{S}|\sum_s\E\big[\ind^{\mathrm{ex}}_{ijs}\|Z_{ijs}\|^2\big] < \infty$ for fixed $T$ by Assumption~\ref{ass:regularity}(i), so $n^{-1}\sum_i A_i^2 = O_p(1)$ and the average squared difference is $o_p(1)$. The same moment bound makes the average in the bound for $\hat{\Sigma}$ equal to $O_p(1)$, so $\hat{\Sigma} - \tilde{\Sigma}(\hat{\theta}) = o_p(1)$. No bound on the maximum over agents is needed, only these averages. The oracle $\tilde{\Sigma}(\hat{\theta})$ is a U-statistic evaluated at $\hat{\theta}$ and converges to $\Sigma_{\mathrm{ex}}$ by the argument of Step 5. For the influence terms, three replacements lead from $\hat{\psi}_i$ to $\psi_i$. First, the one just shown. Second, since $|F(Z_{ijs}'\hat{\theta}) - F(Z_{ijs}'\theta)| \le \tfrac14\|Z_{ijs}\|\,\|\hat{\theta} - \theta\|$, $\|\tilde{\psi}_i(\hat{\theta}) - \tilde{\psi}_i(\theta)\| \le \tfrac14\|\hat{\theta} - \theta\|\,B_i$ with $B_i = (n-1)^{-1}\sum_{j \ne i}\sum_s\ind^{\mathrm{ex}}_{ijs}\|Z_{ijs}\|^2$, and $n^{-1}\sum_i B_i^2 = O_p(1)$ by the fourth moments of Assumption~\ref{ass:regularity}(i), so this average squared difference is $O_p(\|\hat{\theta} - \theta\|^2) = o_p(1)$. Third, given $\mathcal{Z}_i$ the summands of $\tilde{\psi}_i(\theta) = (n-1)^{-1}\sum_{j \ne i} k(\mathcal{Z}_i, \mathcal{Z}_j)$ are i.i.d.\ with mean $\psi_i$, and $\E\|k(\mathcal{Z}_i, \mathcal{Z}_j)\|^2 < \infty$ by Assumption~\ref{ass:regularity}(i), so $\E\|\tilde{\psi}_i(\theta) - \psi_i\|^2 \le \E\|k(\mathcal{Z}_i, \mathcal{Z}_j)\|^2/(n-1)$. For each of the three replacements the average squared difference $n^{-1}\sum_i\|\cdot\|^2$ is $o_p(1)$, so by Minkowski's inequality $n^{-1}\sum_i\|\hat{\psi}_i - \psi_i\|^2 = o_p(1)$. Since $n^{-1}\sum_i\|\psi_i\|^2 = O_p(1)$, the Cauchy--Schwarz inequality shows that the sample mean and second moment of the $\hat{\psi}_i$ differ from those of the $\psi_i$ by $o_p(1)$. By the law of large numbers, $n^{-1}\sum_i\psi_i \to_p \E[\psi_i] = 0$ and $n^{-1}\sum_i\psi_i\psi_i' \to_p V_{\mathrm{ex}}$, so $\hat{V} \to_p V_{\mathrm{ex}}$ and the plug-in variance is consistent \citep[cf.][Chapter 5]{serfling1980}.

\subsection{Conditional Sums of Bernoulli Variables}\label{sec:lemma_rejective}

\begin{lemma}[Conditional Bernoulli sums]\label{lem:rejective}
Let $B_1, \ldots, B_m$ be independent Bernoulli variables with success probabilities $p_j \in (0,1)$ such that $\sum_j p_j = N$ is an integer. Write $q_j = p_j(1 - p_j)$, $d = \sum_j q_j$, $\pi_j = \Prob(B_j = 1 \mid \sum_l B_l = N)$, and, for vectors $z_1, \ldots, z_m$ of fixed dimension, $\bar z = d^{-1}\sum_j q_j z_j$. There are constants $d_0$ and $C$, which do not depend on $m$ or on the $p_j$, such that whenever $d \ge d_0$:
\begin{enumerate}[label=(\alph*)]
\item $\big\|\sum_j (\pi_j - p_j) z_j\big\| \le C d^{-1}\sum_j \|z_j\|$;
\item $\big\|\mathrm{Var}\big(\sum_j B_j z_j \,\big|\, \sum_l B_l = N\big) - \sum_j q_j (z_j - \bar z)(z_j - \bar z)'\big\| \le C d^{-1}\sum_j \|z_j\|^2 + C d^{-2}\big(\sum_j \|z_j\|\big)^2$.
\end{enumerate}
\end{lemma}

\begin{proof}
By \citet{hajek1964} and Theorem~1 of \citet{boistard2012}, $\pi_j = p_j + O(d^{-1})$ and $\mathrm{Cov}(B_j, B_k \mid \sum_l B_l = N) = -q_j q_k/d + O(d^{-2})$ for $j \ne k$, uniformly in the indices, and hence $\pi_j(1 - \pi_j) = q_j + O(d^{-1})$. Part (a) follows by multiplying the first expansion by $z_j$ and summing. For part (b), summing the conditional variances and covariances against $z_j z_k'$ gives the leading term $\sum_j q_j z_j z_j' - d^{-1}(\sum_j q_j z_j)(\sum_j q_j z_j)' = \sum_j q_j (z_j - \bar z)(z_j - \bar z)'$ and the diagonal correction $d^{-1}\sum_j q_j^2 z_j z_j'$, whose norm is at most $(4d)^{-1}\sum_j\|z_j\|^2$ because $q_j \le 1/4$. The $O(d^{-1})$ error in the variances contributes at most $C d^{-1}\sum_j\|z_j\|^2$, and the $O(d^{-2})$ error in the covariances at most $C d^{-2}(\sum_j\|z_j\|)^2$.
\end{proof}

\subsection{Proof of Proposition~\ref{prop:class}}\label{sec:proof_class}

\textit{Part (a).} By Proposition~\ref{prop:bridge}(a), $\max_{i \ne j}|\hat\delta_{ijs}^2 - \delta_{ijs}^2|\ind\{y_{i,s-1} = y_{j,s-1}\} = O_p((\log n/n)^{1/2})$, which is $o_p(h_1)$ because $\sqrt{n}\,h_1/\log n \to \infty$ by Assumption~\ref{ass:discrete}(v). By Proposition~\ref{prop:bridge}(c), a pair that shares the lagged outcome at $s-1$ has $\delta_{ijs}^2 = 0$ if its members are in the same population class and $\delta_{ijs}^2 \ge \bar\epsilon$ otherwise, and $\bar\epsilon \ge 2(T-2)h_1 \ge 2h_1$ by Assumption~\ref{ass:discrete}(v). Hence, with probability approaching one, for every pair of eligible agents that share the lagged outcome, $\hat\delta_{ijs}^2 < h_1$ exactly when the two agents are in the same population class, and the edges of the graph~\eqref{eq:class_graph} are exactly the pairs of eligible agents in the same population class. Equivalence on the support is reflexive, symmetric, and transitive, so these edges form a complete graph on each set $\{j \in \mathcal{E}_s : c_j = c\}$ and join no two different sets. The connected components are these sets.

\textit{Part (b).} Let $\mathcal{H}_s$ be the $\sigma$-field generated by the primitives, all link shocks, the outcomes $\{y_{j,-s}\}_{j \le n}$, and the switch events. By Assumption~\ref{ass:sampling}(i), (ii), and (iv), given the primitives and the link shocks, the shock vectors $\varepsilon_j$ are independent across agents, so the within-window outcomes $(y_{js}, y_{j,s+1})$ are independent across agents given $\mathcal{H}_s$, and the conditional law of agent $j$'s pair depends only on $(\chi_j, y_{j,-s})$ and $\mathcal{T}_{js}$. For an eligible agent the pair is $(1, 0)$ or $(0, 1)$, and the ratio of their probabilities is given by Lemma~\ref{lem:ratio}, whose stability conditions hold because $X_{j,s+2} = X_{j,s+1}$ and by Assumption~\ref{ass:persistence}(ii). Since $\lambda(w_{js}) - \lambda(w_{j,s+1}) = -d_s(w_{js})$ by Assumption~\ref{ass:persistence}(i), $\Prob(b_{js} = 1 \mid \mathcal{H}_s) = F(Z_{js}'\theta - d_s(w_{js}))$. The network $D_t$ is a function of the primitives, the link shocks at $t$, and the outcomes at $t - 1$. For $t \notin \{s+1, s+2\}$ these outcomes are among the $y_{j,-s}$, so $\mathcal{G}_s \subseteq \mathcal{H}_s$, and since the conditional law given $\mathcal{H}_s$ is $\mathcal{G}_s$-measurable, the tower property gives the same law given $\mathcal{G}_s$. For a $\mathcal{G}_s$-measurable set $C$ on which $d_s(w_{js}) = d_C$, the probability of orderings $a \in \{0,1\}^C$ given $\mathcal{G}_s$ is therefore
\[
\prod_{j \in C}\frac{\exp\{a_j(Z_{js}'\theta - d_C)\}}{1 + \exp(Z_{js}'\theta - d_C)} = \exp\Big(\sum_{j \in C} a_j Z_{js}'\theta\Big)\, \frac{e^{-d_C \sum_j a_j}}{\prod_{j \in C}\big(1 + e^{Z_{js}'\theta - d_C}\big)},
\]
and summing over the $a$ with $\sum_j a_j = N_C$ shows that the conditional probability of the observed orderings given $N_C$ is $\exp\{\ell_C(\theta)\}$, free of $d_C$. The feasible classes are $\mathcal{G}_s$-measurable, because they are functions of $D_s$, of the outcomes at $s - 1$, of the covariates, and of the switch events.

\textit{Part (c), Step 1: reduction to the population classes.} By part (a), with probability approaching one the objective in~\eqref{eq:class_estimator} equals $L_n(\theta) = \sum_c \ell_{C_c}(\theta)$ with $C_c = \{j \in \mathcal{E}_s : c_j = c\}$, so it suffices to prove the result for the maximizer of $L_n$. Let $K$ be the number of classes with $\Prob(e_j = 1, c_j = c) > 0$. The other sets $C_c$ are empty almost surely. Since $|C_c|/n$ converges to a positive limit and the orderings are independent logits by part (b), each of the $K$ sets contains both orderings with probability approaching one, which is assumed below.

\textit{Step 2: comparison with the profile likelihood.} For a class $C$ with $0 < N_C < |C|$ and $\gamma \in \R$, let $L_C(\theta, \gamma) = \sum_{j \in C}\big[b_{js}(Z_{js}'\theta - \gamma) - \log(1 + e^{Z_{js}'\theta - \gamma})\big]$, the log likelihood of a logit with intercept $-\gamma$, and let $P_{\theta,\gamma}$ be the law of independent Bernoulli variables with success probabilities $F(Z_{js}'\theta - \gamma)$, $j \in C$. The factorization in part (b) gives, for every $\gamma$,
\[
\ell_C(\theta) = L_C(\theta, \gamma) - \log P_{\theta,\gamma}\Big(\sum_{j \in C} b_{js} = N_C\Big).
\]
Let $\gamma_C(\theta)$ solve $\sum_{j \in C} F(Z_{js}'\theta - \gamma) = N_C$, which has a unique solution because the left side decreases continuously from $|C|$ to $0$. It is the first-order condition of the strictly concave function $\gamma \mapsto L_C(\theta, \gamma)$, so $L_C(\theta, \gamma_C(\theta)) = \max_\gamma L_C(\theta, \gamma) =: L_C^{\mathrm{P}}(\theta)$, the profile log likelihood. Under $P_{\theta, \gamma_C(\theta)}$ the number of successes has mean $N_C$, an integer, so $N_C$ is a mode of its distribution \citep{darroch1964}, and its probability is at least $1/(|C| + 1)$. Hence
\begin{equation}\label{eq:class_profile}
0 \le \ell_C(\theta) - L^{\mathrm{P}}_C(\theta) \le \log(|C| + 1) \qquad \text{for every } \theta.
\end{equation}

\textit{Step 3: consistency.} Let $L_n^{\mathrm{J}}(\theta, \gamma) = \sum_c L_{C_c}(\theta, \gamma_c)$ for $\gamma \in \R^K$, the log likelihood of a logit for the eligible agents' orderings with one intercept per class, which is a sum of i.i.d.\ terms, and $L_n^{\mathrm{P}}(\theta) = \max_\gamma L_n^{\mathrm{J}}(\theta, \gamma) = \sum_c L^{\mathrm{P}}_{C_c}(\theta)$. By~\eqref{eq:class_profile}, $\sup_\theta n^{-1}|L_n(\theta) - L_n^{\mathrm{P}}(\theta)| \le K\log(n + 1)/n \to 0$. The function $n^{-1}L_n^{\mathrm{J}}$ is concave and converges pointwise in probability to $Q(\theta, \gamma) = \E\big[e_j\{b_{js}(Z_{js}'\theta - \gamma_{c_j}) - \log(1 + e^{Z_{js}'\theta - \gamma_{c_j}})\}\big]$, which is strictly concave, because its Hessian is $-\E[e_j F'(Z_{js}'\theta - \gamma_{c_j}) g_j g_j']$ with $g_j = (Z_{js}', -\iota_{c_j}')'$ and $\iota_c$ the $c$-th unit vector, and this matrix is negative definite when the regressor varies within classes, which the positive definiteness of $\mathcal{I}_{\mathrm{C}}$ implies. By part (b) and the information inequality, $Q$ is uniquely maximized at $(\theta, d)$, where $d$ collects the class movements. For each $\theta$, pointwise convergence of the concave functions $n^{-1}L_n^{\mathrm{J}}(\theta, \cdot)$ implies uniform convergence on compact sets (the convexity lemma used in the proof of Theorem~2.7 of \citealp{newey1994}). Since the limit $Q(\theta, \cdot)$ is strictly concave with a unique maximizer, and the maximizers of concave functions that converge uniformly on a neighborhood of it lie in that neighborhood with probability approaching one, $n^{-1}L_n^{\mathrm{P}}(\theta) = \max_\gamma n^{-1}L_n^{\mathrm{J}}(\theta, \gamma)$ converges in probability to $Q^{\mathrm{P}}(\theta) = \max_\gamma Q(\theta, \gamma)$, which is concave with unique maximizer $\theta$. Hence $n^{-1}L_n(\theta) \to_p Q^{\mathrm{P}}(\theta)$ for every $\theta$. Since $L_n$ is concave, Theorem~2.7 of \citet{newey1994} gives $\hat\theta^{\mathrm{C}} \to_p \theta$.

\textit{Step 4: asymptotic normality.} Fix $u$, let $\theta_u = \theta + u/\sqrt{n}$, and write $L_n = L_n^{\mathrm{P}} + R_n$. The profile part admits the expansion
\[
L_n^{\mathrm{P}}(\theta_u) - L_n^{\mathrm{P}}(\theta) = u'U_n - \tfrac12 u'\mathcal{I}_{\mathrm{C}}u + o_p(1), \qquad U_n = n^{-1/2}\sum_{j=1}^n \tilde S_j,
\]
which is the local expansion of the profile log likelihood of a regular parametric model with $\dim\theta + K$ parameters and i.i.d.\ observations: expanding $L_n^{\mathrm{J}}(\theta_u, d + v/\sqrt{n})$ to second order in $(u, v)$, with the Hessian replaced by its expectation by the law of large numbers and the third-order term of order $n^{-1/2}$ by the moments of Assumption~\ref{ass:regularity}(i), and maximizing the quadratic over $v$, whose maximizer is bounded in probability, gives the score for $\theta$ minus its projection on the intercept scores, which is $U_n$, and the Schur complement of the intercept block of the information, which is $\mathcal{I}_{\mathrm{C}}$. For the remainder $R_n(\theta) = \sum_c[\ell_{C_c}(\theta) - L^{\mathrm{P}}_{C_c}(\theta)]$, the gradient of $\ell_C$ is $\sum_{j \in C}(b_{js} - \pi_j(\theta))Z_{js}$, where $\pi_j(\theta)$ is the probability that $b_{js} = 1$ under the conditional law of the orderings given $N_C$, and, by the envelope theorem, the gradient of $L^{\mathrm{P}}_C$ is $\sum_{j \in C}(b_{js} - p_j(\theta))Z_{js}$ with $p_j(\theta) = F(Z_{js}'\theta - \gamma_C(\theta))$. Hence $\nabla R_n(\theta) = \sum_c\sum_{j \in C_c}(p_j(\theta) - \pi_j(\theta))Z_{js}$. The conditional law of the orderings given $N_C$ is the rejective, or conditional Poisson, sampling design with working probabilities $p_j(\theta)$, whose sum is $N_C$. Let $\sigma_C^2(\theta') = \sum_{j \in C}p_j(\theta')(1 - p_j(\theta'))$. For $\|\theta' - \theta\| \le \|u\|/\sqrt{n}$, $\sigma^2_{C_c}(\theta')/n$ converges in probability, uniformly, to $\E[e_j\ind\{c_j = c\}q_j] > 0$, so Lemma~\ref{lem:rejective}(a), applied conditionally on the regressors with $z_j = Z_{js}$ and $d = \sigma_C^2(\theta')$, gives $\sup_{\theta'}\|\nabla R_n(\theta')\| \le O_p(1)\,n^{-1}\sum_j e_j\|Z_{js}\| = O_p(1)$ over that ball, and $R_n(\theta_u) - R_n(\theta) = O_p(n^{-1/2}) = o_p(1)$ by the mean value theorem. Therefore $L_n(\theta_u) - L_n(\theta) = u'U_n - \tfrac12 u'\mathcal{I}_{\mathrm{C}}u + o_p(1)$ for every $u$, and since $L_n$ is concave, the basic corollary of \citet{hjort1993} gives $\sqrt{n}(\hat\theta^{\mathrm{C}} - \theta) = \mathcal{I}_{\mathrm{C}}^{-1}U_n + o_p(1)$. The $\tilde S_j$ are i.i.d.\ with mean zero, because $\E[b_{js} - p_j \mid \mathcal{G}_s] = 0$ on $\{e_j = 1\}$ by part (b), and with variance $\E[e_j q_j(Z_{js} - \bar Z_{c_j})(Z_{js} - \bar Z_{c_j})'] = \mathcal{I}_{\mathrm{C}}$, because the conditional variance of $b_{js}$ is $q_j$. The central limit theorem gives the limit $\N(0, \mathcal{I}_{\mathrm{C}}^{-1})$.

\textit{Step 5: variance estimation.} The negative Hessian of $\ell_C$ at $\theta'$ is the conditional covariance matrix of $\sum_{j \in C}b_{js}Z_{js}$ given $N_C$ under $\theta'$. By Lemma~\ref{lem:rejective}(b) with $z_j = Z_{js}$, $-\nabla^2\ell_C(\theta') = \sum_{j \in C}p_j(\theta')(1 - p_j(\theta'))(Z_{js} - \bar Z_C(\theta'))(Z_{js} - \bar Z_C(\theta'))' + O_p(1)$ uniformly for $\theta'$ in a neighborhood of $\theta$, with $\bar Z_C(\theta')$ the $p_j(1 - p_j)$-weighted mean of $Z_{js}$ over the class, because $\sigma_C^2(\theta') \ge \varrho n$ with probability approaching one for some $\varrho > 0$ and $n^{-1}\sum_j\|Z_{js}\|^2 = O_p(1)$ by Assumption~\ref{ass:regularity}(i). Evaluated at $\hat\theta^{\mathrm{C}}$, whose class intercepts $\gamma_C(\hat\theta^{\mathrm{C}})$ are consistent for the movements $d_c$ by Steps 3 and 4, the leading term divided by $n$ converges in probability to $\mathcal{I}_{\mathrm{C}}$ by the law of large numbers, which proves the consistency of the variance estimator. \qed

\subsection{Proof of Corollary~\ref{cor:bandwidth}}\label{sec:proof_bandwidth}

Let $H_n = [c_- n^{-1/4}, c_+ n^{-1/4}]$. With $T = 3$, the aggregated distance of a discordant pair is $\hat{\delta}_{ij1}^2$, and $2c_+ n^{-1/4} \le \bar\epsilon$ for $n$ large. By Lemma~\ref{lem:codegree_consist}(d), with probability approaching one, every pair that shares the initial outcome and is not equivalent on the support at $t = 1$ has $\hat{\delta}_{ij1}^2 \ge \bar\epsilon - O_p((\log n/n)^{1/2}) > c_+ n^{-1/4}$, and every equivalent pair has $\hat{\delta}_{ij1}^2 = O_p(\log n/n)$. Hence, with probability approaching one, the estimated weights $\hat{W}_{ij1}(h)$ at every $h \in H_n$ are positive exactly on the exact matches, and, by the Lipschitz property of $K_1$,
\[
\sup_{h \in H_n}\ \max_{i<j}\ \ind_{ij1}\big|\hat{W}_{ij1}(h)/\bar{K} - \ind^{\mathrm{ex}}_{ij1}\big| = O_p\!\left(\frac{\log n}{c_- n^{3/4}}\right) = o_p(n^{-1/2}).
\]
This is the conclusion of Step~1 of the proof of Theorem~\ref{thm:normality}, uniformly over $H_n$. Steps~2--5 and the proof for the variance estimator use the weights only through it, so they hold at $\hat{h}_1 \in H_n$. Likewise, the argument for Proposition~\ref{prop:class}(a) holds uniformly over $H_n$, because $O_p((\log n/n)^{1/2}) = o_p(c_- n^{-1/4})$, and part (c) uses the bandwidth only through part (a). The proof of part (b) uses only the $\mathcal{G}_1$-measurability of the feasible classes, which holds when $\hat{h}_1$ is a function of $D_1$ and the initial outcomes, both among the generators of $\mathcal{G}_1$. Finally, each inner average in~\eqref{eq:codegree_est} lies in $[-1, 1]$, so $0 \le S_n \le 1$, and the rule~\eqref{eq:h1rule} lies in $H_n$ with the stated constants. \qed

\section{The Social Influence Function: What Can and Cannot Be Recovered}\label{sec:appendix_lambda}

This appendix records what the model identifies about the social influence function once the agent effect is present, and describes an estimator of its common movement. The material is secondary to the main results, which treat $\lambda$, $f$, and the agent effects as nuisance components, and no asymptotic theory is claimed for the estimator described here.

\subsection{The Level of $\lambda$ Is Not Identified}

In the model without an agent effect, the social influence function can be recovered on the network classes that the codegree distance identifies, under the conditions of Proposition~\ref{prop:bridge} and with conditioning on the network type understood through those classes as in Corollary~\ref{cor:lambda}, by inverting the conditional choice probability: for an agent of network type $\varphi$ at period $t$, $\Prob(y_{it} = 1 \mid y_{i,t-1}, X_{it}, f_{w_{it}} = \varphi) = F(X_{it}'\beta + \alpha y_{i,t-1} + \lambda(w_{it}))$ when $\lambda$ is constant on the class of $\varphi$, and $F^{-1}$ of this probability minus the observed part of the index returns $\lambda(w_{it})$. With the agent effect, the same probability is
\[
\Prob(y_{it} = 1 \mid y_{i,t-1}, X_{it}, f_{w_{it}} = \varphi) = \E\Big[F\big(X_{it}'\beta + \alpha y_{i,t-1} + \lambda(w_{it}) + \mu_i\big) \,\Big|\, y_{i,t-1}, X_{it}, f_{w_{it}} = \varphi\Big],
\]
an average of $F$ over the conditional distribution of $\mu_i$ given the lag, the covariates, and the network type. Because $\mu_i$ may depend on $(w_i, X_i, y_{i0})$ arbitrarily (Assumption~\ref{ass:sampling}(i)) and the lagged outcome depends on $\mu_i$, this distribution is unrestricted, and inverting $F$ does not return $\lambda(w_{it})$. Two observations make the non-identification precise. First, replacing $(\lambda, \mu_i)$ by $(\lambda + c, \mu_i - c)$ for a constant $c$ leaves every outcome probability unchanged. Second, when the latent index is time-invariant, $w_{it} = w_i$, only the sum $\lambda(w_i) + \mu_i$ enters the outcome equation, and any function of $w_i$ can be moved between the two terms. The paper therefore makes no claim about the level of the social influence function, and the name refers to the component of the unobserved heterogeneity whose movement the network reveals, not to a function that the data recover.

\subsection{An Estimator of the Common Movement}

Corollary~\ref{cor:lambda} identifies the common movement $d_s(w_{is})$ of a reference agent $i$ at a transition period $s$ as the zero of a moment over the agents in $i$'s observable network class, or in its codegree neighborhoods: among agents who switch between $s$ and $s+1$, have stable covariates between $s+1$ and $s+2$, and are at zero (or small) codegree distance from $i$, the log odds of switching down rather than up equal $Z_{js}'\theta - d_s(w_{is})$, up to a remainder that vanishes with the neighborhood in continuous designs. Given the first-step estimator $\hat\theta = (\hat\beta', \hat\alpha)'$, the sample counterpart of~\eqref{eq:dident} and~\eqref{eq:dident_h} is a kernel-weighted logit with a single free parameter, in which the codegree kernel replaces membership in the class and the covariate kernel replaces exact stability. Let $\mathcal{T}_s$ be the set of switchers at $s$, and for $j \in \mathcal{T}_s$ with $y_{j,s-1} = y_{i,s-1}$ let $Z_{js}$ be the within-agent regressor defined before Corollary~\ref{cor:lambda} and $b_{js} = \ind\{y_j \in B_s\}$. Then $\hat{d}_s(w_{is})$ solves
\begin{equation}\label{eq:d_est}
\sum_{j \in \mathcal{T}_s:\ y_{j,s-1} = y_{i,s-1}} K_1\!\left(\frac{\hat{\delta}_{ijs}^2}{h_1}\right) K_2\!\left(\frac{\|X_{j,s+2} - X_{j,s+1}\|^2}{h_2}\right)\Big[b_{js} - F\big(Z_{js}'\hat\theta - d\big)\Big] = 0
\end{equation}
in $d$, which exists and is unique whenever the weighted switchers include both orderings, because the left side is strictly increasing in $d$. A cell in which no switcher receives positive weight, or in which all weighted switchers have the same ordering, returns no finite estimate, which is the sample counterpart of the support qualification in Corollary~\ref{cor:lambda}. In finite-type designs with discrete covariates and bandwidths below the separation gaps, the weights select the switchers of agent $i$'s own type, and~\eqref{eq:d_est} is the score equation for the intercept of a conditional logit within that cell with the slopes fixed at $\hat\theta$. The slopes can be taken from the class-conditional estimator~\eqref{eq:class_estimator}, which removes the class intercepts by conditioning on each class's number of switches of each kind and never estimates them, and~\eqref{eq:d_est} then recovers each intercept given those slopes. \label{loc:jointint}Estimating $\theta$ and the class intercepts jointly, by a logit with one dummy per class on the eligible agents, is a different procedure: it is first-order equivalent in finite-type designs, as the proof of Proposition~\ref{prop:class} shows, but it estimates one nuisance parameter per class from that class's members, and in the simulations of Section~\ref{sec:montecarlo} its small-sample bias is larger than that of exact conditioning (Table~\ref{tab:mc_joint}). Section~\ref{sec:continuous} explains why no distribution theory is available in continuous designs. In finite-type designs, a theory for~\eqref{eq:d_est} would have to account for the first-step estimation of $\theta$, whose error is of the same order as the cell-average error, and it is not developed here.

\newpage

\section{Extension to Ordered Outcomes}\label{sec:appendix_ordered}

The binary logit framework of the main text can be extended to accommodate ordered discrete outcomes. Many outcomes of interest in applied work are naturally ordinal: self-reported health takes values from ``poor'' to ``excellent,'' educational attainment is coded in ordered categories, and substance use, as in the Glasgow TFLS data where tobacco consumption is classified as non-smoker, occasional, or regular, is measured on ordinal scales. This section shows how the identification and estimation strategy extends to the dynamic ordered logit model with endogenous networks, using the threshold-dichotomization insight of \citet{muris2025} and \citet{honore2025ordered}. The extension presented here is deliberately minimal: it dichotomizes the outcome at the single, known state-dependence threshold $c$, which reduces the ordered model exactly to the binary framework of the preceding sections. A composite version, pooling the conditional likelihoods across all thresholds (whose dichotomized errors are perfectly dependent, since they share the same latent $\varepsilon_{it}$), is left for future work, and no formal results are claimed beyond those inherited from the binary case.

\subsection{Model}

Consider an ordered outcome $Y_{it} \in \{1, \ldots, J\}$ generated by a latent variable model:
\begin{equation}\label{eq:ordered_latent}
Y_{it}^* = X_{it}'\beta + \alpha^{\text{ord}}\, \ind\{Y_{it-1} \geq c\} + \lambda(w_{it}) + \mu_i - \varepsilon_{it},
\end{equation}
\begin{equation*}
Y_{it} = \begin{cases}
1 & \text{if }  Y_{it}^* < \varsigma_{2}, \\
2 & \text{if } \varsigma_{2} \leq Y_{it}^* < \varsigma_{3}, \\
\;\vdots & \\
J & \text{if } Y_{it}^* \geq \varsigma_{J},
\end{cases}
\end{equation*}
where $\varsigma_{2} < \varsigma_{3} < \cdots < \varsigma_{J}$ are unknown threshold parameters, $c \in \{2, \ldots, J\}$ is a known cutoff for the lagged dependent variable, $\alpha^{\text{ord}}$ is the autoregressive parameter, and all other quantities are as in the binary model. The error term $\varepsilon_{it}$ follows a standard logistic distribution, and $\lambda(w_{it})$ is the unknown social influence function of the latent index $w_{it}$, which also governs network formation through~\eqref{eq:network}. The state dependence specification $\alpha^{\text{ord}}\, \ind\{Y_{it-1} \geq c\}$ uses a binary indicator of whether the previous outcome exceeded a threshold, following \citet{muris2025}, and it nests the binary model when $J = 2$ and $c = 2$.

\subsection{Dichotomization and Identification}

The key insight, due to \citet{muris2025}, is that the ordered model can be analyzed through a collection of binary logit models obtained by dichotomizing the outcome at each threshold. I normalize $\varsigma_{c}=0$ without loss of generality, where $c$ is as in equation~\eqref{eq:ordered_latent}. Define the binary indicator:
\begin{equation*}
B_{it}(c) = \ind\{Y_{it} \geq c\}.
\end{equation*}
Under the logistic distributional assumption, $B_{it}(c)$ follows a binary logit model:
\begin{equation}\label{eq:binary_cutoff}
B_{it}(c) = \ind\{X_{it}'\beta + \alpha^{\text{ord}}\, B_{it-1}(c) + \lambda(w_{it}) + \mu_i - \varepsilon_{it} \geq 0\}.
\end{equation}

This has the same structure as the binary model in~\eqref{eq:outcome}, with $B_{it-1}(c)$ playing the role of the lagged outcome. For the network part of the argument to carry over, the link function must depend on the lagged outcome only through $B_{i,t-1}(c)$. If it depends on the full lagged outcome $Y_{i,t-1}$, the codegree distances would have to be stratified on $Y_{i,t-1}$, and the conditioning and swap-invariance arguments of Section~\ref{sec:identification} would have to be redone for that finer stratification, in which a discordant pair need not share the lagged outcome at the periods used for matching. That case is not covered by the binary proof and is not claimed here. The identification strategy of Section~\ref{sec:identification} therefore applies to the binary model at this specified cutoff, subject to the same assumptions: matching agents with the same network type eliminates the change of $\lambda(w_{it})$ between the transition periods.

The estimator replaces $y_{it}$ with $B_{it}(c)$ throughout the objective function~\eqref{eq:objective}, that is, in the discordant-pair sets $\mathcal{P}_s$, the conditioning variables $v_{ijs}$, and the contributions $m_{ijs}$. Let $\Omega_n^{\text{ord}}(b, a)$ denote the resulting objective. Information at the other thresholds does not enter it, so the estimator is the binary estimator applied to $B_{it}(c)$, and applying the results of Sections~\ref{sec:identification}--\ref{sec:asymptotics} requires their assumptions to hold for this dichotomized model. The consistency result and the finite-type distribution theory then carry over, and the limitations for continuous designs discussed in Section~\ref{sec:continuous} carry over unchanged. The estimator is:
\begin{equation*}
(\hat{\beta}^{\text{ord}}, \hat{\alpha}^{\text{ord}}) = \arg\max_{(b, a)} \; \Omega_n^{\text{ord}}(b, a).
\end{equation*}

Recovery of the other thresholds $\varsigma_{j}$, $j\ne c$, would require additional identification and estimation arguments for probabilities across cutoffs, which are not developed here. 

The ordered logit extension is directly relevant to the empirical application in Section~\ref{sec:empirical}. The TFLS records tobacco use on a three-point scale: non-smoker ($Y_{it} = 1$), occasional smoker ($Y_{it} = 2$), and regular smoker ($Y_{it} = 3$). The latent ordered model distinguishes these categories, but the proposed single-cutoff estimator uses only whether the student is above or below the chosen cutoff. With $J = 3$ and $c = 2$ (state dependence then operates through whether the student smoked at all in the previous wave), the binary reduction distinguishes non-use from any tobacco use and does not separately estimate the transition from occasional to regular use. This tobacco-only illustration is distinct from the tobacco-or-cannabis composite used in Section~\ref{sec:empirical}.

\section{The Observation Window in the Smoking Example}\label{sec:appendix_window}

\label{loc:timing}In the smoking example, Assumption~\ref{ass:persistence} describes a particular observation window. Consider a cohort observed at four dates, $t = 0, 1, 2, 3$, at comparable intervals, around a scheduled school transition between $t = 1$ and $t = 2$, for instance the move into a new stage of secondary school with new subjects, classes, and timetables. The outcome at $t = 0$ precedes the window and fixes the lag stratum. The outcome and the network at $t = 1$ precede the transition, and that network carries the information about each student's network class. The outcomes at $t = 2$ and $t = 3$ follow the transition, the second later within the same stage. A special case of Assumption~\ref{ass:persistence} is
\[
\lambda(w_{i1}) = \ell_{c(i)}, \qquad \lambda(w_{i2}) = \lambda(w_{i3}) = \ell_{c(i)} + \Delta_{c(i)},
\]
where $c(i)$ is student $i$'s network class before the transition, the baseline $\ell_c$ and the increment $\Delta_c$ may differ in sign and size across classes, and persistent differences between students of the same class are carried by the agent effect $\mu_i$. The display is a special case of the model when a latent index and a link function that represent it exist, for instance a finite latent state that records the class and whether the transition has occurred. Together with the model, the assumption then has three substantive parts. First, students who formed friendships in the same way before the transition share the change in their pull toward smoking that the transition brings, for example because the change operates through their position in the friendship structure, such as the salience of status or of a peer crowd, while their persistent predispositions can differ. Second, the adjustment is complete by $t = 2$, and the latent term does not move between $t = 2$ and $t = 3$, which is more plausible when the last two dates are close together within the same stage than when they are a year apart. Third, the state-dependence parameter is the same across the intervals, which requires comparable spacing of the dates. The second and third parts pull in different directions if only the last interval is shortened, so the design calls for dates that are close together and evenly spaced throughout, which in turn makes switches rarer. The realized network can still change within the window, because links depend on lagged smoking and link shocks change, so stability concerns the latent contribution to the outcome and not every friendship. A scheduled transition does not by itself imply any of these conditions: continued individual adjustment after $t = 2$, a latent propensity that responds to recent smoking, and a direct effect of friends' behavior all fall outside the model, and the example illustrates the design that the assumption requires rather than showing that school transitions satisfy it. If every increment $\Delta_c$ is zero, within-agent conditioning is valid and the network is not needed. If the increments differ across classes, a single common time effect does not remove them. In the Glasgow data of Section~\ref{sec:empirical}, the three waves are a year apart and were collected in the second, third, and fourth years of secondary school. In the Scottish system of the time, the third and fourth years were those of the Standard Grade courses \citep{mcvittie2008}, so the first two waves straddle a scheduled change of stage and the last two lie within one stage. The data nevertheless fall short of this design: the outcome before the window is not observed, and stability of the latent term over a full year is a strong requirement.

\section{Descriptive Fits for the Glasgow Data}\label{sec:appendix_glasgow}

This appendix reports the five descriptive fits referred to in Section~\ref{sec:localized}, together with the robustness exercises on the outcome definition and the initial condition. None of these fits is the estimator of Section~\ref{sec:estimation}, and no inferential claim is attached to any of them.

Five specifications are reported in Table~\ref{tab:any_friend_network}: the naive dynamic logit, with $y_{i,t-1}$ and all seven covariates (pocket money, romantic relationship, alcohol, leisure, female, age, and smoking at home); the dynamic logit with network controls, which adds peer averages of every covariate and of lagged smoking, computed on the any-friend network at each wave; a within-agent conditional logit in the spirit of \citet{honore2000panel}; a cross-agent pairwise conditional logit using the four time-varying covariates ($k = 4$); and the class-conditional logit~\eqref{eq:class_estimator} at the same unlocalized bandwidth, whose single feasible class contains the 22 switching students, so that it is the within-agent conditional logit with one common movement of the social influence term, conditioned out. The last two columns are not the proposed estimators, for the reasons given in Section~\ref{sec:localized}: the covariate-stability kernel is infeasible and the codegree kernel does not bind, so they omit both. The network-localized version is reported in Section~\ref{sec:localized}. No significance levels are attached to any column: the pooled standard errors ignore the dependence within students, and the conditional standard errors are not justified by any result in this paper.

The three conditional columns are reported with Firth penalization \citep{firth1993}: the criterion is the log-likelihood plus one half of the log determinant of its information matrix, computed as if the agents (column 3) or the pairs (column 4) were independent observations, and, for column 5, the information of the class-conditional likelihood, as in the penalized conditional logistic regression of \citet{heinze2010}. Penalization is used because the unpenalized likelihoods are near-separated: a linear-programming check shows that none of them is separated in the sense of \citet{albert1984}, so the three unpenalized maxima exist, but many fitted probabilities at the within-agent and pairwise maxima are extreme. For a pairwise criterion in which pairs sharing a student are dependent, the bias-reduction interpretation of Firth's penalty for independent observations does not carry over, and the standard errors reported for the penalized fits are not justified by any result in this paper and are reported for scale only.

\begin{table}[!ht]
\centering
\caption{Empirical estimates: Any friend network.}
\label{tab:any_friend_network}
\smallskip
\small
\begin{tabular}{l ccccc}
\toprule
 & (1) & (2) & (3) & (4) & (5) \\
 & Naive & Controls & HK & Pairwise & Class \\
\midrule
Intercept & $\underset{(2.779)}{-4.931}$ & $\underset{(3.523)}{-8.039}$ & --- & --- & --- \\[3pt]
Pocket money & $\underset{(0.019)}{0.014}$ & $\underset{(0.020)}{0.010}$ & $\underset{(0.046)}{0.062}$ & $\underset{(0.050)}{0.099}$ & $\underset{(0.064)}{0.059}$ \\[3pt]
Romantic relationship & $\underset{(0.346)}{0.815}$ & $\underset{(0.381)}{0.543}$ & $\underset{(1.731)}{2.423}$ & $\underset{(2.272)}{4.943}$ & $\underset{(2.365)}{2.741}$ \\[3pt]
Alcohol & $\underset{(0.174)}{0.808}$ & $\underset{(0.185)}{0.800}$ & $\underset{(0.414)}{1.501}$ & $\underset{(0.858)}{3.253}$ & $\underset{(0.988)}{1.698}$ \\[3pt]
Leisure & $\underset{(0.532)}{0.371}$ & $\underset{(0.564)}{0.087}$ & $\underset{(1.712)}{2.436}$ & $\underset{(2.740)}{5.403}$ & $\underset{(2.713)}{2.881}$ \\[3pt]
Female & $\underset{(0.315)}{-0.773}$ & $\underset{(0.638)}{-0.021}$ & --- & --- & --- \\[3pt]
Age & $\underset{(0.195)}{0.070}$ & $\underset{(0.245)}{0.285}$ & --- & --- & --- \\[3pt]
Smoking at home & $\underset{(0.303)}{0.138}$ & $\underset{(0.324)}{0.207}$ & --- & --- & --- \\[3pt]
Lagged smoking ($\hat\alpha$) & $\underset{(0.420)}{2.222}$ & $\underset{(0.444)}{2.101}$ & $\underset{(0.668)}{1.827}$ & $\underset{(1.368)}{1.413}$ & $\underset{(1.536)}{0.893}$ \\[3pt]
\midrule
Peer average covariates & --- & YES & --- & --- & --- \\
Firth penalization & --- & --- & YES & YES & YES \\
Observations & $288$ & $288$ & $22$ agents & $96$ pairs & $22$ agents \\
\bottomrule
\end{tabular}

\smallskip
\par\raggedright
{\footnotesize Standard errors in parentheses, reported for scale only. No significance levels are attached to any column.
Columns (1) and (2) are pooled dynamic logits with unclustered MLE standard errors on the $288$ student-wave observations, which ignore dependence within students. Column (2) adds peer averages of all covariates and of lagged smoking, whose coefficients are not reported. Column (3) is a within-agent conditional fit in the spirit of \citet{honore2000panel}, omitting $K_2$, on the transition event $y_{i1} + y_{i2} = 1$ with regressors $X_{i1} - X_{i2}$ and $y_{i0} - y_{i3}$. Column (4) is the cross-agent pairwise conditional logit obtained as the unlocalized limit of the proposed estimator (all $96$ discordant lag-matched pairs at unit weight, no codegree kernel, and no $K_2$). It does not use the network and is identical across network definitions. Column (5) is the class-conditional logit~\eqref{eq:class_estimator} at the same unlocalized bandwidth: the only feasible class is the lag stratum $y_{i0} = 0$, which contains all $22$ switchers, $6$ of them with the ordering $B_1$, and the likelihood conditions on that number, so it also does not use the network. Columns (3) and (4) use sandwich standard errors clustered on the $22$ discordant students, and column (5) the inverse of the class-conditional information at the penalized estimate, so their precision is not comparable with that of columns (1) and (2). The conditional columns have no intercept and no coefficients on female or smoking at home, which are time-invariant and difference out. Age rises by one year per wave for every student, so its difference is a constant: it cancels in the cross-agent column (4) and in the class-conditional column (5), while column (3) omits it, which amounts to setting the common time effect of the within-agent ordering equation to zero. All three are reported with the Firth penalization described in the text. No unpenalized conditional likelihood is separated, as a linear-programming check confirms (for column 5, with the class intercept), so the unpenalized maxima exist. The HK and pairwise likelihoods are near-separated, with $50\%$ of HK fitted probabilities and $76\%$ of column (4)'s pair-level fitted probabilities outside $[0.01, 0.99]$, although extreme fitted probabilities alone do not establish separation. The unpenalized estimates of $\alpha$ are $3.884$ (HK, standard error $2.018$), $2.281$ (pairwise, standard error $1.823$), and $2.134$ (class, standard error $2.411$). Penalization shrinks every conditional coefficient, by roughly a third to a half in columns (3) and (4) and by more than half in column (5), and leaves the sign pattern intact.}\end{table}

The state dependence estimates fall as more conditioning is imposed. The naive and controls estimates of $\alpha$ are 2.222 and 2.101. The within-agent fit returns 1.827, the cross-agent pairwise fit 1.413, with a standard error for scale of 1.368, and the class-conditional fit 0.893, with a standard error for scale of 1.536. The unpenalized class-conditional value, 2.134, is close to the unpenalized pairwise value, 2.281. The two maximize the full and the pairwise conditional likelihoods of the same single class, but with 22 agents and the $\Gamma$ remainder of equation~\eqref{eq:worked} present in both objectives their closeness is descriptive only, and the gap between their penalized values reflects the different penalties. The last two steps are not the mechanism this paper proposes: columns (4) and (5) compare all switchers with one another without using the network, which allows a common movement of $\lambda$ for all students but not a movement that differs across network classes, so the move from 1.827 to 1.413 and 0.893 reflects a change from within-agent to cross-agent conditioning, not the removal of the class-specific movement that codegree matching targets, and it carries no theoretical claim to remove confounding that HK leaves behind. Coefficient differences do not measure how much confounding each method removes, since the estimators also differ in conditioning and penalization, and the last step should not be read as evidence about the network channel, whose feasibility in these data is the subject of Section~\ref{sec:localized} and Table~\ref{tab:bandwidth}.

Three qualifications apply, because the data alone cannot support the ordering. The estimates are imprecise, and their pairwise differences are untested. The ordering depends on the penalization: without it, the within-agent estimate rises to 3.884 and the monotonicity breaks, so the declining sequence is an artifact of the penalization. And $n = 96$ is far below the sample sizes of Section~\ref{sec:montecarlo}, so nothing there implies that column (4) or column (5) is more reliable than column (3) \emph{in this sample}. The ordering of the point estimates is descriptive, and these data cannot adjudicate between the conditional fits.

Where the estimators differ most is in the covariate effects, and there the ordering runs the other way. Across all four covariates, the three conditional estimators produce coefficients larger in magnitude than the naive and controls specifications, and column (4) produces the largest. The larger magnitudes are not evidence that confounding was removed. In a logit model, omitted heterogeneity that is independent of the covariates shrinks pooled coefficients toward zero, so conditional estimates are larger in magnitude even without any confounding. Confounding does not have a predictable sign here either: the pooled coefficient combines the direction of the correlation between the omitted heterogeneity and the covariate with the attenuation from averaging over that heterogeneity, and positive correlation, as under the motivating story in which a taste for risk raises both alcohol use and smoking, is compatible with a pooled coefficient below the structural one. With the stability kernel omitted, the $\Gamma$ remainder of equation~\eqref{eq:worked} generally still depends on the level of the individual effect and of $\lambda$, so none of the conditional columns delivers the cancellation of Lemma~\ref{lem:ratio}, and columns (4) and (5), which do not use the network, cannot be credited with removing the class-specific part of the movement of $\lambda$. The differences between the columns are differences between objectives. The near-separation diagnostics give a further reason for caution: at the unpenalized maxima, half of the HK fitted probabilities and 76\% of column (4)'s pair-level fitted probabilities lie outside $[0.01, 0.99]$, a configuration in which conditional logit coefficients are known to inflate. The conditional columns are therefore reported under Firth penalization, which shrinks every conditional coefficient, by roughly a third to a half in columns (3) and (4) and by more than half in column (5), while leaving the sign pattern intact. The unpenalized magnitudes should not be quoted as effect sizes, since odds-ratio translations of coefficients near 5 on 96 pairs are not meaningful.

The peer-average controls change the alcohol and state-dependence coefficients little, by 0.008 and 0.12, respectively. This parallels the simulations of Section~\ref{sec:montecarlo}, in which peer averages remove little of the confounding when neighbors' characteristics do not proxy for an agent's own: peer averages recover only the part of the confounding that is reflected in neighbors' \emph{observables}, and small coefficient movements do not establish how much confounding remains.

Alcohol consumption is the strongest of the covariate results: the coefficient is positive in every column, moderate in the naive and controls specifications (0.80--0.81), and larger in the conditional fits (1.50, 3.25, and 1.70, under penalization). Its different magnitudes cannot by themselves establish a causal alcohol effect or the direction of omitted-variable bias, and alcohol use is among the covariates for which the strict exogeneity of Assumption~\ref{ass:sampling}(iv) is least plausible.

The remaining covariates show the same directional pattern with less strength. Leisure is small relative to its standard error in both pooled specifications and positive in the three conditional fits. Because the index rises as reported activity falls, the positive coefficient associates less frequent reported leisure activity with use. Romantic relationship is positive in the naive specification, smaller once peer averages are added, and positive in the three conditional fits. The pocket money coefficient, per pound per month, is imprecisely estimated in every column. The differences between columns (3) and (4) are within the imprecision of the fits and are not evidence that cross-agent pairing detects effects that within-agent conditioning misses.

Among the additional covariates available to the naive and controls specifications, the female coefficient is negative in the naive logit and close to zero once peer averages are included. Age, which increases by one year per wave by construction, has a small coefficient in both specifications, which does not by itself establish whether age or other time-varying factors account for the prevalence trend. Smoking at home is positive but small relative to its standard error in both specifications.

\begin{table}[!ht]
\centering
\caption{Robustness of the empirical estimates to the outcome definition and the initial condition.}
\label{tab:robust}
\smallskip
\footnotesize
\setlength{\tabcolsep}{5pt}
\begin{tabular}{l cccc}
\toprule
 & Baseline & Tobacco only & $y_{i0} = 1$ for wave-1 & $y_{i0} = y_{i1}$ \\
 & (any use, $y_{i0} = 0$) & & regular smokers & \\
\midrule
Students / discordant / pairs & 96 / 22 / 96 & 97 / 16 / 48 & 96 / 22 / 96 & 96 / 22 / 0 \\
\addlinespace[4pt]
\multicolumn{5}{l}{\textit{Lagged smoking ($\hat\alpha$)}} \\[2pt]
\quad Naive & $\underset{(0.420)}{2.222}$ & $\underset{(0.480)}{2.816}$ & $\underset{(0.406)}{2.438}$ & $\underset{(0.459)}{3.893}$ \\[6pt]
\quad Controls & $\underset{(0.444)}{2.101}$ & $\underset{(0.521)}{2.447}$ & $\underset{(0.426)}{2.299}$ & $\underset{(0.473)}{3.743}$ \\[6pt]
\quad HK & $\underset{(0.668)}{1.827}$ & $\underset{(0.874)}{1.001^{\ddagger}}$ & $\underset{(0.668)}{1.827}$ & $\underset{(0.502)}{2.873^{\dagger}}$ \\[6pt]
\quad Pairwise & $\underset{(1.368)}{1.413}$ & $\underset{(1.683)}{-1.983^{\dagger}}$ & $\underset{(1.368)}{1.413}$ & --- \\[6pt]
\quad Class & $\underset{(1.536)}{0.893}$ & $\underset{(2.981)}{-1.274^{\dagger}}$ & $\underset{(1.536)}{0.893}$ & --- \\[6pt]
\addlinespace[2pt]
\multicolumn{5}{l}{\textit{Alcohol}} \\[2pt]
\quad Naive & $\underset{(0.174)}{0.808}$ & $\underset{(0.182)}{0.805}$ & $\underset{(0.179)}{0.763}$ & $\underset{(0.221)}{0.917}$ \\[6pt]
\quad HK & $\underset{(0.414)}{1.501}$ & $\underset{(0.555)}{1.308^{\ddagger}}$ & $\underset{(0.414)}{1.501}$ & $\underset{(0.457)}{1.241^{\dagger}}$ \\[6pt]
\quad Pairwise & $\underset{(0.858)}{3.253}$ & $\underset{(0.821)}{2.304^{\dagger}}$ & $\underset{(0.858)}{3.253}$ & --- \\[6pt]
\quad Class & $\underset{(0.988)}{1.698}$ & $\underset{(1.258)}{1.768^{\dagger}}$ & $\underset{(0.988)}{1.698}$ & --- \\[6pt]
\bottomrule
\end{tabular}

\smallskip
\par\raggedright
{\footnotesize Standard errors in parentheses. HK, pairwise, and class rows are under Firth penalization as in Table~\ref{tab:any_friend_network}, where the class fit uses the single unlocalized class. ``Tobacco only'' redefines the outcome as any tobacco use, dropping cannabis. One student with complete tobacco records but a missing cannabis record re-enters the complete-case sample, which is why that column has 97 students. The third column sets $y_{i0} = 1$ for the ten students who report regular tobacco use at wave~1 and zero otherwise, and the fourth sets $y_{i0} = y_{i1}$, so that every wave-1 user is assumed to have been a user one year earlier. $^{\dagger}$: the unpenalized likelihood is completely separated, as certified by a linear program (a direction orders every row correctly, and the largest minimum margin over unit-norm directions is 0.235 for the pairwise fit, 0.109 for the class fit with its class intercept, and 0.083 for the within-agent fit), so it has no maximum and the penalized value is reported for completeness only. $^{\ddagger}$: the unpenalized within-agent likelihood is quasi-completely separated, because only one of the 16 switching students changes romantic status between waves 1 and 2, so its maximum is not attained, and the penalized value is reported for completeness only. The remaining conditional fits in the table are not separated by the same check, so their unpenalized maxima exist. ---: no discordant lag-matched pair, or no class that contains both orderings, exists, so the estimator is not defined.}
\end{table}

Table~\ref{tab:robust} reports the robustness exercises that the initial-condition and outcome-definition assumptions call for. Under a tobacco-only outcome, the discordant sample falls from 22 to 16 students and the pair count from 96 to 48. The within-agent state dependence estimate drops to 1.00 (standard error 0.87) under penalization, although the unpenalized within-agent likelihood is quasi-completely separated, because only one of the 16 switching students changes romantic status between waves 1 and 2. The pairwise likelihood is completely separated, as a linear-programming certificate confirms, so its penalized estimate of $-1.98$ carries no information, and the class-conditional likelihood is completely separated as well, so its penalized estimate of $-1.27$ carries no information either. The positive state dependence in the conditional fits therefore rests on the composite tobacco-or-cannabis outcome, whose persistence may partly reflect the recording of past cannabis initiation. The initial-condition exercises are more informative about what the assumption does than about the estimates. Setting $y_{i0} = 1$ for the ten wave-1 regular smokers leaves the three conditional fits unchanged to three decimals, because none of those ten students is among the 22 who transition between waves 1 and 2, so the imputation never enters a conditional likelihood, and only the pooled columns move. Setting $y_{i0} = y_{i1}$ instead, the most pessimistic reading in which every wave-1 user already used a year earlier, leaves the pairwise estimator with no pairs at all: lag matching requires $y_{i0} = y_{j0}$ while discordance requires $y_{i1} \neq y_{j1}$, and the two cannot hold together once $y_{i0} = y_{i1}$. For the same reason, each lag stratum then contains switchers of a single ordering, and the class-conditional estimator is not defined. It also separates the within-agent likelihood completely: with $y_{i0} = y_{i1}$, the lagged-outcome regressor $y_{i0} - y_{i3}$, oriented by the observed switch, is never negative, and together with the covariates it orders every switching student correctly, so the within-agent value of 2.87 is a penalized value without an unpenalized counterpart. In this sense, the initial condition is load-bearing here. With three waves the estimator has exactly the four outcomes per student that the transition event needs, and the first of them must be supplied by assumption. An application in which the paper's mechanism could be tested needs at least four waves, so that $y_{i0}$ is observed, together with a network observed without a nomination cap.

\clearpage
\section{Supplementary Tables}\label{sec:appendix_figures}

This appendix reports the Monte Carlo results that are summarized but not tabulated in Section~\ref{sec:montecarlo}: the finite-type design with lag-dependent links (Tables~\ref{tab:mc_sbm4dyn_beta} and~\ref{tab:mc_sbm4dyn_alpha}), the benchmark designs at $n = 1000$ (Table~\ref{tab:mc_bench_1000}), the matching and inference diagnostics for the benchmark designs (Table~\ref{tab:mc_diag_bench}), the comparisons on common successful replications (Table~\ref{tab:mc_common}), the pairwise and class-conditional estimators in the benchmark and continuous designs (Table~\ref{tab:mc_class_bench}), the class-based alternatives (Table~\ref{tab:mc_joint}), and the continuous design (Tables~\ref{tab:mc_cont_beta} and~\ref{tab:mc_cont_alpha}). The simulation code, the population constants it uses, and the replication-level output are available with the paper.

\begin{table}[!ht]
\centering
\caption{Monte Carlo results for $\beta$: finite-type design with lag-dependent links}\label{tab:mc_sbm4dyn_beta}
\smallskip
\footnotesize
\setlength{\tabcolsep}{4.5pt}
\begin{tabular}{l cccccccc}
\toprule
$n$ & Naive & Controls & HK & Proposed & Oracle & Class & Class orc. & Infeas. \\
\midrule
\multicolumn{9}{l}{\textit{Bias}} \\[2pt]
500 & 0.245 & 0.240 & 0.372 & 0.234 & 0.230 & 0.145 & 0.135 & 0.007 \\
1000 & 0.239 & 0.233 & 0.285 & 0.143 & 0.143 & 0.084 & 0.084 & 0.000 \\
2000 & 0.242 & 0.234 & 0.264 & 0.068 & 0.068 & 0.047 & 0.047 & 0.004 \\
3000 & 0.239 & 0.231 & 0.251 & 0.043 & 0.043 & 0.033 & 0.033 & 0.000 \\
\addlinespace[4pt]
\multicolumn{9}{l}{\textit{Standard deviation}} \\[2pt]
500 & 0.069 & 0.069 & 0.468 & 0.632 & 0.619 & 0.528 & 0.523 & 0.073 \\
1000 & 0.052 & 0.052 & 0.295 & 0.411 & 0.412 & 0.344 & 0.344 & 0.057 \\
2000 & 0.035 & 0.035 & 0.200 & 0.262 & 0.262 & 0.227 & 0.227 & 0.038 \\
3000 & 0.030 & 0.030 & 0.156 & 0.187 & 0.187 & 0.169 & 0.169 & 0.032 \\
\addlinespace[4pt]
\multicolumn{9}{l}{\textit{Root mean squared error}} \\[2pt]
500 & 0.254 & 0.250 & 0.598 & 0.673 & 0.660 & 0.547 & 0.540 & 0.074 \\
1000 & 0.245 & 0.239 & 0.410 & 0.435 & 0.435 & 0.354 & 0.354 & 0.057 \\
2000 & 0.245 & 0.237 & 0.331 & 0.270 & 0.271 & 0.232 & 0.232 & 0.038 \\
3000 & 0.241 & 0.233 & 0.295 & 0.191 & 0.191 & 0.172 & 0.172 & 0.032 \\
\addlinespace[4pt]
\multicolumn{9}{l}{\textit{Coverage of nominal 95\% interval}} \\[2pt]
500 & 0.050 & 0.068 & 0.964 & 0.965 & 0.967 & 0.984 & 0.981 & 0.956 \\
1000 & 0.002 & 0.002 & 0.890 & 0.950 & 0.952 & 0.968 & 0.970 & 0.938 \\
2000 & 0.000 & 0.000 & 0.812 & 0.948 & 0.948 & 0.960 & 0.960 & 0.954 \\
3000 & 0.000 & 0.000 & 0.686 & 0.950 & 0.950 & 0.962 & 0.962 & 0.948 \\
\addlinespace[4pt]
\multicolumn{9}{l}{\textit{Replications without an estimate}} \\[2pt]
500 & 0 & 0 & 3 & 75 & 76 & 74 & 76 & 0 \\
1000 & 0 & 0 & 1 & 1 & 1 & 1 & 1 & 0 \\
2000 & 0 & 0 & 0 & 0 & 0 & 0 & 0 & 0 \\
3000 & 0 & 0 & 0 & 0 & 0 & 0 & 0 & 0 \\
\bottomrule
\end{tabular}
\vspace{4pt}\par\raggedright\footnotesize \textit{Notes:} $R = 500$ replications; $\beta_0 = 1$, $\alpha_0 = 0.5$. Naive: pooled dynamic logit. Controls: pooled dynamic logit with neighbors' average covariate and lagged outcome. HK: within-agent conditional logit of \citet{honore2000panel}. Proposed: estimator~\eqref{eq:estimator}. Oracle: the same pairwise estimator with population matches. Class: the class-conditional estimator~\eqref{eq:class_estimator}, with classes formed from the same estimated distances and bandwidth. Class orc.: the same likelihood with the population classes. Infeas.: pooled dynamic logit including the true unobserved term $\lambda(w_{it}) + \mu_i$. Bandwidths: $h_1 = 0.5\,\bar\epsilon\,(n/500)^{-1/4}$, with $\bar\epsilon$ the population separation gap of the design, a fixed constant computed from the block matrix and the population shares of the $(\text{type}, y_{i0})$ cells, and exact covariate stability ($h_2 = 0.5$ on the integer covariate). Coverage: share of replications in which the nominal 95\% Wald interval contains the true value. Standard errors are MLE for the pooled logits, a sandwich for HK, the sandwich clustering pair-level scores by agent for the pairwise estimators, and the inverse of the negative Hessian of the class-conditional log likelihood for the class estimators. Statistics are computed over the replications with an estimate, and failures are defined in Section~\ref{sec:mc_bandwidths}.
\end{table}

\begin{table}[!ht]
\centering
\caption{Monte Carlo results for $\alpha$: finite-type design with lag-dependent links}\label{tab:mc_sbm4dyn_alpha}
\smallskip
\footnotesize
\setlength{\tabcolsep}{4.5pt}
\begin{tabular}{l cccccccc}
\toprule
$n$ & Naive & Controls & HK & Proposed & Oracle & Class & Class orc. & Infeas. \\
\midrule
\multicolumn{9}{l}{\textit{Bias}} \\[2pt]
500 & 0.715 & 0.666 & 0.430 & $-$0.055 & $-$0.053 & $-$0.046 & $-$0.046 & $-$0.005 \\
1000 & 0.711 & 0.650 & 0.425 & 0.008 & 0.008 & 0.005 & 0.005 & $-$0.005 \\
2000 & 0.711 & 0.643 & 0.400 & 0.000 & 0.000 & 0.001 & 0.001 & 0.001 \\
3000 & 0.715 & 0.644 & 0.363 & 0.004 & 0.004 & 0.009 & 0.009 & 0.004 \\
\addlinespace[4pt]
\multicolumn{9}{l}{\textit{Standard deviation}} \\[2pt]
500 & 0.126 & 0.141 & 0.752 & 1.427 & 1.421 & 1.206 & 1.226 & 0.145 \\
1000 & 0.093 & 0.111 & 0.498 & 0.917 & 0.916 & 0.792 & 0.792 & 0.102 \\
2000 & 0.070 & 0.079 & 0.314 & 0.551 & 0.551 & 0.472 & 0.472 & 0.073 \\
3000 & 0.056 & 0.064 & 0.269 & 0.415 & 0.415 & 0.391 & 0.391 & 0.060 \\
\addlinespace[4pt]
\multicolumn{9}{l}{\textit{Root mean squared error}} \\[2pt]
500 & 0.726 & 0.680 & 0.865 & 1.427 & 1.420 & 1.205 & 1.225 & 0.145 \\
1000 & 0.717 & 0.659 & 0.655 & 0.917 & 0.915 & 0.791 & 0.791 & 0.103 \\
2000 & 0.715 & 0.648 & 0.508 & 0.550 & 0.550 & 0.472 & 0.472 & 0.073 \\
3000 & 0.717 & 0.647 & 0.452 & 0.415 & 0.415 & 0.391 & 0.391 & 0.060 \\
\addlinespace[4pt]
\multicolumn{9}{l}{\textit{Coverage of nominal 95\% interval}} \\[2pt]
500 & 0.000 & 0.006 & 0.909 & 0.941 & 0.946 & 0.984 & 0.981 & 0.962 \\
1000 & 0.000 & 0.000 & 0.876 & 0.942 & 0.942 & 0.952 & 0.952 & 0.954 \\
2000 & 0.000 & 0.000 & 0.790 & 0.938 & 0.938 & 0.950 & 0.950 & 0.944 \\
3000 & 0.000 & 0.000 & 0.702 & 0.946 & 0.946 & 0.950 & 0.950 & 0.956 \\
\addlinespace[4pt]
\multicolumn{9}{l}{\textit{Replications without an estimate}} \\[2pt]
500 & 0 & 0 & 3 & 75 & 76 & 74 & 76 & 0 \\
1000 & 0 & 0 & 1 & 1 & 1 & 1 & 1 & 0 \\
2000 & 0 & 0 & 0 & 0 & 0 & 0 & 0 & 0 \\
3000 & 0 & 0 & 0 & 0 & 0 & 0 & 0 & 0 \\
\bottomrule
\end{tabular}
\vspace{4pt}\par\raggedright\footnotesize \textit{Notes:} $R = 500$ replications; $\beta_0 = 1$, $\alpha_0 = 0.5$. Naive: pooled dynamic logit. Controls: pooled dynamic logit with neighbors' average covariate and lagged outcome. HK: within-agent conditional logit of \citet{honore2000panel}. Proposed: estimator~\eqref{eq:estimator}. Oracle: the same pairwise estimator with population matches. Class: the class-conditional estimator~\eqref{eq:class_estimator}, with classes formed from the same estimated distances and bandwidth. Class orc.: the same likelihood with the population classes. Infeas.: pooled dynamic logit including the true unobserved term $\lambda(w_{it}) + \mu_i$. Bandwidths: $h_1 = 0.5\,\bar\epsilon\,(n/500)^{-1/4}$, with $\bar\epsilon$ the population separation gap of the design, a fixed constant computed from the block matrix and the population shares of the $(\text{type}, y_{i0})$ cells, and exact covariate stability ($h_2 = 0.5$ on the integer covariate). Coverage: share of replications in which the nominal 95\% Wald interval contains the true value. Standard errors are MLE for the pooled logits, a sandwich for HK, the sandwich clustering pair-level scores by agent for the pairwise estimators, and the inverse of the negative Hessian of the class-conditional log likelihood for the class estimators. Statistics are computed over the replications with an estimate, and failures are defined in Section~\ref{sec:mc_bandwidths}.
\end{table}

\begin{table}[!ht]
\centering
\caption{Benchmark designs at $n = 1000$}\label{tab:mc_bench_1000}
\smallskip
\footnotesize
\setlength{\tabcolsep}{3.2pt}
\resizebox{\textwidth}{!}{%
\begin{tabular}{l ccc ccc ccc ccc }
\toprule
 & \multicolumn{3}{c}{HK} & \multicolumn{3}{c}{Proposed} & \multicolumn{3}{c}{Oracle} & \multicolumn{3}{c}{Class} \\
\cmidrule(lr){2-4} \cmidrule(lr){5-7} \cmidrule(lr){8-10} \cmidrule(lr){11-13}
Design & Bias & RMSE & Cov. & Bias & RMSE & Cov. & Bias & RMSE & Cov. & Bias & RMSE & Cov. \\
\midrule
\multicolumn{13}{l}{\textit{$\beta$ ($\beta_0 = 1$)}} \\[2pt]
\multicolumn{13}{l}{\quad\textit{Designs that satisfy the assumptions}} \\[1pt]
Baseline & 0.296 & 0.433 & 0.88 & 0.151 & 0.458 & 0.95 & 0.151 & 0.459 & 0.95 & 0.085 & 0.381 & 0.97 \\
Lag-dependent links & 0.285 & 0.410 & 0.89 & 0.143 & 0.435 & 0.95 & 0.143 & 0.435 & 0.95 & 0.084 & 0.354 & 0.97 \\
Fixed types & 0.075 & 0.315 & 0.94 & 0.132 & 0.397 & 0.96 & 0.132 & 0.397 & 0.96 & 0.097 & 0.351 & 0.96 \\
Weak network & 0.277 & 0.403 & 0.92 & -- & -- & -- & 0.133 & 0.451 & 0.95 & -- & -- & -- \\
\multicolumn{13}{l}{\quad\textit{Designs outside the assumptions}} \\[1pt]
Idiosyncratic movement & 0.265 & 0.398 & 0.91 & 0.190 & 0.474 & 0.95 & 0.190 & 0.472 & 0.95 & 0.151 & 0.391 & 0.97 \\
Covariates in links & 0.257 & 0.380 & 0.91 & 0.223 & 0.733 & 0.97 & 0.132 & 0.453 & 0.95 & 0.119 & 0.511 & 0.97 \\
Sparse network & 0.306 & 0.438 & 0.88 & 0.857 & 1.336 & 0.85 & 0.166 & 0.467 & 0.93 & 0.445 & 0.671 & 0.98 \\
\addlinespace[4pt]
\multicolumn{13}{l}{\textit{$\alpha$ ($\alpha_0 = 0.5$)}} \\[2pt]
\multicolumn{13}{l}{\quad\textit{Designs that satisfy the assumptions}} \\[1pt]
Baseline & 0.377 & 0.614 & 0.90 & 0.000 & 0.872 & 0.94 & 0.000 & 0.872 & 0.94 & 0.003 & 0.757 & 0.96 \\
Lag-dependent links & 0.425 & 0.655 & 0.88 & 0.008 & 0.917 & 0.94 & 0.008 & 0.915 & 0.94 & 0.005 & 0.791 & 0.95 \\
Fixed types & $-$0.004 & 0.466 & 0.96 & $-$0.037 & 0.719 & 0.97 & $-$0.036 & 0.719 & 0.97 & $-$0.041 & 0.645 & 0.97 \\
Weak network & 0.395 & 0.601 & 0.90 & -- & -- & -- & 0.018 & 0.849 & 0.95 & -- & -- & -- \\
\multicolumn{13}{l}{\quad\textit{Designs outside the assumptions}} \\[1pt]
Idiosyncratic movement & 0.365 & 0.588 & 0.90 & 0.137 & 0.786 & 0.95 & 0.136 & 0.784 & 0.95 & 0.132 & 0.686 & 0.97 \\
Covariates in links & 0.413 & 0.644 & 0.88 & 0.036 & 1.114 & 0.96 & 0.017 & 0.869 & 0.94 & $-$0.004 & 0.875 & 0.98 \\
Sparse network & 0.400 & 0.616 & 0.88 & 0.512 & 2.026 & 0.85 & 0.005 & 0.890 & 0.94 & 0.387 & 1.063 & 0.97 \\
\addlinespace[4pt]
\bottomrule
\end{tabular}}
\vspace{4pt}\par\raggedright\footnotesize \textit{Notes:} $R = 500$ replications. Bias, root mean squared error, and coverage of the nominal 95\% interval, computed over the replications with an estimate. HK: within-agent conditional logit of \citet{honore2000panel}. Proposed: estimator~\eqref{eq:estimator}. Oracle: the same pairwise estimator with population matches. Class: the class-conditional estimator~\eqref{eq:class_estimator} with classes formed from the same estimated distances and bandwidth. The sparse design uses the 0.10 quantile of the estimated distances as $h_1$, and all other finite-type designs use $h_1 = 0.5\,\bar\epsilon\,(n/500)^{-1/4}$ with $\bar\epsilon$ the population separation gap of the design (Section~\ref{sec:mc_bandwidths}). In the weak-network design at $n = 1000$, no pair has an estimated distance below the bandwidth, so the proposed and class estimators cannot be computed (--). The covariate-link and sparse designs violate the maintained network model (covariates in the link function and bounded degrees, respectively), and the idiosyncratic design violates Assumption~\ref{ass:persistence}(i). Failed fits are optimizer failures, all certified as separated, unless stated otherwise. Lag-dependent links: 1 each for Proposed, Oracle, HK, and Class. Weak network: Proposed 500 (fewer than three pairs) and Class 500 (fewer than three agents in classes with both orderings). Idiosyncratic movement: 1 each for Proposed, Oracle, and Class. Covariates in links: Proposed 18 and Class 12. Sparse network: Proposed 79 (78 certified as separated), Oracle 7, HK 1, and Class 21.
\end{table}

\begin{table}[!ht]
\centering
\caption{Benchmark designs: matching, contributing-agent, and inference diagnostics for the proposed estimator}\label{tab:mc_diag_bench}
\smallskip
\footnotesize
\setlength{\tabcolsep}{3.2pt}
\resizebox{\textwidth}{!}{%
\begin{tabular}{ll rr rrrr cc cccccc}
\toprule
 & & & & \multicolumn{4}{c}{Agents} & \multicolumn{2}{c}{Matching} & \multicolumn{6}{c}{Proposed estimator} \\
\cmidrule(lr){5-8} \cmidrule(lr){9-10} \cmidrule(lr){11-16}
Design & $n$ & Pairs & ESS & Elig. & Contrib. & Unm. & Min. & Prec. & Recall & Par. & Med.\ bias & SD & Mean SE & SD/SE & Coverage \\
\midrule
Fixed types & 1000 & 206 & 206 & 79 & 78 & 1 & 1.5 & 1.00 & 1.00 & $\beta$ & 0.103 & 0.375 & 0.344 & 1.09 & 0.960 \\
 & & & & & & & & &  & $\alpha$ & $-$0.037 & 0.719 & 0.718 & 1.00 & 0.966 \\
\addlinespace[2pt]
Fixed types & 2000 & 830 & 830 & 159 & 159 & 0 & 3.5 & 1.00 & 1.00 & $\beta$ & 0.024 & 0.219 & 0.213 & 1.03 & 0.958 \\
 & & & & & & & & &  & $\alpha$ & 0.005 & 0.453 & 0.452 & 1.00 & 0.954 \\
\addlinespace[2pt]
Weak network & 1000 & 0 & -- & 100 & 0 & 100 & -- & -- & -- & $\beta$ & -- & -- & -- & -- & -- \\
 & & & & & & & & &  & $\alpha$ & -- & -- & -- & -- & -- \\
\addlinespace[2pt]
Weak network & 2000 & 346 & 262 & 200 & 163 & 36 & 1.0 & 1.00 & 0.45 & $\beta$ & 0.076 & 0.346 & 0.299 & 1.15 & 0.940 \\
 & & & & & & & & &  & $\alpha$ & $-$0.060 & 0.719 & 0.656 & 1.10 & 0.926 \\
\addlinespace[2pt]
Weak network & 3000 & 1369 & 1214 & 299 & 289 & 10 & 1.0 & 1.00 & 0.79 & $\beta$ & 0.040 & 0.219 & 0.200 & 1.09 & 0.924 \\
 & & & & & & & & &  & $\alpha$ & 0.010 & 0.462 & 0.440 & 1.05 & 0.940 \\
\addlinespace[2pt]
Idiosyncratic movement & 1000 & 219 & 218 & 99 & 90 & 8 & 1.3 & 1.00 & 1.00 & $\beta$ & 0.118 & 0.434 & 0.364 & 1.19 & 0.950 \\
 & & & & & & & & &  & $\alpha$ & 0.101 & 0.775 & 0.776 & 1.00 & 0.952 \\
\addlinespace[2pt]
Idiosyncratic movement & 2000 & 889 & 889 & 198 & 195 & 3 & 2.0 & 1.00 & 1.00 & $\beta$ & 0.104 & 0.256 & 0.232 & 1.10 & 0.926 \\
 & & & & & & & & &  & $\alpha$ & 0.116 & 0.520 & 0.500 & 1.04 & 0.952 \\
\addlinespace[2pt]
Covariates in links & 1000 & 61 & 54 & 100 & 57 & 43 & 1.0 & 1.00 & 0.32 & $\beta$ & 0.123 & 0.699 & 0.678 & 1.03 & 0.969 \\
 & & & & & & & & &  & $\alpha$ & $-$0.012 & 1.114 & 1.075 & 1.04 & 0.959 \\
\addlinespace[2pt]
Covariates in links & 2000 & 245 & 222 & 200 & 153 & 47 & 1.0 & 1.00 & 0.32 & $\beta$ & 0.073 & 0.381 & 0.375 & 1.02 & 0.946 \\
 & & & & & & & & &  & $\alpha$ & $-$0.047 & 0.641 & 0.638 & 1.00 & 0.964 \\
\addlinespace[2pt]
Sparse network & 1000 & 122 & 85 & 100 & 47 & 53 & 1.0 & 0.17 & 0.11 & $\beta$ & 0.640 & 1.027 & 0.685 & 1.50 & 0.850 \\
 & & & & & & & & &  & $\alpha$ & 0.322 & 1.963 & 1.343 & 1.46 & 0.848 \\
\addlinespace[2pt]
Sparse network & 2000 & 496 & 344 & 201 & 103 & 98 & 1.0 & 0.17 & 0.11 & $\beta$ & 0.449 & 0.622 & 0.420 & 1.48 & 0.783 \\
 & & & & & & & & &  & $\alpha$ & 0.408 & 1.093 & 0.913 & 1.20 & 0.873 \\
\addlinespace[2pt]
\bottomrule
\end{tabular}}
\vspace{4pt}\par\raggedright\footnotesize \textit{Notes:} See the notes to Table~\ref{tab:mc_diag}. Failed fits of the proposed estimator are optimizer failures, all certified as separated, unless stated otherwise. Weak network: 500 at $n = 1000$, all with fewer than three pairs. Idiosyncratic movement: 1 at $n = 1000$. Covariates in links: 18 at $n = 1000$. Sparse network: 79 at $n = 1000$ (78 certified as separated) and 2 at $n = 2000$.
\end{table}

\begin{table}[!ht]
\centering
\caption{Comparisons on common successful replications}\label{tab:mc_common}
\smallskip
\footnotesize
\setlength{\tabcolsep}{3.2pt}
\resizebox{\textwidth}{!}{%
\begin{tabular}{ll rrr ccc rcc rcc}
\toprule
 & & \multicolumn{3}{c}{Successful fits} & \multicolumn{3}{c}{Own successes} & \multicolumn{3}{c}{Common with HK} & \multicolumn{3}{c}{Common with oracle} \\
\cmidrule(lr){3-5} \cmidrule(lr){6-8} \cmidrule(lr){9-11} \cmidrule(lr){12-14}
Design & $n$ & Prop. & HK & Orc. & Prop. & HK & Orc. & Fits & Prop. & HK & Fits & Prop. & Orc. \\
\midrule
\multicolumn{14}{l}{\textit{$\beta$}} \\[2pt]
Baseline & 500 & 431 & 488 & 431 & 0.677 & 0.568 & 0.670 & 431 & 0.677 & 0.506 & 431 & 0.677 & 0.670 \\
Lag-dependent links & 500 & 425 & 497 & 424 & 0.673 & 0.598 & 0.660 & 425 & 0.673 & 0.491 & 424 & 0.674 & 0.660 \\
Idiosyncratic movement & 1000 & 499 & 500 & 499 & 0.474 & 0.398 & 0.472 & 499 & 0.474 & 0.395 & 499 & 0.474 & 0.472 \\
Covariates in links & 1000 & 482 & 500 & 500 & 0.733 & 0.380 & 0.453 & 482 & 0.733 & 0.364 & 482 & 0.733 & 0.426 \\
Sparse network & 1000 & 421 & 499 & 493 & 1.336 & 0.438 & 0.467 & 421 & 1.336 & 0.399 & 419 & 1.334 & 0.438 \\
Sparse network & 2000 & 498 & 500 & 500 & 0.833 & 0.292 & 0.239 & 498 & 0.833 & 0.292 & 498 & 0.833 & 0.239 \\
Continuous & 500 & 495 & 500 & 494 & 1.808 & 0.947 & 1.655 & 495 & 1.808 & 0.950 & 493 & 1.797 & 1.657 \\
\addlinespace[4pt]
\multicolumn{14}{l}{\textit{$\alpha$}} \\[2pt]
Baseline & 500 & 431 & 488 & 431 & 1.300 & 0.825 & 1.297 & 431 & 1.300 & 0.758 & 431 & 1.300 & 1.297 \\
Lag-dependent links & 500 & 425 & 497 & 424 & 1.427 & 0.865 & 1.420 & 425 & 1.427 & 0.758 & 424 & 1.426 & 1.420 \\
Idiosyncratic movement & 1000 & 499 & 500 & 499 & 0.786 & 0.588 & 0.784 & 499 & 0.786 & 0.587 & 499 & 0.786 & 0.784 \\
Covariates in links & 1000 & 482 & 500 & 500 & 1.114 & 0.644 & 0.869 & 482 & 1.114 & 0.629 & 482 & 1.114 & 0.829 \\
Sparse network & 1000 & 421 & 499 & 493 & 2.026 & 0.616 & 0.890 & 421 & 2.026 & 0.564 & 419 & 2.029 & 0.859 \\
Sparse network & 2000 & 498 & 500 & 500 & 1.195 & 0.494 & 0.501 & 498 & 1.195 & 0.492 & 498 & 1.195 & 0.502 \\
Continuous & 500 & 495 & 500 & 494 & 1.176 & 0.497 & 1.131 & 495 & 1.176 & 0.494 & 493 & 1.175 & 1.131 \\
\addlinespace[4pt]
\bottomrule
\end{tabular}}
\vspace{4pt}\par\raggedright\footnotesize \textit{Notes:} $R = 500$ replications. A fit succeeds if it returns finite estimates and standard errors for both parameters. Successful fits: number of successes. Own successes: RMSE over each estimator's own successes, as in the other tables. Common with HK (oracle): RMSE of the proposed estimator and HK (the oracle) over the replications in which both succeed (Fits). Only cells whose success sets differ are shown, as elsewhere the two comparisons coincide. No comparison is possible in the weak-network design at $n = 1000$, where the proposed estimator never succeeds.
\end{table}

\begin{table}[!ht]
\centering
\caption{Pairwise and class-conditional estimators: benchmark and continuous designs}\label{tab:mc_class_bench}
\smallskip
\footnotesize
\setlength{\tabcolsep}{3.2pt}
\resizebox{\textwidth}{!}{%
\begin{tabular}{ll rrr ccc ccc cc cc cc}
\toprule
 & & \multicolumn{3}{c}{Successful fits} & \multicolumn{3}{c}{RMSE for $\beta$} & \multicolumn{3}{c}{RMSE for $\alpha$} & \multicolumn{2}{c}{Class: coverage} & \multicolumn{2}{c}{Classes} & \multicolumn{2}{c}{Time (ms)} \\
\cmidrule(lr){3-5} \cmidrule(lr){6-8} \cmidrule(lr){9-11} \cmidrule(lr){12-13} \cmidrule(lr){14-15} \cmidrule(lr){16-17}
Design & $n$ & Prop. & Class & Common & Prop. & Class & Class orc. & Prop. & Class & Class orc. & $\beta$ & $\alpha$ & Exact & Used & Prop. & Class \\
\midrule
\addlinespace[2pt]
Fixed types & 1000 & 500 & 500 & 500 & 0.397 & 0.351 & 0.351 & 0.719 & 0.645 & 0.645 & 0.962 & 0.970 & 99.6 & 78 & 0.6 & 11.1 \\
 & 2000 & 500 & 500 & 500 & 0.223 & 0.207 & 0.207 & 0.454 & 0.433 & 0.433 & 0.958 & 0.946 & 100.0 & 159 & 1.0 & 21.1 \\
\addlinespace[2pt]
Weak network & 1000 & 0 & 0 & 0 & -- & -- & -- & -- & -- & -- & -- & -- & 0.0 & 0 & 0.1 & 0.8 \\
 & 2000 & 500 & 500 & 500 & 0.362 & 0.246 & 0.236 & 0.721 & 0.522 & 0.513 & 0.940 & 0.934 & 1.6 & 187 & 0.8 & 21.3 \\
 & 3000 & 500 & 500 & 500 & 0.226 & 0.190 & 0.190 & 0.462 & 0.387 & 0.386 & 0.940 & 0.956 & 89.0 & 294 & 1.3 & 34.2 \\
\addlinespace[2pt]
Idiosyncratic movement & 1000 & 499 & 499 & 499 & 0.474 & 0.391 & 0.391 & 0.786 & 0.686 & 0.687 & 0.970 & 0.970 & 99.6 & 90 & 0.5 & 11.8 \\
 & 2000 & 500 & 500 & 500 & 0.287 & 0.253 & 0.253 & 0.534 & 0.493 & 0.493 & 0.942 & 0.942 & 100.0 & 195 & 1.0 & 24.4 \\
\addlinespace[2pt]
Covariates in links & 1000 & 482 & 488 & 482 & 0.733 & 0.499 & 0.342 & 1.114 & 0.862 & 0.744 & 0.969 & 0.984 & 0.0 & 67 & 0.5 & 17.8 \\
 & 2000 & 500 & 500 & 500 & 0.394 & 0.288 & 0.221 & 0.641 & 0.516 & 0.474 & 0.954 & 0.952 & 0.0 & 173 & 0.8 & 20.3 \\
\addlinespace[2pt]
Sparse network & 1000 & 421 & 479 & 419 & 1.334 & 0.604 & 0.362 & 2.029 & 1.014 & 0.763 & 0.977 & 0.967 & 0.0 & 53 & 0.6 & 15.6 \\
 & 2000 & 498 & 500 & 498 & 0.833 & 0.426 & 0.210 & 1.195 & 0.730 & 0.457 & 0.904 & 0.912 & 0.0 & 114 & 0.9 & 16.5 \\
\addlinespace[2pt]
Continuous & 500 & 495 & 496 & 495 & 1.808 & 1.108 & -- & 1.176 & 0.788 & -- & 0.966 & 0.966 & -- & 48 & 0.4 & 11.3 \\
 & 1000 & 500 & 500 & 500 & 0.768 & 0.581 & -- & 0.619 & 0.473 & -- & 0.954 & 0.956 & -- & 109 & 0.5 & 10.3 \\
 & 2000 & 500 & 500 & 500 & 0.474 & 0.367 & -- & 0.396 & 0.315 & -- & 0.964 & 0.950 & -- & 218 & 0.8 & 22.8 \\
 & 3000 & 500 & 500 & 500 & 0.378 & 0.330 & -- & 0.315 & 0.249 & -- & 0.944 & 0.952 & -- & 315 & 1.2 & 34.9 \\
\addlinespace[2pt]
\bottomrule
\end{tabular}}
\vspace{4pt}\par\raggedright\footnotesize \textit{Notes:} See the notes to Table~\ref{tab:mc_class}. In the idiosyncratic-movement, covariate-link, and sparse designs, the class oracle uses the true type at $t = 1$ crossed with $y_{i0}$, as the pairwise oracle does. Failed fits are optimizer failures, all certified as separated, unless stated otherwise. Weak network, $n = 1000$: Proposed 500 (fewer than three pairs) and Class 500 (fewer than three agents in classes with both orderings). Idiosyncratic movement, $n = 1000$: 1 each for Proposed, Class orc., and Class. Covariates in links, $n = 1000$: Proposed 18 and Class 12. Sparse network: Proposed 79 (78 certified as separated), Class 21, and Class orc. 7 at $n = 1000$, and Proposed 2 at $n = 2000$. Continuous, $n = 500$: Proposed 5 and Class 4.
\end{table}

\begin{table}[!ht]
\centering
\caption{Class-based alternatives: exact conditioning, joint class intercepts, and a single class}\label{tab:mc_joint}
\smallskip
\footnotesize
\setlength{\tabcolsep}{3.2pt}
\resizebox{\textwidth}{!}{%
\begin{tabular}{ll cccc cccc cccc cccc}
\toprule
 & & \multicolumn{8}{c}{$\beta$} & \multicolumn{8}{c}{$\alpha$} \\
\cmidrule(lr){3-10} \cmidrule(lr){11-18}
 & & \multicolumn{4}{c}{Bias} & \multicolumn{4}{c}{RMSE} & \multicolumn{4}{c}{Bias} & \multicolumn{4}{c}{RMSE} \\
\cmidrule(lr){3-6} \cmidrule(lr){7-10} \cmidrule(lr){11-14} \cmidrule(lr){15-18}
Design & $n$ & Prop. & Class & Joint & One & Prop. & Class & Joint & One & Prop. & Class & Joint & One & Prop. & Class & Joint & One \\
\midrule
Baseline & 500 & 0.240 & 0.149 & 0.485 & 0.361 & 0.677 & 0.563 & 0.906 & 0.602 & $-$0.076 & $-$0.044 & 0.098 & 0.335 & 1.300 & 1.163 & 1.521 & 1.049 \\
 & 1000 & 0.151 & 0.085 & 0.242 & 0.296 & 0.458 & 0.381 & 0.516 & 0.441 & 0.000 & 0.003 & 0.077 & 0.355 & 0.872 & 0.757 & 0.871 & 0.723 \\
 & 2000 & 0.070 & 0.042 & 0.109 & 0.261 & 0.262 & 0.222 & 0.262 & 0.326 & 0.016 & $-$0.005 & 0.028 & 0.343 & 0.543 & 0.477 & 0.510 & 0.548 \\
 & 3000 & 0.026 & 0.009 & 0.051 & 0.229 & 0.201 & 0.174 & 0.190 & 0.277 & 0.027 & 0.024 & 0.046 & 0.350 & 0.419 & 0.380 & 0.399 & 0.485 \\
\addlinespace[2pt]
Lag-dependent links & 500 & 0.234 & 0.145 & 0.486 & 0.397 & 0.673 & 0.547 & 0.893 & 0.642 & $-$0.055 & $-$0.046 & 0.086 & 0.296 & 1.427 & 1.205 & 1.593 & 1.055 \\
 & 1000 & 0.143 & 0.084 & 0.240 & 0.289 & 0.435 & 0.354 & 0.487 & 0.418 & 0.008 & 0.005 & 0.079 & 0.372 & 0.917 & 0.791 & 0.918 & 0.743 \\
 & 2000 & 0.068 & 0.047 & 0.115 & 0.270 & 0.270 & 0.232 & 0.275 & 0.339 & 0.000 & 0.001 & 0.034 & 0.344 & 0.550 & 0.472 & 0.505 & 0.543 \\
 & 3000 & 0.043 & 0.033 & 0.076 & 0.251 & 0.191 & 0.172 & 0.194 & 0.296 & 0.004 & 0.009 & 0.031 & 0.342 & 0.415 & 0.391 & 0.409 & 0.497 \\
\addlinespace[2pt]
Fixed types & 1000 & 0.132 & 0.097 & 0.260 & 0.084 & 0.397 & 0.351 & 0.492 & 0.326 & $-$0.037 & $-$0.041 & 0.028 & $-$0.039 & 0.719 & 0.645 & 0.742 & 0.601 \\
 & 2000 & 0.046 & 0.032 & 0.098 & 0.030 & 0.223 & 0.207 & 0.242 & 0.201 & 0.023 & 0.023 & 0.058 & 0.018 & 0.454 & 0.433 & 0.465 & 0.413 \\
\addlinespace[2pt]
\bottomrule
\end{tabular}}
\vspace{4pt}\par\raggedright\footnotesize \textit{Notes:} $R = 500$ replications, with statistics over each estimator's own successful fits. Prop.: pairwise estimator. Class: class-conditional estimator~\eqref{eq:class_estimator}. Joint: logit with one intercept per class, estimated together with $(\beta, \alpha)$, on the eligible agents of the feasible classes that contain both orderings (classes with a single ordering are dropped, as in the class-conditional likelihood). One: the class-conditional likelihood with a single class per lag stratum, that is, the within-agent conditional logit with a common time effect conditioned out, which allows one common movement of $\lambda$ for all eligible agents that share $y_{i0}$ and uses no network information. All failed fits are optimizer failures certified as separated: Joint 65 and One 14 in the baseline design at $n = 500$, Joint 74 and One 9 with lag-dependent links at $n = 500$, and Joint 1 and One 1 with lag-dependent links at $n = 1000$.
\end{table}

\begin{table}[!ht]
\centering
\caption{Monte Carlo results for $\beta$: continuous design}\label{tab:mc_cont_beta}
\smallskip
\footnotesize
\setlength{\tabcolsep}{4.5pt}
\begin{tabular}{l cccccc}
\toprule
$n$ & Naive & HK & Proposed & Oracle & Class & Infeas. \\
\midrule
\multicolumn{7}{l}{\textit{Bias}} \\[2pt]
500 & 0.973 & 0.437 & 0.431 & 0.392 & 0.226 & 0.011 \\
1000 & 0.960 & 0.325 & 0.111 & 0.105 & 0.078 & $-$0.003 \\
2000 & 0.957 & 0.299 & 0.032 & 0.032 & 0.049 & $-$0.007 \\
3000 & 0.961 & 0.321 & 0.057 & 0.056 & 0.087 & 0.001 \\
\addlinespace[4pt]
\multicolumn{7}{l}{\textit{Standard deviation}} \\[2pt]
500 & 0.153 & 0.841 & 1.758 & 1.610 & 1.086 & 0.178 \\
1000 & 0.104 & 0.529 & 0.761 & 0.772 & 0.576 & 0.125 \\
2000 & 0.073 & 0.367 & 0.473 & 0.469 & 0.364 & 0.085 \\
3000 & 0.055 & 0.316 & 0.375 & 0.378 & 0.319 & 0.065 \\
\addlinespace[4pt]
\multicolumn{7}{l}{\textit{Root mean squared error}} \\[2pt]
500 & 0.985 & 0.947 & 1.808 & 1.655 & 1.108 & 0.178 \\
1000 & 0.966 & 0.620 & 0.768 & 0.778 & 0.581 & 0.125 \\
2000 & 0.959 & 0.473 & 0.474 & 0.470 & 0.367 & 0.085 \\
3000 & 0.963 & 0.450 & 0.378 & 0.381 & 0.330 & 0.064 \\
\addlinespace[4pt]
\multicolumn{7}{l}{\textit{Coverage of nominal 95\% interval}} \\[2pt]
500 & 0.000 & 0.922 & 0.960 & 0.947 & 0.966 & 0.944 \\
1000 & 0.000 & 0.922 & 0.964 & 0.958 & 0.954 & 0.948 \\
2000 & 0.000 & 0.892 & 0.968 & 0.964 & 0.964 & 0.958 \\
3000 & 0.000 & 0.818 & 0.958 & 0.956 & 0.944 & 0.966 \\
\addlinespace[4pt]
\multicolumn{7}{l}{\textit{Replications without an estimate}} \\[2pt]
500 & 0 & 0 & 5 & 6 & 4 & 0 \\
1000 & 0 & 0 & 0 & 0 & 0 & 0 \\
2000 & 0 & 0 & 0 & 0 & 0 & 0 \\
3000 & 0 & 0 & 0 & 0 & 0 & 0 \\
\bottomrule
\end{tabular}
\vspace{4pt}\par\raggedright\footnotesize \textit{Notes:} $R = 500$ replications. Steep beta-model links $f(u,v) = F(8(u + v - 1))$. The controls estimator is omitted. Bandwidths: $h_1 = c\,(n/500)^{-1/5}$, with $c = 5.13 \times 10^{-4}$ the population 0.10 quantile of the squared codegree distance, computed by numerical integration, and $h_2 = n^{-1/5}/5$. Oracle: the same pairwise estimator with the population codegree distances $\delta_1(w_{i1}, w_{j1})^2$, computed by quadrature, so free of link noise and of the sampling variation of the other agents' latent indices. Class: the class-conditional estimator, with the connected components of the thresholded graph as classes (not classes of equal movement here, Section~\ref{sec:class_estimator}) and the covariate-stability window as an indicator. With no distribution theory here, coverage is only a diagnostic. Failed fits at $n = 500$ are optimizer failures, all certified as separated except one of the oracle's: Proposed 5, Oracle 6, and Class 4.
\end{table}
\begin{table}[!ht]
\centering
\caption{Monte Carlo results for $\alpha$: continuous design}\label{tab:mc_cont_alpha}
\smallskip
\footnotesize
\setlength{\tabcolsep}{4.5pt}
\begin{tabular}{l cccccc}
\toprule
$n$ & Naive & HK & Proposed & Oracle & Class & Infeas. \\
\midrule
\multicolumn{7}{l}{\textit{Bias}} \\[2pt]
500 & 0.667 & 0.157 & 0.145 & 0.148 & 0.067 & $-$0.008 \\
1000 & 0.670 & 0.142 & 0.052 & 0.052 & 0.025 & 0.002 \\
2000 & 0.663 & 0.105 & 0.013 & 0.013 & 0.003 & $-$0.006 \\
3000 & 0.665 & 0.094 & 0.007 & 0.008 & 0.013 & $-$0.003 \\
\addlinespace[4pt]
\multicolumn{7}{l}{\textit{Standard deviation}} \\[2pt]
500 & 0.119 & 0.472 & 1.168 & 1.123 & 0.793 & 0.130 \\
1000 & 0.090 & 0.333 & 0.617 & 0.607 & 0.473 & 0.094 \\
2000 & 0.063 & 0.242 & 0.396 & 0.401 & 0.315 & 0.066 \\
3000 & 0.053 & 0.198 & 0.315 & 0.313 & 0.249 & 0.054 \\
\addlinespace[4pt]
\multicolumn{7}{l}{\textit{Root mean squared error}} \\[2pt]
500 & 0.678 & 0.497 & 1.176 & 1.131 & 0.795 & 0.131 \\
1000 & 0.676 & 0.362 & 0.619 & 0.609 & 0.473 & 0.094 \\
2000 & 0.666 & 0.264 & 0.396 & 0.401 & 0.315 & 0.066 \\
3000 & 0.667 & 0.219 & 0.315 & 0.313 & 0.249 & 0.054 \\
\addlinespace[4pt]
\multicolumn{7}{l}{\textit{Coverage of nominal 95\% interval}} \\[2pt]
500 & 0.000 & 0.960 & 0.964 & 0.964 & 0.966 & 0.974 \\
1000 & 0.000 & 0.940 & 0.978 & 0.966 & 0.956 & 0.960 \\
2000 & 0.000 & 0.938 & 0.954 & 0.944 & 0.950 & 0.954 \\
3000 & 0.000 & 0.942 & 0.952 & 0.956 & 0.952 & 0.968 \\
\addlinespace[4pt]
\multicolumn{7}{l}{\textit{Replications without an estimate}} \\[2pt]
500 & 0 & 0 & 5 & 6 & 4 & 0 \\
1000 & 0 & 0 & 0 & 0 & 0 & 0 \\
2000 & 0 & 0 & 0 & 0 & 0 & 0 \\
3000 & 0 & 0 & 0 & 0 & 0 & 0 \\
\bottomrule
\end{tabular}
\vspace{4pt}\par\raggedright\footnotesize \textit{Notes:} $R = 500$ replications. Steep beta-model links $f(u,v) = F(8(u + v - 1))$. The controls estimator is omitted. Bandwidths: $h_1 = c\,(n/500)^{-1/5}$, with $c = 5.13 \times 10^{-4}$ the population 0.10 quantile of the squared codegree distance, computed by numerical integration, and $h_2 = n^{-1/5}/5$. Oracle: the same pairwise estimator with the population codegree distances $\delta_1(w_{i1}, w_{j1})^2$, computed by quadrature, so free of link noise and of the sampling variation of the other agents' latent indices. Class: the class-conditional estimator, with the connected components of the thresholded graph as classes (not classes of equal movement here, Section~\ref{sec:class_estimator}) and the covariate-stability window as an indicator. With no distribution theory here, coverage is only a diagnostic. Failed fits at $n = 500$ are optimizer failures, all certified as separated except one of the oracle's: Proposed 5, Oracle 6, and Class 4.
\end{table}

\clearpage
\bibliographystyle{apalike}

\end{document}